\documentclass[11pt]{article}
\usepackage[margin=1in]{geometry}

\usepackage[sort&compress]{natbib}

\usepackage{sectsty}
\usepackage{graphicx}
\usepackage{amsmath}
\usepackage{amssymb}
\usepackage{amsthm}
\usepackage{float}

\newtheorem{lemma}{Lemma}[section]

\newtheorem{innercustomgeneric}{\customgenericname}
\providecommand{\customgenericname}{}
\newcommand{\newcustomtheorem}[2]{%
  \newenvironment{#1}[1]
  {%
   \renewcommand\customgenericname{#2}%
   \renewcommand\theinnercustomgeneric{##1}%
   \innercustomgeneric
  }
  {\endinnercustomgeneric}
}

\newcustomtheorem{customthm}{Theorem}
\newcustomtheorem{customassumption}{Assumption}

\usepackage[hidelinks]{hyperref}

\usepackage{tabularx}
\usepackage{xcolor}
\usepackage{float,subcaption,placeins}
\usepackage{amsmath,amsfonts}
\usepackage{thm-restate}
\usepackage{float}
\usepackage{enumitem}
\usepackage[capitalise,noabbrev]{cleveref}

\crefname{section}{Section}{Sections}
\Crefname{section}{Section}{Sections}
\crefname{subsection}{Section}{Sections}
\Crefname{subsection}{Section}{Sections}
\crefname{subsubsection}{Section}{Sections}
\Crefname{subsubsection}{Section}{Sections}

\crefname{appsec}{Appendix}{Appendices}
\Crefname{appsec}{Appendix}{Appendices}

\usepackage{xr}
\usepackage{graphicx}
\usepackage{thmtools}
\usepackage{url}
\usepackage{nicefrac}       %

\newcommand{\parbold}[1]{\paragraph{#1}}

\let\cite\citep

\newcommand{\Comments}{1}
\newcommand{\mynote}[3]{\ifnum\Comments=1\textcolor{#1}{#2:#3}\fi}

\usepackage{booktabs} %
\usepackage{xfrac}
\usepackage{algorithm}

\usepackage{color}
\usepackage{verbatim}
\usepackage{enumitem}
\usepackage{bm}
\usepackage{tikz}
\usepackage{subcaption}
\usepackage{float}
\usepackage{listings}
\usepackage{array}
\usepackage{makecell}
\usepackage{mathtools}

\definecolor{DarkBlue}{rgb}{0.1,0.1,0.5}

\usepackage[group-separator={,}]{siunitx}
\newcommand\blfootnote[1]{%
	\begingroup
	\renewcommand\thefootnote{}\footnote{#1}%
	\addtocounter{footnote}{-1}%
	\endgroup
}

\newcommand\bbN{\ensuremath{\mathbb{N}}}
\newcommand\mV{\ensuremath{\mathcal{V}}}
\newcommand\mB{\ensuremath{\mathcal{B}}}

\DeclarePairedDelimiter\ceil{\lceil}{\rceil}
\DeclarePairedDelimiter\floor{\lfloor}{\rfloor}

\DeclareUrlCommand\email{\urlstyle{rm}}

\usepackage[utf8]{inputenc} %
\usepackage[T1]{fontenc}    %
\usepackage{url}            %
\usepackage{amsfonts}       %
\usepackage{microtype}      %
\usepackage{xcolor}

\usepackage{authblk}
\usepackage[T1]{fontenc}
\usepackage{lmodern}

\DeclareMathOperator*{\argmax}{arg\,max}

 \bibpunct[, ]{(}{)}{,}{a}{}{,}%
\title{Combatting Gerrymandering with Ranked Choice Voting: An Experimental Analysis of Multi-member Districts in the United States}

\author[1,2]{Nikhil Garg}
\author[3]{Wes Gurnee}
\author[4]{David Rothschild}
\author[1]{David Shmoys}
\affil[1]{Operations Research and Information Engineering, Cornell}
\affil[2]{Cornell Tech}
\affil[3]{Operations Research Center, MIT}
\affil[4]{Microsoft Research}

\date{\today}

\begin{document}
\maketitle

\begin{abstract}
Every representative democracy must specify a mechanism under which voters choose their representatives. The most common mechanism in the United States -- winner-take-all single-member districts -- both enables substantial partisan gerrymandering and constrains `fair' redistricting, preventing proportional representation in legislatures. We study the design of \textit{multi-member districts (MMDs)}, in which each district elects multiple representatives, potentially through a non-winner-take-all voting rule. We carry out large-scale empirical analyses for the U.S. House of Representatives under MMDs with different social choice functions and algorithmically generated maps optimized for either partisan benefit or proportionality. Doing so requires efficiently incorporating predicted partisan outcomes -- under various multi-winner social choice functions -- into an algorithm that optimizes over an ensemble of maps. We find that with three-member districts using Single Transferable Vote, fairness-minded independent commissions would be able to achieve proportional outcomes in every state up to rounding, \textit{and} advantage-seeking partisans would have their power to gerrymander significantly curtailed. Simultaneously, such districts would preserve geographic cohesion. Through simulation, we find that the insights are robust to cross-party voting. In the process, we advance a rich research agenda at the intersection of social choice and computational gerrymandering.

\blfootnote{\email{ngarg@cornell.edu}, \email{wesg@mit.edu}, \email{david@researchdmr.com}, \email{david.shmoys@cornell.edu}. The authors thank FairVote and Moon Duchin (and others in the Data and Democracy lab) for helpful discussions. NG was supported by NSF CAREER IIS-2339427, and Cornell Tech Urban Tech Hub, Google, Meta, and Amazon research awards.}
\end{abstract}

\section{Introduction}
\begin{quote}
    This bill requires (1) that ranked choice voting \dots be used for all elections for Members of the House of Representatives, (2) that states entitled to six or more Representatives establish districts such that three to five Representatives are elected from each district, and (3) that states entitled to fewer than six Representatives elect all Representatives on an at-large basis---\cite{fra}. 
\end{quote}

\noindent The Fair Representation Act (the FRA), first introduced in 2017 and reintroduced since,\footnote{\url{https://beyer.house.gov/news/documentsingle.aspx?DocumentID=5276}} would mandate the use of multi-member districts (MMDs) to elect members to the United States House of Representatives, i.e., having fewer, larger districts each with multiple representatives. The American Academy of Arts and Sciences in 2020 released a report advocating for states to use multi-member districts  -- however, ``on the condition that they adopt a non-winner-take-all election model'' \citep{reportamericansciences_2020}. Despite the popular focus on single-member district (SMD) elections, such MMDs have a long history in the United States, especially at the state and local levels. In 1962, 41 state legislatures had MMDs, often with winner-take-all models \citep{niemiCandidaciesCompetitivenessMultimember1991}; as of 2021, 10 state legislatures elect representatives for at least one chamber in such a manner -- Arizona, for example, has two-member districts, where each voter is given two votes and the top two vote-getters in each district are elected \citep{ballotpedia_2021,ballotpedia_2021stateleg}, and there is a debate in Maryland about whether to keep its (winner-take-all) MMDs. City councils, state parties, and other organizations often adopt more sophisticated techniques, using variations on Ranked Choice Voting (RCV) to elect multiple winners from each of several districts \citep{fairvoteorgrcv_2021}.

This paper considers the design of such multi-member districts and, in the process, advances a rich research agenda at the intersection of two well-studied aspects of the design of representative democracies: partisan gerrymandering, and social choice for multiple winners. The gerrymandering literature largely (but not exclusively) assumes winner-take-all voting in single-member districts (SMDs) and, given a fixed set of voters, studies how to divide the voters into districts such that the set of winners across districts satisfies desirable properties (e.g., for a gerrymandering party, maximizing the number of winners belonging to their party; for a neutral rule-maker, devising rules to ensure `proportionality,' such that the fraction of winners of each party matches the fraction of voters). The theoretical social choice literature, on the other hand, primarily considers a single district and studies voting rules (functions from voter rankings to a set of winners) such that the set of winners in that district satisfies similar properties, including proportionality. A map consisting of multiple MMDs (e.g., 30 two-member districts in Arizona) requires \textit{both} partitioning voters into districts \textit{and} devising rules for how each district collects and aggregates votes. The challenge cannot be decomposed: effectively drawing districts depends on the social choice function, and the same social choice function has different effects depending on district size and composition.

We study this joint challenge, carrying out large-scale empirical analyses for the U.S. House of Representatives under MMDs with different social choice functions, under maps that could be drawn by either partisan gerrymanderers or independent commissions.
We show that indeed the choices of (how many and which) districts to draw and which voting method to use should not be made independently. For example, while MMDs with winner-take-all elections are commonly believed to be discriminatory against minorities when compared to SMDs \citep{niemiImpactMultimemberDistricts1985,derfnerMultimemberDistrictsBlack72,bullock1993changing}, we empirically find that MMDs with more proportional voting methods support political minorities. We further find that `interior' solutions in this joint optimization -- with neither SMDs nor one large MMD -- often best balance the multiple objectives promoted by good governance groups, such as proportionality, competitiveness, and geographic district size. In more detail, our contributions and findings are as follows.

\parbold{Methodological contributions. } We provide a scalable, empirical methodology to algorithmically study partisan gerrymandering and fair redistricting under multi-member districts. We leverage an efficient way to calculate partisan outcomes under Single Transferable Vote (a common variant of RCV for multiple winners) and extend an  algorithm by \citet{gurnee2021fairmandering} to efficiently calculate near-optimal multi-member maps, under a variety of objective functions and social choice rules. We further develop an individual-level, voter-file-based methodology to study party crossover and intra-party effects of ranked choice voting.
Our results, using voting and individual-level data from across the United States, further illustrate how theoretical social choice guarantees heterogeneously translate to practice, away from the asymptotic regimes in which they are often studied.
We believe that this work and methodological approach can be applied beyond the context of multi-member districts for legislatures in the United States, for questions regarding the design of representative democracies at the intersection of gerrymandering and social choice, across computer science, optimization, and political science.

\parbold{Application-oriented findings. } We apply our methods to consider the design of MMDs empirically in the House of Representatives in the United States. We primarily study \textit{proportionality},\footnote{In this work, we use ``fair'' and ``proportional'' interchangeably. There are many other redistricting desiderata in the literature, which we briefly discuss in the related work and also show results for in \Cref{sec:proportionality}. Similarly, we refer to independent commissions as one group who may desire proportionality, though in practice such commissions have other possible objectives.} as a function of the number of districts and the voting rule used. Surprisingly, our results indicate that two-member or three-member districts with proportional voting rules are often sufficient to prevent the worst gerrymanders by an adversary and to enable independent commissions to find a proportional map -- even small MMDs are enough to eliminate ``natural'' gerrymandering, in which geography and the distribution of citizens lead to a natural advantage for one party, and purposeful ones.  This result holds with no other optimization constraints besides district contiguity and continues to hold when voters behave noisily and potentially vote across parties. On the other hand, we find that poor design constraints (in \textit{either} the number of districts \textit{or} the voting rule used) enable more extreme gerrymanders than is possible with SMDs.

These effects occur because MMDs with appropriate social choice functions protect concentrated (political) minorities by making packing and cracking more difficult, and support diffuse minorities by joining them in the same large multi-member district, with a lower threshold to win a seat. To underscore the practical relevance of our methods, we analyze the specific proposal in the Fair Representation Act and show that it achieves an effective balance in the MMD design space---though, as emphasized above, its prescribed three- to five-member districts are larger than are necessary; this part of the proposal could be relaxed without sacrificing its proportionality goals.

Finally, in the appendix, we study how MMDs affect \textit{intra-party} diversity, measuring two competing claims: that small SMDs protect diverse, geographically-correlated interests within a party (e.g., urban, Black Democrats versus suburban, white Democrats), and conversely that RCV with large MMDs enables ideological differences within a party (cf. \cite{fairvoteorgrcv_2021,page_gilens_2018}). We find, under reasonable assumptions on how voters rank candidates within a party, that two- or three-member districts preserve geographic compactness, and thus the notion that legislators \textit{represent} a geographically cohesive set of voters, while increasing the partisan diversity of the winning set---and thus striking a balance between the competing claims.

\Cref{sec:modelmethod} presents the model and methods. \Cref{sec:proportionality} contains our main results, regarding inter-party effects of multi-member districts. \Cref{sec:newwithoutsolidcoalitions} extends the results to a setting in which voters behave noisily and may vote across parties. We conclude in \Cref{sec:discussion}. Appendix \ref{sec:intra} considers \textit{intra-party} effects. %
Our code is available at \url{https://github.com/nikhgarg/gerrymandering_mmd}.

\subsection*{Related work}
This work is at the intersection of two rich literatures, showing how joint optimization over how voters are split into districts (gerrymandering) and how they vote within a district (social choice) yields outcomes that neither could achieve separately. It further contributes to the literature studying multi-winner districts in the United States. Given each field's vast literature, we focus on the most relevant work.%

\parbold{Redistricting.} Gerrymandering is the practice of using district boundaries to engineer electoral outcomes by ``packing'' and ``cracking'' voters within and between districts \citep{issacharoff2002gerrymandering, erikson1972malapportionment, mcgann2016gerrymandering}. Due to the status quo rules in the United States, past computational research has primarily studied gerrymandering and redistricting more broadly in the context of a fixed number of single-member districts with winner-take-all elections.

Optimal political districting is well known to be an \textit{NP}-hard computational problem \citep{puppe2008computational, kueng2019fair, chatterjee2019partisan,van2015network,dyer1985complexity,lewenberg2017divide}. Therefore, to study redistricting, researchers have developed ensemble methods \citep{Deford2021recombination, liu2016pear, chen2013unintentional, autry2020multi, fifield2020automated, mccartan2020sequential, gurnee2021fairmandering, autry2021metropolized} that generate huge quantities of legal district plans to explore the exponentially large space of feasible maps.

Such techniques are commonly used to detect gerrymandering \citep{chikina2017assessing, duchin2018outlier, duchinHomologicalPersistenceGerrymandering2020, herschlag2017evaluating, herschlag2020quantifying} or to study the impact of different redistricting rules on the distribution of outcomes \citep{deford2019redistricting, deford2020computational}. An important result is that even under `neutral' maps drawn without regard to underlying political geography, the natural geographic segregation among partisan voters can create skewed political outcomes \citep{borodin2018big, chen2013unintentional}. More recently, \citet{DuchinSchoenbach2022371393} find that plurality single-member districts are compatible with proportionality in many states, if used as a target for redistricting.

There are many proposed metrics of fairness in redistricting, including: the efficiency gap \citep{stephanopoulos2015partisan}, the mean-median gap \citep{wang2016three}, partisan-symmetry \citep{warrington2018quantifying}, competitiveness \citep{deford2020computational}, and most simply, proportionality. Such metrics can be used as the objective function of an optimization algorithm to generate district maps with (un)fair outcomes \citep{swamyMultiobjectiveOptimizationPolitical, gurnee2021fairmandering, king2015efficient, cannon2020voting}. However, because of partisan segregation, in some states creating a proportional plan with SMDs is actually impossible \citep{duchin2019locating}, motivating the use of alternative election procedures. In this work, we primarily consider proportionality as our fairness notion.

Methods to restrict gerrymandering are also commonly proposed. \citet{borodin2018big} find that compactness constraints do not effectively prevent gerrymandering. Inspired by the fair division literature, researchers have also recently proposed methods in which two adversarial parties together draw a map \citep{benade2021you,benade2020abating,pegden2017partisan,tucker2018cut}.

\parbold{Social choice.} The social choice literature, by contrast, typically fixes a set of voters and a number of winners and studies how voting rules (how preferences are elicited and aggregated) affect outcomes~\citep{brandt2016handbook,arrow2010handbook,dewanPoliticalEconomyModels2011}. Here, we focus on the literature on multi-winner elections. %
The theoretical social choice literature considers properties of voting rules, often for arbitrary distributions of voters~\citep{caragiannisSubsetSelectionImplicit2017,lacknerQuantitativeAnalysisMultiwinner2018,azizJustifiedRepresentationApprovalbased2017,elkindPropertiesMultiwinnerVoting2015,conitzerParadoxesMultipleElections,azizComputationalAspectsMultiwinner2015,garg2019your,chengGroupFairnessCommittee2019,fishburn2018approval}. Most relevant is \citet{skowronProportionalityDegreeMultiwinner2019}, which studies the proportionality guarantees of various multi-winner voting rules, including Thiele rules; they find Proportional Approval Voting (a type of Thiele rule) to be the most proportional with respect to their measure, though other Thiele methods may outperform it on objectives beyond proportionality. The optimality of PAV and the parametrizable nature of the class motivates us to focus on Thiele rules in this work, and our results complement the theoretical work in comparing such rules in non-asymptotic settings with empirical voter distributions. Related to our intra-party analysis, \citet{elkind2017multiwinner} simulate the properties of multiwinner voting rules (such as STV and PAV) when voters lie on a two-dimensional Euclidean space. 

On the applied side, a strand of the literature studies the effects of voting reforms, especially on minority voters.~\citet{mcdanielDoesMoreChoice2018} finds that RCV increased racially polarized voting in comparison to runoff elections.~\citet{mcgheeHasTopTwo2017} find mixed evidence for whether top-two primaries have reduced polarization, and~\citet{rogowskiPrimarySystemsCandidate2015} find no evidence that open versus closed primaries affect candidate ideology;~\citet{groseReducingLegislativePolarization2020} finds the opposite on both points.~\citet{spencerEscapingThicketRanked2015} advocate for ranked choice voting as a Voting Rights Act remedy, supporting its legality.

\parbold{Multi-member districts in the United States.} As chronicled by~\citet{klain1955new} in 1955, there is a rich history of multi-member districts in the United States in state legislatures and city councils. In 1962, for example, 30 states used MMDs to elect state senators, and 41 states used them to elect state representatives, mostly in two- to four-member districts with winner-take-all procedures \citep{niemiCandidaciesCompetitivenessMultimember1991}. Historically, at-large city council members have often also been elected with equivalent procedures. On the other hand, since 1967 federal law has banned the use of MMDs to elect members to the U.S. House of Representatives, a status quo the Fair Representation Act would repeal.

The political science community has studied these districts, often focusing on their effect on race \citep{silva1964compared,dauerMultimemberDistrictsDade1966,banzhafMultimemberElectoralDistricts1966}, though some have questioned the optimal district size  \citep{hamiltonLegislativeConstituenciesSingemember1967}. In particular, while the empirical evidence is mixed due to the observational nature of analysis, the predominant view is that winner-take-all MMDs harmed racial minorities, especially in Southern states~\citep{niemiImpactMultimemberDistricts1985,derfnerMultimemberDistrictsBlack72,bullock1993changing}. Similarly,  winner-take-all at-large city council elections are thought to dilute the votes of minorities~\citep{stillVoluntaryConstituenciesModified1991,bullock1989symbolics,walawenderAtlargeElectionsVote1985,trounstineContextMattersEffects2008,davidsonAtlargeElectionsMinoritygroup1981}. As theoretical models \citep{gerberMinorityRepresentationMultimember1998} elucidate, such districts with winner-take-all rules enable the majority -- even without substantially distorting district lines as might be necessary with SMDs -- to ensure that a minority even with 49\% of the vote elects none of its preferred candidates. Partially due to such evidence, by 1982 the majority of states had eliminated the use of MMDs for their state legislatures, and several cities have eliminated their at-large council seats under court mandate. \citet{lempert2021ranked} analyzes the legality of RCV and multi-member districts, and advocates for it as a court-ordered remedy to illegal gerrymanders.

Most related is work favorably comparing MMDs using ranked choice voting to SMDs for city councils. \citet{benadeRankedChoiceVoting2021} compare outcomes using RCV versus single-member districts in four empirical case studies, finding that it yields better representation for demographic minorities; key to their approach is constructing various ranking-based models for how voters prefer candidates, beyond a ``solid coalitions'' assumption in which voters prefer same-party candidates before other-party ones. From the same group,~\citet{lowell2019} and~\citet{chicago2019} propose using RCV with several MMDs instead of SMDs for the Lowell, MA and Chicago, IL city councils, respectively. Their work primarily considers non-partisan representation, such as race and other demographic characteristics -- these factors are more salient in municipal contexts, where such MMDs with RCV are more common today. \Cref{sec:newwithoutsolidcoalitions} and \Cref{sec:intra} explore similar questions with cross-party voting and intra-party effects, respectively, with complementary voter-file-based methods. More recently in subsequent and complementary work, the \citet{FRAanalysisDuchin} also analyzes the effects of the Fair Representation Act; they consider a more limited set of district sizes and focus on racial representation, as opposed to our focus on partisan effects.

To this literature, our work adds a systematic analysis of the partisan gerrymandering effects of such districts: providing a scalable methodology,  characterizing the effects on proportional representation under both adversarial gerrymandering and `fair' redistricting across the range of possible district sizes, and providing design recommendations for legislation. We further study how MMDs affect intra-party results on two dimensions, political preference and geography.

\parbold{International contexts and other political solutions.} Countries beyond the United States also use MMDs -- in Singapore, for example, the ruling party has been accused of using MMDs with winner-take-all rules to gerrymander (with boundaries determined immediately before the election) \citep{singaporegerry}. Empirically comparing elections across 81 countries, \citet{carey2011electoral} find that larger district sizes are correlated with higher proportionality; they find that multi-member districts of size four to eight exhibit much of the proportionality benefits and argue for them as the ``electoral sweet spot;'' our work suggests that in the United States, even smaller districts would suffice. \citet{horwill1925proportional} and \citet{taagepera1989seats} also highlight district size as an important factor in proportionality. \citet{lijphart2000patterns} compares majority and proportional representation systems internationally and highlights that two-member districts (with plurality voting) have been used in the United Kingdom, Canada, India, and Barbados, but not since 1970; Mauritius has continued to use three-member districts, using plurality voting.

Multi-member districts are far from the only possible governmental system to achieve proportionality or other related desiderata. Many countries have seat allocation systems designed to be proportional with multiple parties (see~\citet{pukelsheim2017proportional} for a survey). For example, party list systems (as in Israel and the Netherlands \citep{lijphart2000patterns}) allow voters to specify a party (potentially alongside specific candidates); each party above a threshold then receives an approximately proportional number of seats, with methods such as Jefferson/D'Hondt used to convert proportions to integral seats. Such systems and related ones have also been studied in the EconCS and Operations communities (e.g., \citet{cembrano2021multidimensional,cembrano2021proportional}). \citet{guney2018efficient} and \citet{bredereck2021strategic} analyze the margin of victory in proportional representation systems, which is related to our competitiveness analysis for multi-member districts.

Finally, we note that other reforms are also of interest in the United States, with other proposals to achieve similar goals (e.g., \citet{balinski2010fair,ernst1994appointment}). For example, \citet{bachrach2016misrepresentation} consider the {misrepresentation ratio} of single-member districting with various score-based ranking rules, in \textit{Electoral College}-type settings: when these rules are used to elect a winner within each district but then an overall winner is declared based on the majority vote of districts.

A comparative empirical analysis of such approaches---and especially which may be most politically feasible in various jurisdictions---is outside the scope of this work.

\section{Model and methods}
\label{sec:modelmethod}

\Cref{sec:probdef} formalizes the joint redistricting and social choice challenge, \Cref{sec:socialchoiceprops} introduces the social choice functions we analyze and shows how to efficiently compute per-party seat shares under them, and \Cref{sec:empiricalmethod} presents our empirical method.

\subsection{Problem definition}
\label{sec:probdef}

The task is to elect a legislature composed of $N$ seats. There is a population of voters $\mV$; each voter $v\in \mV$ lives in an \textit{atomic} block $b_v \in \mB$ and belongs to a party $p_v \in \{R, D\}$. The blocks $\mB$ are organized in an adjacency graph where two blocks share an edge when they are geographically adjacent to each other.

 \parbold{Multi-member redistricting.} In redistricting, the blocks $\mB$ are partitioned into $K$ geographically contiguous districts, where each district $k$ is allocated $N_k \in \bbN$ seats in the legislature such that $\sum_k N_k = N$, and $N_k\geq 1$. Each district also must be population balanced in accordance with $N_k$, such that the population ratio of district $k$ to the whole state is bounded between $\frac{N_k \pm \epsilon}{N}$ for population tolerance $\epsilon$.\footnote{    Our chosen tolerance provides a 1\% upper bound in population deviation. {More precisely, in the optimization, let $\tau=0.01$ and the single-seat ideal population be $P_{\text{ideal}}=P_{\text{total}}/N$. For a district with $N_k$ seats at tree depth $L$, we enforce $P_k \in [N_k P_{\text{ideal}} \pm (\tau P_{\text{ideal}}/L)]$, i.e., a relative deviation $\delta_k=\tau/(L N_k)$ from $P_{\text{ideal},k} = N_k P_{\text{ideal}}$, bounded above by 1\%. This modeling tolerance should not be interpreted as a legal threshold: for U.S.\ congressional redistricting, the constitutional standard is population equality ``as
 		mathematically equal as reasonably possible;'' the U.S. Supreme Court has struck down a plan with less than a 1\% deviation between the smallest and largest district \citep{aclu_redistricting_2001}.}
 } The resulting set of districts is referred to as a map or a plan.  Maps can be drawn to fulfill various objective functions. For example, a party can gerrymander a map to maximize the seats won by their party, or an independent commission can attempt to satisfy one or more notions of fairness.

 \parbold{Social choice function.} Each district runs a separate election to fill its $N_k$ seats, with each voter $v$ voting in the district $k(b_v)$ to which their block $b_v$ belongs. Each election is run according to a social choice function $F$ that determines what information each voter provides (in this work, either a set of approved candidates or a ranking over candidates) and how votes are aggregated to produce the $N_k$ winners in each district $k$.

 \medskip
 Together, the choice of map ($K$ contiguous districts with $N_k$ seats each) and social choice function $F$ define an election. Given data on voters and assumptions on how they rank candidates (detailed in \Cref{sec:socialchoiceprops}), the procedure determines the set of $N$ winners. (In our empirical optimization approach, we use a probabilistic approach to compute the \textit{expected} seat share based on the variance of past elections to model the heterogeneous political elasticity and noise among different constituencies.)

 \parbold{Evaluation metrics.} An election procedure can be evaluated based on the set of winners it produces. To measure outcomes across parties, we primarily consider \textit{proportionality}: how does the fraction of winners $w_p$ belonging to each party $p \in \{R,D\}$ (seat share) compare to the fraction of voters $y_p$ (vote share); the \textit{proportionality gap} is the difference $| w_R - y_R | $. The larger the gap, the more that the procedure favors one party over another. In \Cref{sec:intra}, we further consider \textit{intra-party} measures, i.e., how the election procedure influences within-party differences.

\parbold{Research questions.} We ask: \textit{how does the election procedure determine the induced outcomes, and what is the joint influence of the social choice function $F$ and the map?} Note that the various components of the election procedure may be determined by different actors with competing interests. The Fair Representation Act, for example, would mandate that $N_k \in \{3, 4, 5\}$ and a RCV-based social choice function be used. However, barring a separate mandate concerning independent commissions, in each state a partisan state legislature may still seek to draw maps most favorable to their party.\footnote{Of course, to the extent that these laws are enforced in court, maps have to be in accordance with the Voting Rights Act and any state-specific laws concerning partisan balance.} Thus, we in particular are interested in the following design question: \textit{how do various constraints on the election procedure design space (restrictions on $K$, $N_k$, or $F$) affect the range of outcomes possible}, under maps drawn by either partisan actors or independent commissions?

\subsection{Social choice functions, voter assumptions, and calculating seat shares}
\label{sec:socialchoiceprops}

We study our research questions empirically, analyzing outcomes under counterfactual election procedure designs. However, doing so requires converting historical voting data to expected electoral outcomes under various social choice functions and with hypothetical candidates. Further, outcomes must be efficiently computable, as our redistricting optimization algorithm (introduced in \Cref{sec:empiricalmethod}) calculates as a subroutine the seat share under a given district and social choice function.
 We now lay the groundwork for our empirical method, by introducing our social choice functions and assumptions and showing how to efficiently compute per-party seat shares in a district.

\parbold{Social choice functions considered.} We consider two well-studied classes of multi-winner social choice functions: Thiele rules and Single Transferable Vote (STV). Variants of STV are used in elections in Ireland, Australia, Scotland, Malta, and locally in the United States \citep{electoralreformsociety_stv,bowler2000elections,fairvoteorgrcv_2021,tideman1995single}. PAV, while not as widely used in practice, is well studied in the literature, due to its strong theoretical guarantees \citep{azizJustifiedRepresentationApprovalbased2017,lacknerQuantitativeAnalysisMultiwinner2018,skowronProportionalityDegreeMultiwinner2019}.

 A Thiele rule is characterized by a function $\lambda$. Each voter $v$ in district $k$ provides a set $S_v$ of the $N_k$ candidates they like most. The set of winners is determined as follows. Consider a potential set of winners $C$. The amount of points that voter $v$ contributes to $C$ is $\sum_{i = 1}^{| S_v \cap C|} \lambda(i)$, and the set $C$ with the most points across voters (after tie-breaking) is selected. The (non-increasing) function $\lambda$ establishes that votes have diminishing returns; the $i$th approved candidate in the set is worth $\lambda(i)$. For example, winner-take-all approval voting is a Thiele rule with $\lambda(i) = 1$: for each voter, the set $C$ receives a number of points equal to the number of candidates in the set that the voter likes, and so the $N_k$ candidates with the most votes win. We further consider \textit{Proportional Approval Voting} (PAV), a Thiele rule  with $\lambda_{\text{PAV}}(i) = \frac{1}{i}$; and \textit{Thiele Squared}, with $\lambda_{\text{TS}} = \frac{1}{i^2}$. PAV is well-studied in the theoretical social choice literature and comes with optimality guarantees (in terms of proportionality) \citep{skowronProportionalityDegreeMultiwinner2019}. Thiele Squared induces even sharper diminishing returns than PAV in a voter having more of their preferred candidates winning, and thus prioritizes more voters having at least one approved candidate (and 50/50 outcomes between two parties, regardless of underlying vote shares). More generally, Thiele rules both contain common voting rules and parameterize a rule's tendency to favor proportional or balanced outcomes.

  In Single Transferable Vote (STV), a type of ranked choice voting for multiple winners, each voter instead submits a ranking over candidates (here, we assume the rankings are complete, not partial). Consider a district with $N_k$ seats and $V_k$ voters; the set of winners is created iteratively, as follows. The number of first-place votes for each candidate is counted. Any candidate with a number of votes at least the ``Droop'' threshold $Q = \floor{\frac{V_k}{N_k + 1}} + 1$ is selected as a winner, and their surplus votes (number of votes minus $Q$) are transferred to each voter's next preference. If no candidate has enough votes, then the candidate with the least number of votes (with tie-breaking) is eliminated, and their votes are transferred. The process continues until $N_k$ candidates have been selected, or the number of remaining candidates equals the number of remaining seats. STV is commonly used in practice in multi-winner elections.

  Details can differ in how such votes are transferred, but our results in \Cref{sec:proportionality} do not depend on the transfer rule; in our simulations in \Cref{sec:newwithoutsolidcoalitions,sec:intra}, we use the ``fractional'' STV rule, in which each voter keeps a fractional number of votes equal to their share of the winning candidate's surplus votes (we describe this formally in \Cref{sec:newwithoutsolidcoalitions}). For consistency, for all rules we assume that ties are broken in favor of candidates from party $D$, but randomly within each party.\footnote{Tie breaking is assumed for theoretical precision. It does not play a role in our empirical analysis, given the data. }

\parbold{Voter assumptions.} Next, we require assumptions on how voters rank or approve candidates.
For each Thiele rule, we assume that in each district $k$ there are \textit{exactly} $N_k$ candidates from each party, and each voter simply approves all the candidates in their party.\footnote{For example, with SMDs, each party typically advances exactly 1 candidate, chosen via a primary.}
For STV, we assume that in each district $k$ there are \textit{at least} $N_k$ candidates from each party.\footnote{For approval voting rules (such as the Thiele rules), having more candidates from each party than the number that a voter is allowed to vote for would induce stronger strategic concerns and might not lead to proportionality, even with solid coalitions. For example, suppose there are 2 candidates from each party in a district with 1 seat; then, since each voter can vote for only 1 candidate, the party with a majority of the votes might not win the seat if its voters split their votes between the two candidates; for this reason, in standard elections with 1 winner and 1 vote per voter, there is a primary for each party to decide their candidate. For ranked voting rules such as STV in which the voter can rank all candidates, having more candidates than seats does not create such concerns, since voters can rank all candidates in their preferred order.} For our primarily empirical analysis, we further assume that parties are \textit{solid coalitions} \citep{dummett1984voting}: that each voter ranks all candidates of their party over each candidate of the other party, though intra-party rankings may vary.
These assumptions reflect that party membership is a strong predictor of the party of the candidate for which a voter votes. Such an assumption, using historical party vote shares to predict vote shares in future elections, is used throughout the redistricting literature for SMDs; our additional assumption is that this behavior extends to partisan voting in multi-member elections. We relax this assumption in \Cref{sec:newwithoutsolidcoalitions} and study intra-party effects, where voters may disagree on within-party rankings, in \Cref{sec:intra}. %

\parbold{Calculating seat share given vote shares.} Suppose -- for a given district with $N_k$ seats and $V_k$ voters -- we know the fraction $y_p$ of voters that belong to each party $p \in \{R, D\}$ (vote share). How do we calculate the number of winners $n_p$ from each party for a given social choice function?

Given our assumptions (including \textit{solid coalitions}), it is straightforward to efficiently do so for Thiele rules, as candidates from a given party receive the same votes---and so we do not have to consider individual candidates, only the number of winners from each party directly. Then, we have that the party \textit{R} seat share is:
\begin{align}
n_R(y_R, \lambda)
&= \min_n \left[ \argmax_n \left[
y_R \sum_{i=1}^n \lambda(i) \right.\right. \notag\\
&\qquad\left.\left.
+ (1-y_R)\sum_{i = 1}^{N_k - n} \lambda(i)
\right]\right].
\end{align}

The inner $\argmax$ comes directly from the rule definition -- for a given $n$, suppose a fraction $y_R$ approves $n$ candidates. Then that will yield $y_R\sum_{i=1}^n \lambda(i)$ points, with the remaining fraction $1 - y_R$ approving $N_k - n$ candidates and yielding an additional $(1-y_R)\sum_{i = 1}^{N_k - n} \lambda(i)$ points. The outer $\min$ simply encodes tie-breaking in favor of party $D$.\footnote{Ties occur when $y_R\lambda(i+1) = (1 - y_R)\lambda(N_k-i)$ for some $i$.} The party \textit{D} seat share is $n_D(y_R,\lambda) \triangleq N_k - n_R(y_R,\lambda)$.

On the other hand, in STV, candidates of the same party may have different numbers of votes (in any round), as intra-party voter rankings may be arbitrary. One may thus, naively, believe that calculating party seat shares would require further assumptions on voter behavior (on voters' intra-party rankings) and carrying out the iterative procedure defined above, as the elimination and transfer scheme introduces substantial path dependence across candidates and rounds; in general with arbitrary voter rankings, such calculations are required. Doing so, with many voters and candidates, would be prohibitively computationally expensive as a subroutine in a redistricting optimization algorithm. The following proposition establishes that under our assumptions, party seat shares can be efficiently calculated (up to rounding) just from the vote shares, without dependence on voter rankings.\footnote{The first part of the result is a known property of STV (see, e.g.,~\citet{tideman1995single} and \citet{dummett1984voting}), under an assumption of ``solid coalitions'' as we have here. The second part also follows from, for our setting, the equivalence between PAV, d'Hondt's rule, and (up to rounding) STV. We provide a self-standing proof in the Appendix.}

\begin{restatable}[Seat shares under STV]{prop}{propSTVclosed}
Suppose -- for a given district with $M$ seats and $V \geq M(M+1)$ voters -- that a fraction $y_p$ of voters belong to each party $p \in \{R, D\}$, and that there are at least $M$ candidates per party.
Assume that each party's voters rank all same-party candidates above all other-party candidates, and ties are broken in party \textit{D}'s favor.

Let $n_R(y_R, STV)$ be the number of winners belonging to party $R$ using STV. Then, both $n_R(y_R, STV)$ and $n_R(y_R,\lambda_{\text{PAV}})$ are in $\{\floor{y_R M}, \ceil{y_RM} \}$.

\label{prop:stvequalspav}
\end{restatable}
At a high level, the proof establishes that no candidate receives meaningful (in terms of party vote share) votes from voters of the other party; in terms of party vote share, we can proceed with running two separate STV processes. Then, the sequential round path dependence does not matter, as the overall number of first-place votes for candidates of each party is invariant, and so the final partisan outcome depends just on initial vote shares. We note that the proof does not depend on our assumption of fractional STV; other transfer rules (such as random voter transfer) that preserve the number of surplus votes would induce seat shares according to the same formula.

Importantly for this work, \Cref{prop:stvequalspav} enables calculating party seat shares in a district that uses STV, without needing to (either computationally or in terms of data requirements) consider voters' individual rankings -- substantially simplifying the task for a (human or algorithmic) map drawer.\footnote{For precision, we further show that $n_R(y_R,\lambda_{\text{PAV}})$ is the unique integer such that
\begin{align}
y_R\cdot(M+1) - 1 \leq n_R(y_R, \lambda_{\text{PAV}}) < y_R\cdot(M+1),
\end{align}
and use this value.} While simulating STV is not computationally intractable for a single election, in our computational optimization below we need to run STV elections in each district as a subroutine to map generation and optimization, in the process of optimizing over about $10^{12}$ maps for a single state; as explained next in \Cref{sec:empiricalmethod}, we leverage this result to draw gerrymandered and fair maps for STV. Note that equivalence between PAV and STV is not true in general \citep{faliszewski2019proportional}. We use the relationship for the setting in this work to transfer intuition from PAV to STV.

\paragraph{Discussion of the solid coalitions assumption.} A core rationale for ranked choice voting is that it allows elicitation of preferences beyond two-party identity; our primary methods, assuming solid coalitions, do not capture this rationale. However, we believe that the assumption is suitable for the goals of this paper, to study \textit{partisan} implications of MMDs in the United States. The U.S. has single-winner elections for the executive (president, governor) and so will tend toward two parties for those offices. While it is possible that the solid coalitions assumption substantially breaks for legislatures with MMDs, it is at least equally reasonable that the two-party structure approximately remains, with subparties nested within them (preserving two-party solid coalitions). That being said, formally extending results beyond solid coalitions is of substantial interest. The technical hurdle is finding an alternative assumption that (1) is believable for the political context, and (2) does not require simulating STV \textit{within the inner loop of a redistricting optimization}, which is computationally prohibitive at scale. In other words, studying gerrymandering with social choice as we do requires \textit{some assumption} to enable efficient calculation of seat shares in each district, without simulation of elections. For example, more recent work by \citet{FRAanalysisDuchin} uses the same assumption as we do when \textit{generating maps}, but then simulates ranked choice votes beyond solid coalitions (as we do in \Cref{sec:intra}) \textit{given} those maps; their results, focusing on racial representation, are similar to ours.

We believe that finding a more reasonable assumption than (approximate) solid coalitions -- that is still suitable for optimization -- for the entire United States is unlikely, but such assumptions could be reasonable for specific political contexts (i.e., awareness of possible third parties or coalition splits that would form within a state). We leave such analysis for future work after more RCV elections have been run in the United States. In \Cref{sec:newwithoutsolidcoalitions}, we relax this assumption by calculating outcomes -- on maps calculated by assuming solid coalitions -- on simulated voter preferences where it does not hold; our insights generalize to this setting. %

\subsection{Empirical method}
\label{sec:empiricalmethod}

We study our research questions empirically, in the context of the United States House of Representatives. At a high level, we proceed as follows, {separately} for each state. Given historical voting data, we algorithmically generate maps for each state, social choice function $F$, and district size $K$: the most gerrymandered map for each party (the map that, given the voter data and the social choice rule, would produce the most winners from that party), the map with the smallest proportionality gap, and a neutral ensemble of random maps. For each map, we then calculate our metrics of interest for the relevant social choice rule.
We detail steps of this process next. In \Cref{sec:newwithoutsolidcoalitions,sec:intra}, we simulate full STV elections to study cross-party and intra-party effects, detailing the relevant methodological differences there.

\parbold{Data.} To calculate partisan seat shares, we require vote shares in each atomic block that will be used to compose districts. For each multi-district state, we use the geography-matched precinct-level statewide election returns from multiple election data repositories \citep{mgggstates, VOQCHQ_2018, DVN/NH5S2I_2018, DVN/UBKYRU_2019, DVN/UUCWPP_2011, DVN/AWE39N_2011, DVN/KX0YGR_2011, DVN/AN00LH_2011, DVN/WYXFW3_2011} for elections since 2008, aggregated by \citet{gurnee2021fairmandering} (including presidential, senate, and statewide gubernatorial elections). These data include vote shares for each candidate in each precinct in a collection of state and national elections. We filter these results to just Republican and Democratic candidates and average the two-party vote share across elections at both the block-level and the state-level.\footnote{Precinct-level results do not have a 1-1 map to census tracts. For each tract, the two-party vote share is generated as an average of its constituent precincts weighted by population and area overlap. About 80\% of precincts are wholly contained in one census tract.} For the atomic geographic blocks, we use census tracts, with population counts provided by the 5-year American Community Survey estimates.
As a general caveat, given the infrequency of statewide elections, these averages are noisy and sensitive to the idiosyncrasies of each race, candidate, and political environment, as well as true underlying shifts over time.

For the analysis in \Cref{sec:newwithoutsolidcoalitions,sec:intra}, we further require individual-level characteristics that can be used to simulate voter rankings. There, we use a \textit{voter file} provided to us by a private election analytics company, with individual-level voter data: in each block, an anonymized list of voters, along with their (potentially modeled) demographics (including ethnicity, gender, census block of home address); each voter is further assigned a party (either modeled or ground truth in states with party registration) and is scored on several ideological dimensions through a mixture of methods, including survey modeling, ecological inference, and individual-level voting history.\footnote{Calculating accurate scores is a challenging task, but we note that the dataset originates from a prominent company whose scores are used by many state and national parties, campaigns, and outside groups.} We calibrate the voter file to match the calculated party vote shares, by sub-sampling the voters according to their assigned party; for replicability, we provide a subsample of this dataset (with noise added) in our code repository. %

\parbold{Generating maps.} To generate maps, we extend the stochastic hierarchical partitioning (SHP) algorithm by \citet{gurnee2021fairmandering}, which recursively samples subdivisions of a state to create a large random\footnote{These districts are drawn without partisan information and favor compact districts; however, the exact distribution is difficult to characterize given the nested hierarchical dependencies.} ensemble of population-constrained and contiguous districts. These subdivisions are organized into a tree of nested regions that implicitly encodes an exponential number of distinct maps. {We extend the original tree growth logic to guarantee the existence of a split with child nodes each able to produce appropriately sized districts for a given MMD size policy.} Each district (leaf node) is scored by expected party seat share for each social choice function,\footnote{This step for STV is enabled by \Cref{prop:stvequalspav}. Due to the equivalence up to rounding, we use the same maps for PAV and STV, optimized for PAV.} and the tree can be efficiently traversed with a dynamic program to select the most gerrymandered map, or the tree leaves can be gathered into a secondary integer program to select the districts that create the most fair plan (for arbitrary definitions of fairness).

To ensure robustness, we estimate the variance of district vote share across elections to calculate the expected seat share per district and optimize for this expectation, rather than simply optimizing for the number of strict majority seats. This prevents, for example, the algorithm from drawing districts in which the historical vote share was 50.1\% for one party and declaring them safe seats for that party. We note that this also empirically lessens the sensitivity to the \textit{solid coalitions} assumption, if historical variance correlates with crossover votes in STV.\footnote{We note that variance of vote share (how much an area swings between parties \textit{across} elections) is not the same as splitting ballots within an election, and so this sensitivity check is not exact. However, both variance across elections and splitting ballots within an election would cause vote leakage from a party.}

All details of the original Stochastic Hierarchical Partitioning (SHP) algorithm can be found in the paper by \citet{gurnee2021fairmandering} and associated appendices; we further use their same geographic and electoral data (see Appendix Table 2 of \citet{gurnee2021fairmandering}). Here, we discuss the relevant algorithmic parameters and necessary modifications; see code for further details.
The main algorithmic difference from the original SHP algorithm is that we adapted ours to generate multi-member districts. To do this, instead of parameterizing a sample tree node by a region $R$ and total number of seats $s$, we needed to also specify the number of districts that node contains. This is required because at an intermediate node, the number of districts is not immediately derivable from the total number of seats in that node because of ambiguity in the number of $N_k$ versus $N_k+1$ sized districts (and, analogously, in our experiments simulating 3-, 4-, or 5-member districts in a state, as in the Fair Representation Act). Therefore, the number of seats is used just to balance the population, and the number of districts is used for all other tree operations (sampling valid splits, maintaining balance, etc.).

For each state, let $N$ be the number of seats for that state in the House of Representatives. Then, we construct maps for that state with $K$ districts, for each of $K \leq N$ and within $\{2, 3, \dots 10\} \cup \{12, 14, 16, 18, 20 \} \cup \{23, 26, 29, \dots 53\}$, where $53$ is the largest number of seats for any state.\footnote{Selecting every other number starting at 10 and every third number starting at 20 is done for computational tractability.} For each pair (state, $K$), we sampled the root node $\left(\frac{1000}{K}\right)^{1.2}$ times and each internal node $\left(\frac{300}{K}\right)^{0.5}$ times. These constants were chosen to balance computational cost and optimization quality. We used random-iterative center selection with Voronoi-weighted capacity matching to sample region centers and sizes. All districts are population-balanced within a $1\%$ tolerance, as described in \Cref{sec:empiricalmethod}. Each of these ensembles was then scored, optimized, and subsampled to create a final distribution over partisan outcomes. For the subsample used to calculate the distribution of outcomes under ``neutral'' maps, we use $1500\sqrt{K}$ maps for each (state, $K$) pair. The number over which the ``optimal'' maps (most fair, most gerrymandered) are calculated is more complicated: these maps are implicit in the tree structure. The set of tree leaves (districts) can be used to compose an exponentially large number of maps -- we can optimize over all these maps without enumerating them through a dynamic program (for most gerrymandered maps) or integer program (for most fair maps). For each setting, there are between $10^3$ (smallest states) and $10^{12}$ (largest states, most districts) such implicit maps that are optimized over.

\paragraph{Comparison to other methods.} We adopt this approach rather than the standard recombination Markov chain algorithm \citep{Deford2021recombination} and related local update methods \citep{fifield2020automated, autry2020multi, mccartan2020sequential, cannon2020voting, autry2021metropolized} to more efficiently model both partisan and nonpartisan mapmakers under multiple social choice functions. Our approach generates just one district ensemble (set of districts as leaf nodes of the tree), that can then be combined into one of exponentially many valid maps (in number of districts), optimizing each objective under each modeling scenario; in contrast, a single neutral recombination map ensemble would be unlikely to capture the full range of potential partisan gerrymandering in every setting, necessitating the use of a separate biased chain optimized for each combination of party and voting rule and increasing the amount of computation by almost an order of magnitude.  Furthermore, in a policy setting with a variable number of districts (as in the Fair Representation Act), a single hierarchical sample tree produced by SHP can efficiently capture all possible combinations of district sizes. In contrast, recombination would either require running 29 separate chains for the 29 different combinations of three-, four-, and five-member districts to fill California's 53 seats under the Fair Representation Act rules (times each combination of party and social choice rule for the full experiment), or would require creating rules for how to combine and re-partition districts of varying sizes. For example, when studying the effects of the Fair Representation Act, \citet{FRAanalysisDuchin} considers just one combination of district sizes per state.

On the other hand, an advantage of recombination methods is that they sometimes come with guarantees of the sampling distribution over maps, while no such analysis is available for the SHP algorithm. However, in the context of SMDs, \citet{gurnee2021fairmandering} find that SHP finds more extreme maps than other methods do.

\parbold{Our experiments.}
Putting things together, for each generated map, we can simply calculate our metrics of interest. %
We carry out the process for each social choice function and each state with 2 or more seats in the House, fixing as $N$ the number of seats allocated to that state after the 2010 census. The number of districts $K$ is from $1$ (a large MMD with $N$ seats) to $N$ (only SMDs). For each $K$, the set of district sizes $\{ N_k \}$ is selected such that all districts differ in size by at most one seat,\footnote{Formally, let $j = \floor{\frac{N}{K}}$, and $L = N \text{ mod } K$. Then there are $L$ districts of size $j+1$ and $K-L$ districts of size $j$. For example, when $N=6$ and $K=3$, $\{N_k\} = \{2, 2, 2\}$.  When $N=7$, $\{N_k\} = \{2, 2, 3\}$.} except in the analysis of the Fair Representation Act, in which case we generate maps with districts of size either 3, 4, or 5.

Overall, generating districts and calculating the metrics utilized over 40 CPU-weeks of computation, not counting the full STV elections simulated in \Cref{sec:intra}.

\section{Results: Inter-party balance with multi-member districts}
\label{sec:proportionality}
\begin{figure*}[tb]
			\centering
	\includegraphics[width=1\linewidth]{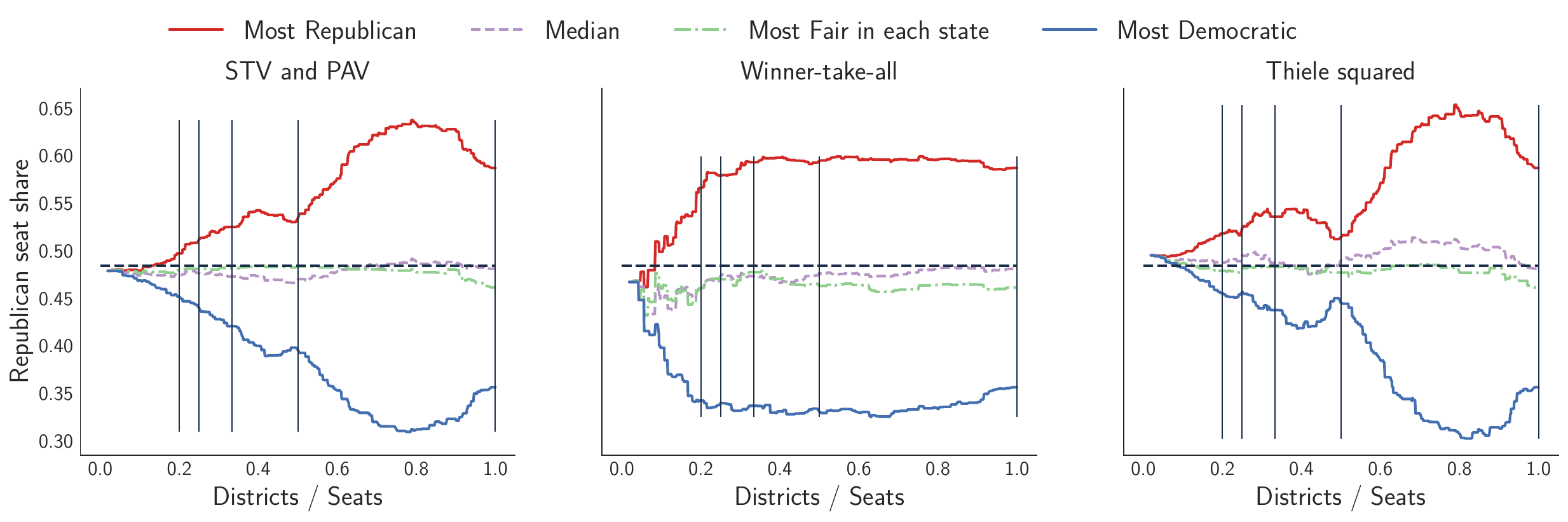}
	\caption{The Republican seat share \textit{over all states} as the number of districts is varied in each state. The horizontal line denotes the vote share fraction, i.e., the proportionality value. Each point aggregates over every state, with rounding and weighting by number of seats in the state, and vertical lines corresponding to when $N/K$, the average number of seats per district, is an integer. For example, the vertical line at $0.5$ corresponds to two-member districts in states with an even number of seats and all two-member districts except one one-member or three-member district in states with an odd number of seats. The right-most point is with SMDs, and the left-most point is if each state has one large MMD. ``Median'' refers to the median value found across random maps from the SHP algorithm, and ``Most fair in each state'' to the maps with the smallest proportionality gap. Overall, MMDs are effective at preventing the worst gerrymanders, especially with non-winner-take-all rules. Note that national seat share with Median and Fair maps only look invariant to district size because gaps cancel out between states (some favor Democrats, others Republicans). \Cref{fig:propbystate} shows that seat share is not invariant at the per-state level: even when optimizing for fairness, we require MMDs to reduce the per-state proportionality gap. }
	\label{fig:propdifferentmethods}
\end{figure*}

\begin{figure}[tb]
	\centering
	\begin{subfigure}[b]{0.42\textwidth}
		\centering
		\includegraphics[width=\textwidth]{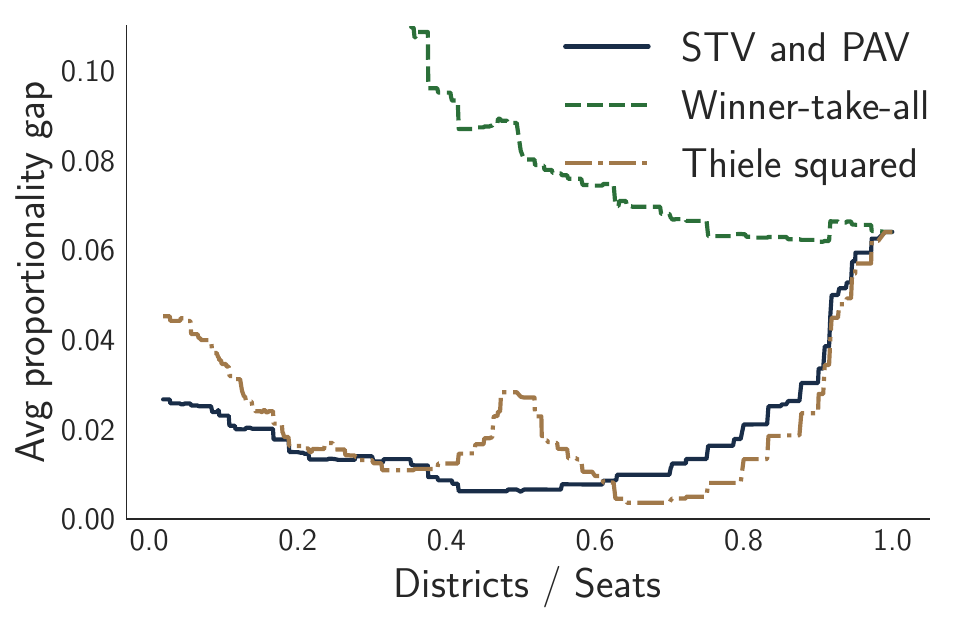}
		\caption{Avg per-state magnitude of gap in ``Most Fair'' maps}
		\label{fig:propmethods_mostfair}
	\end{subfigure}
	\hfill
	\begin{subfigure}[b]{0.56\textwidth}
		\centering
		\includegraphics[width=\textwidth]{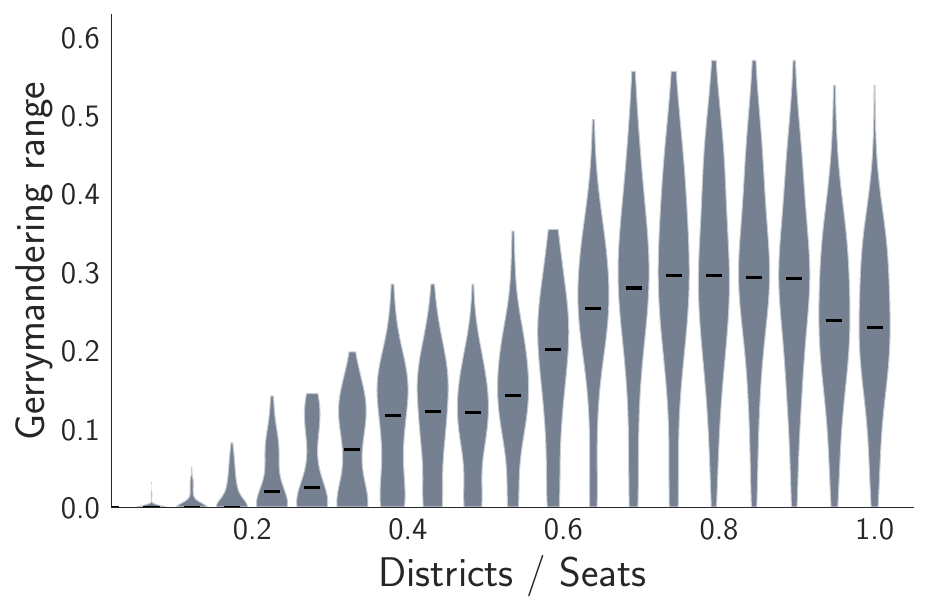}
		\caption{Gerrymandering range under STV and PAV}
		\label{fig:gerrymanderingrange}
	\end{subfigure}
	
	\begin{subfigure}[b]{0.56\textwidth}
		\centering
		\includegraphics[width=\textwidth]{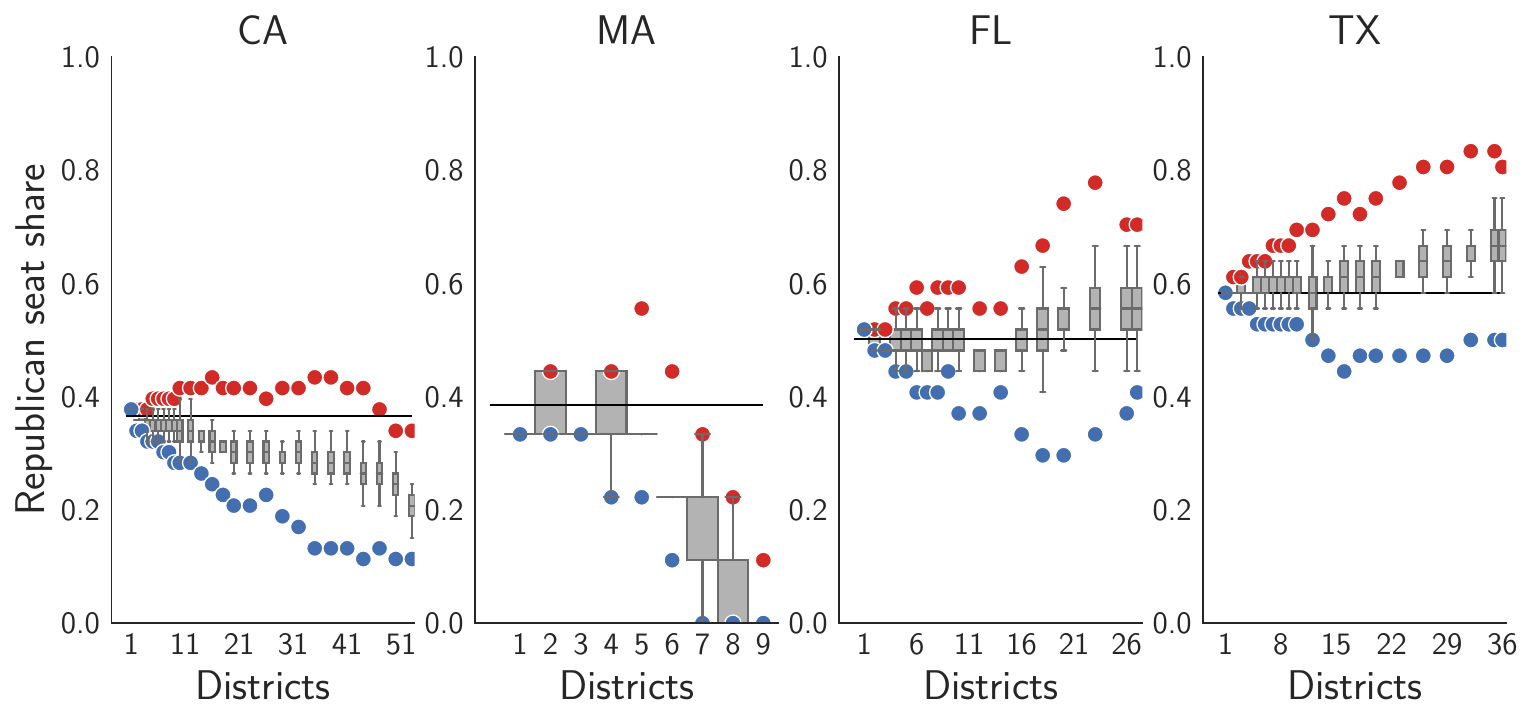}
		\caption{Partisan lean distribution under STV and PAV}
		\label{fig:propbystatebox}
	\end{subfigure}
	\caption{How the partisan lean and proportionality gap vary \textit{at the state level} with voting method and the number of districts. (a) The absolute value of the state-wise proportionality gap in the ``Most Fair'' map in each state. Even if a redistricting agent wanted to close the proportionality gap, it could not do so with SMDs. (b) The distribution of the ``gerrymandering range'' (difference in Republican seat share between extreme Republican-optimized and Democrat-optimized maps), as a measure of how much redistricting matters. Each violin plot shows the distribution over states. While the range varies substantially at the state level for SMDs, it shrinks across states as district size increases. (c) The per-state version of \Cref{fig:propdifferentmethods}, showing the full distribution of maps and the extreme gerrymanders for four states. While there are substantive state-level gaps with SMDs even in the most proportional maps, the gaps become negligible with even just two-member districts and STV. Note that in \Cref{fig:propdifferentmethods}, these gaps cancel out nationally, as some state-wise gaps favor Democrats and others favor Republicans. Qualitatively similar results to (a) are shown for the Median and gerrymandered maps in Appendix \Cref{fig:propmethods_median}, and plots for all states as in (c) are in Appendix~\Cref{fig:boxall}.}

	\label{fig:propbystate}
\end{figure}

We first analyze how MMDs affect partisan balance, both in how MMDs empower independent commissions and in how they constrain adversarial partisan gerrymanderers. \Cref{sec:mitiggerr} contains our main results; we illustrate the effects of each social choice function and district size, for both extreme gerrymanders and fair redistricting. In \Cref{sec:design}, we turn to \textit{design} questions, discussing the effect of the Thiele rule parameter and heterogeneity across states in terms of `optimal' district size; here, we further analyze how the Fair Representation Act navigates such tradeoffs.

\subsection{Gerrymandering, proportionality, and competitiveness with multi-member districts}
\label{sec:mitiggerr}
\Cref{fig:propdifferentmethods} contains our main result; for each social choice function, it shows how the overall (across states) partisan seat share varies with the number of districts. \Cref{fig:propbystate} zooms in on the state level.

\parbold{Preventing intentional gerrymanders.}  MMDs are effective at preventing the most extreme gerrymanders, but only with non-winner-take-all rules. As expected, in the limit, with one large MMD in each state and using STV (equivalently, PAV), the proportionality gap in each state is negligible. Perhaps more surprisingly, just moving from SMDs to two-member districts provides about half the benefit in terms of reducing the proportionality gap for the most extreme gerrymanders. Three-member districts further reduce this gap, e.g., from a maximal seat share of around 60\% to around 57\% for Democrats. With five-member districts, the range of outcomes is between 45\% and 50\% Republican seat share, even if one party controlled all redistricting nationally. To further understand the role of gerrymandering as district size increases, define the gerrymandering \textit{range} as the difference in Republican seat share between extreme Republican-optimized and Democrat-optimized maps; a small range suggests that outcomes would not vary much within the range of feasible maps. \Cref{fig:gerrymanderingrange} shows, for each district size, the distribution of the gerrymandering range across states. While the range varies substantially at the state level for SMDs, it shrinks across states as district size increases, showing that the effect appears both at the state and national level.

What explains these results? With SMDs, an \textit{R} gerrymander would draw a map such that as many districts as possible have a majority of \textit{R}s (up to a tolerance for robustness). It does so by \textit{cracking}, creating many districts in which the \textit{D} vote share is just below $\frac{1}{2}$ (thus electing an $R$ candidate), and, as necessary, \textit{packing}, creating a few districts in which the \textit{D} vote share is as close to 1 as possible. Such a map would maximize \textit{D} \textit{wasted votes}. Now, consider PAV, i.e., a Thiele rule with $\lambda(i) = \frac{1}{i}$, and two-member districts. Any district in which the \textit{D} vote share is between $\frac{1}{3}$ and $\frac{2}{3}$ would elect one member from each party. Cracking thus requires creating districts in which the \textit{D} vote share is less than $\frac{1}{3}$. Packing -- wasting \textit{D} votes -- could mean targeting the  \textit{D} vote share to be just \textit{below} $\frac{2}{3}$ (e.g., a \textit{D} vote share of $0.6$ to elect one member for each party), or close to $1$ but then giving up both seats in the district. In each case, gerrymandering requires more precision and wastes fewer votes for the opposing party.\footnote{Under our voting assumptions, for every PAV election with $N_k$ seats, $\frac{1}{N_k + 1}$ of the votes are wasted in total, and so the potential wasted votes for either party also goes to $0$ with district size.} Simultaneously, having fewer districts to draw reduces the degrees of freedom available. The corresponding bands become narrower as district size increases. 

The trend toward proportionality is not monotonic; a mixture of district sizes increases degrees of freedom, enabling gerrymanders even more extreme than those possible under SMDs. Consider one- and two-member districts; the gerrymandering party \textit{R} could then pack or crack party \textit{D} in the one-member districts, and use the two-member districts to elect one candidate from each party with \textit{R} vote share just above $\frac13$. Similarly, an urban gerrymandering party \textit{D} could waste fewer votes winning one-member urban districts, while still getting above $\frac13$ of the vote in two-member rural districts. (In other words, all but the most extreme two-member districts will result in one member elected from each party, with STV). The pattern repeats, to a lesser extent, between every integer division.

Similar arguments apply to other Thiele rules with strictly decreasing $\lambda$, and are amplified with sharper diminishing returns. Consider the Thiele squared voting rule. It favors -- even more than PAV -- 50/50 outcomes in each district, since voters who have fewer of their preferred candidates elected gain more from having an additional candidate elected than voters who already have many of their preferred candidates elected. Thus, when there are two seats per district, the Thiele squared rule strongly incentivizes outcomes where each party wins one seat in each district. In \Cref{fig:propbystate}, this increases the proportionality gap for Thiele squared voting rules when there are about two seats in each district, as 50/50 may not be proportional. %

\parbold{Enabling proportional redistricting.} The above analysis considers the most extreme gerrymanders. However, commissions that draw `fair' maps are increasingly common \citep{spencer_who_2020,cain2011redistricting}; we now show that MMDs enable the drawing of proportional maps that would not be possible under SMDs.
Consider \Cref{fig:propmethods_mostfair}, which shows the average absolute proportionality gap in each state for the map that minimizes this gap. We find that there is a substantive state-level gap that virtually disappears -- up to rounding with a fixed number of districts -- even with two-member districts and STV.

\Cref{fig:propbystatebox} further demonstrates the SMD `Massachusetts problem,' as elucidated by \citet{duchin2019locating}, and shows that MMDs fix it;\footnote{While \citet{duchin2019locating} find \textit{no} SMD in MA with a majority of Republicans, we find it possible to draw such a district. The difference is data. They choose to include only presidential and senate elections, and their data is from a different time period -- they find the effect in elections from 2000 to 2010. We also include statewide gubernatorial races where Republicans have had more electoral success, and so our averaging concludes that Republicans compose $40\%$ of the state. All redistricting approaches are sensitive to such data choices; results remain qualitatively the same, though details for any given election may differ.}  because of how evenly Republicans are distributed (diffused) across the state, there is no way to draw SMDs such that a proportional number of Representatives are Republican.  Intuitively, suppose party \textit{R} has vote share $\frac13$. With SMDs, to achieve proportionality a commission would need to draw districts such that party \textit{R} is in the majority in $\frac13$ of the districts, which may not be possible or may require atypical, contorted maps.

A single MMD with PAV provably, approximately achieves proportionality \citep{skowronProportionalityDegreeMultiwinner2019}, solving the Massachusetts problem up to rounding. Our results indicate that even multiple two-member districts enable fair maps and solve the problem in practice. Intuitively, in the above example, a commission would just need to draw a map with enough districts  with \textit{R} vote share above $\frac13$. Note that, by the pigeonhole principle, with overall party $p$ vote share $v_p$, for any map there will always be at least one district in which its vote share is at least $v_p$; so, as the threshold to win 1 seat decreases with district size, an ever-smaller minority party is guaranteed at least one seat.

Finally, we note that, as illustrated by the average proportionality gaps of the ``Median'' maps in Appendix \Cref{fig:propmethods_median}, even small multi-member districts with PAV or STV eliminate the issue of `natural' gerrymanders, in which even `neutral' maps -- drawn by a redistricting algorithm that ignores partisan vote shares -- favor one party, due to the natural geographic distribution of voters. Such natural gerrymanders have played a central role in discussion of a justiciable gerrymandering standard; our results indicate that with even small MMDs, an independent commission would not need to draw maps that substantially differ from neutral maps in order to achieve proportionality.

\begin{figure}
	\centering
	\begin{subfigure}[b]{0.24\textwidth}
		\centering
		\includegraphics[width=\textwidth]{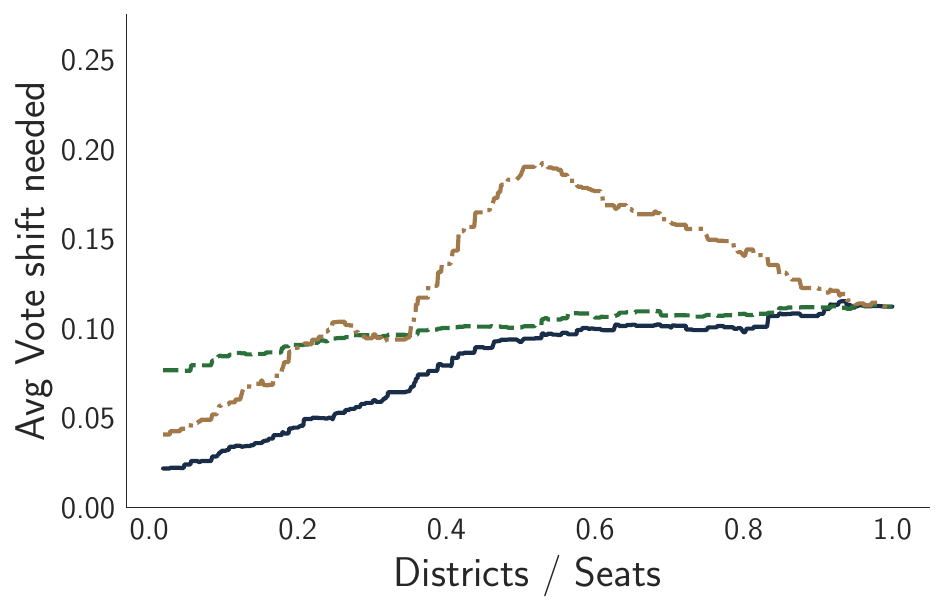}
		\caption{Median map}
		\label{fig:competmedian}
	\end{subfigure}
	\hfill
	\begin{subfigure}[b]{0.24\textwidth}
		\centering
		\includegraphics[width=\textwidth]{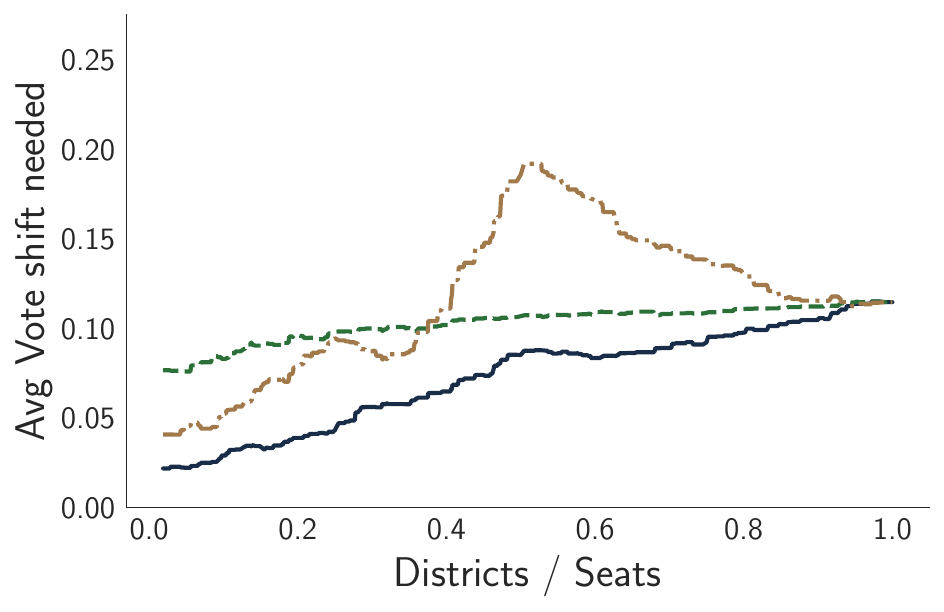}
		\caption{Most fair map}
		\label{fig:competmostfair}
	\end{subfigure}
	\begin{subfigure}[b]{0.24\textwidth}
		\centering
		\includegraphics[width=\textwidth]{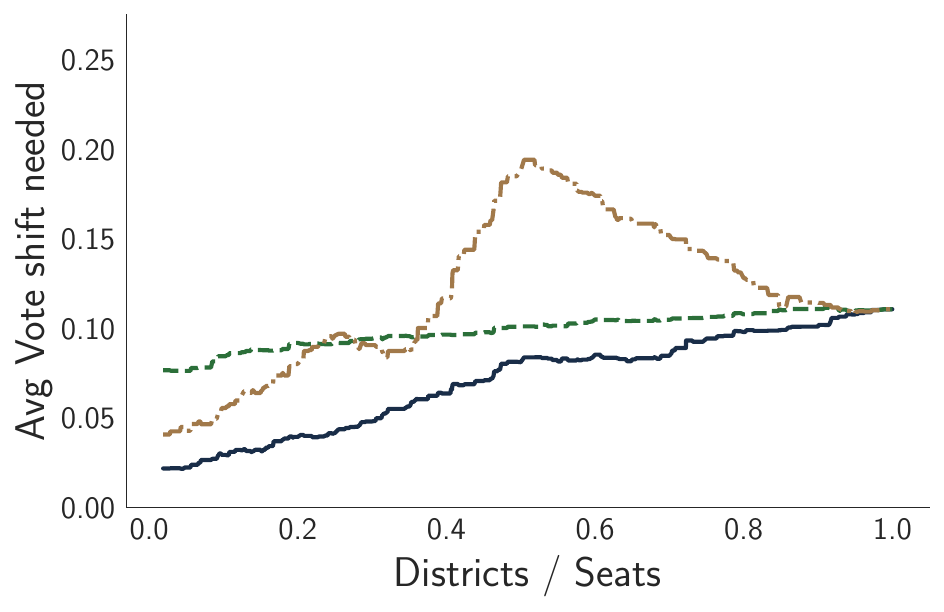}
		\caption{Most Republican}
		\label{fig:competmostrep}
	\end{subfigure}
	\hfill
	\begin{subfigure}[b]{0.24\textwidth}
		\centering
		\includegraphics[width=\textwidth]{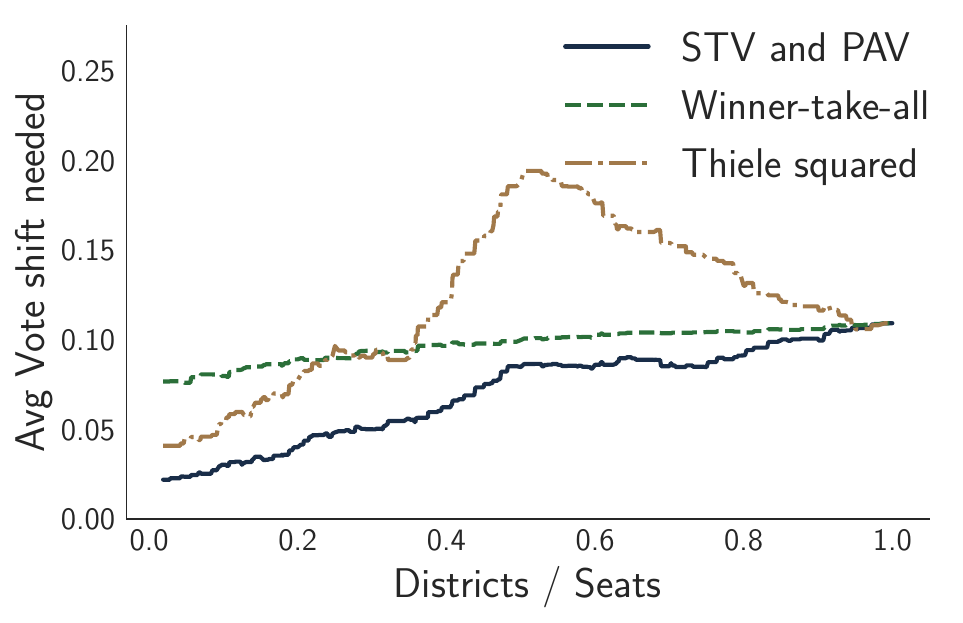}
		\caption{Most Democratic}
		\label{fig:competmostdem}
	\end{subfigure}
	\caption{The average vote shift needed in each map (averaged across districts in a map and across states) to shift the number of seats won by each party by at least one. The larger the vote shift needed in a district, the less \textit{competitive} is the district. The four plots show the various maps selected. For example, (a) shows the vote shift needed when the median (in terms of party seat share) maps are selected. With PAV/STV, competitiveness generally increases with district size.}
	\label{fig:competitiveness}
\end{figure}

\parbold{Competitiveness.}
Up to now, we have primarily considered \textit{proportionality} -- how the partisan seat share reflects the underlying vote share in each state. Here, we consider \textit{competitiveness}: the average vote shift needed to shift the number of seats won by each party by at least one. For example, in a winner-take-all election,  a Republican vote share of 0.6 in a district would lead to a margin of victory of 0.1. In a PAV election with two winners, the same district would have a victory margin of $(\frac{2}{3} - .6)$, as a Republican vote of $\frac{2}{3}$ (with appropriate tie-breaking) would lead to two elected Republicans as opposed to one. While competitiveness is controversial as a \textit{goal} for redistricting (as an objective when drawing maps) \citep{deford2020computational}, it is considered an important dimension along which to evaluate a map. Some claim that uncompetitive districts, for example, could depress participation or contribute to polarization \citep{altman2015redistricting,ainsworth2022district,cancela2016explaining,moskowitz2019reevaluating,gerber2020one}.

One potential concern with multi-member districts is that they may lead to uncompetitive districts if the number of members in each district is small. With two-member districts, for example, most districts will have one member from each party -- and one of the parties would need at least $\frac{2}{3}$ of the vote to win both seats in the district. Our results, however, indicate that this concern is overstated. \Cref{fig:competitiveness} illustrates the average vote shift needed to shift the seat share; regardless of which map is considered, with PAV or STV competitiveness increases with district size (the average vote shift necessary to change at least one seat goes down). {However, it is the case that most two-member districts would be split 1D-1R, even in what would have previously been considered politically safe regions.} This effect is amplified with Thiele squared voting rules, which further favors balanced outcomes. %

\subsection{Design recommendations and discussion}
\label{sec:design}
Above in \Cref{sec:mitiggerr}, we primarily discuss how partisan balance -- proportionality and competitiveness -- vary with respect to district size, generally across settings. Here, we expand on the role of the exact social choice function, as well as heterogeneity across settings when it comes to the effect of district size. We note that our results -- especially the non-monotonicity, heterogeneity, and non-asymptotic behavior in district size -- underscore the importance of an empirical approach to supplement theoretical analysis.

\begin{figure}
   \centering
   \includegraphics[width=\linewidth]{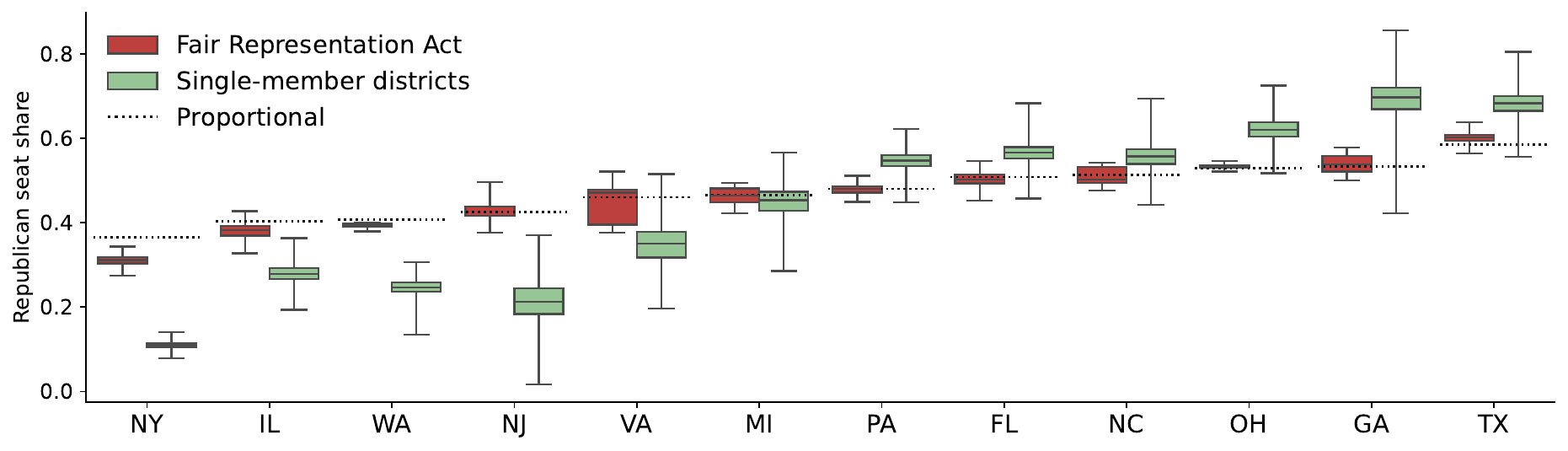}
   \caption{Comparison of Republican seat share distribution of the Fair Representation Act-compliant ensembles using the STV/PAV voting rule to single-member district baseline for a selection of states (see Appendix \cref{fig:hr4000_by_rule} for additional voting rules).}
   \label{fig:hr4000_smd_baseline}
\end{figure}

\parbold{Effect of social choice function.} Perhaps the most obvious takeaway is that using a winner-take-all procedure with MMDs, as was popular in the twentieth century and still in effect in some state legislatures, enables rampant gerrymandering at nearly all district sizes and should be avoided (in fact, while proportionality gaps with one large winner-take-all MMD in each state cancel out at the national level in \Cref{fig:propdifferentmethods}, as shown in \Cref{fig:propmethods_mostfair} they increase to close to the theoretical maximum with larger MMDs).

Our results also illustrate the differential effects of the Thiele design parameter $\lambda$, especially as a function of district size. Our experiments test $\lambda(\cdot) = 1$ (winner-take-all), $\lambda(i) = \frac{1}{i}$ (PAV), and $\lambda(i) = \frac{1}{i^2}$ (Thiele Squared). Among these, we find that as we increase the rate at which $\lambda$ decays in $i$, the corresponding Thiele rule more strongly results in legislatures that are equally composed of members of different groups, regardless of the size of the underlying group. Especially with even-member districts, then, such a rule leads to proportional representation when the parties are approximately equal-sized---even more so than does STV or PAV---at the cost of proportionality for settings with imbalanced party size. For example, consider the difference between the national-scale proportionality of Thiele squared in \Cref{fig:propdifferentmethods} with the per-state proportionality gap in \Cref{fig:propbystate}, with two-member districts. More generally, our empirical framework may be helpful in complementing theoretical analysis of Thiele rules (see, e.g., \citet{skowronProportionalityDegreeMultiwinner2019}) in understanding the effect of the design parameter in a given real-world setting, given the actual distribution of voters.

\parbold{District size design heterogeneity and recommendations.}  Our discussion above suggests that two- or three-member districts (as possible depending on the number of seats) with STV or PAV would work well across most political conditions in the United States. In particular, fairness-minded independent commissions can achieve proportional outcomes in every state up to rounding, and advantage-minded partisans have their power to gerrymander significantly curtailed.

 However, there is also important heterogeneity; individual states can further promote proportional outcomes by tuning MMD parameters based on their population, partisan lean, and political geography. See Appendix \cref{fig:boxall} for full state-by-state results; here, we discuss the most relevant factors influencing `optimal' district size. We observe that small and highly partisan states both benefit from larger districts. Smaller states have fewer opportunities for wasted votes to cancel out across districts and so can more easily become disproportionate. In contrast, highly partisan states often suffer from the `Massachusetts problem,' where minority parties have difficulty breaking the requisite threshold in any district due to their diffuse voter base. In general, states with significant geographic self-segregation could protect the political power of both concentrated and diffuse voters with larger districts that inherently waste fewer surplus and losing votes.

 As such, flexibility across states may be beneficial; however, as \Cref{fig:propdifferentmethods} shows, such flexibility may also increase gerrymandering, if within a state partisan gerrymanderers can use a mixture of district sizes. We evaluate this trade-off by studying the potential outcomes of redistricting under the Fair Representation Act, a proposal mandating that the number of seats per district satisfy $N_k \in \{3, 4, 5\}$. As before, we sample random maps and construct optimal maps for each party, when a map can have districts of size 3, 4, or 5, such that the total number of seats adds up to the number of seats in the state. \cref{fig:hr4000_smd_baseline} suggests that this choice would allow flexibility without giving gerrymanderers too many degrees of freedom in district size: it enables proportionality almost everywhere, and it substantially reduces gerrymandering capacity over the baseline of using just single-member districts. Appendix \cref{fig:hr4000_size_baseline} further shows that such flexibility does not offer significant additional capacity to gerrymander over a baseline of using just three- or just five-member districts, in contrast to the case of mixing single- and two-member districts.

 {This finding is important because coherent communities substantially vary in size and so are best represented by variable-sized districts. For instance, in Texas, the Houston metro area may be best covered by a five-member district, whereas San Antonio, a city with roughly half the population, may best be covered by a three-member district. Under a fixed number of seats per district, either Houston is split, or San Antonio is grouped together with disjoint communities.} The Fair Representation Act's approach would allow both district sizes. In Appendix \ref{sec:intra}, we continue evaluating such trade-offs, studying \textit{intra-party} effects of MMDs with STV, where the \textit{solid coalitions} assumption no longer holds.

\FloatBarrier
\section{Results: Beyond Solid Coalitions}
\label{sec:newwithoutsolidcoalitions}
So far, we have made a \textit{solid coalitions} assumption with STV: that each voter ranks all candidates of their own party above all candidates of the other party. In practice, voters may not behave this way, and may cross party lines in how they rank candidates in any given election. In this section, we empirically show that such crossover voting does not affect our primary results -- that multi-member districts with STV blunt the effect of gerrymandering -- under a standard model of such voting. Intuitively, gerrymandering with any number of districts becomes less effective with unpredictable cross-party voting (since the optimization target has higher variance), and we empirically find that this effect is additive to the effects of multi-member districts.

The biggest challenge to studying such effects is that \Cref{prop:stvequalspav} does not hold without the solid coalitions assumption, and so we must now generate voter rankings and simulate STV, for each district in each potential map.\footnote{In this section, we exclusively consider STV, as it allows voters within a party to prioritize different candidates, without risking their party overall representation by not approving a same-party candidate (as could happen with approval voting based methods like Thiele rules). Studying approval methods would require analyzing primaries, to select exactly $N_k$ candidates from each party.} To simulate voter rankings, we leverage additional data that characterizes individual voter partisan lean, and then construct voter rankings over simulated candidates using a random utility model.

\subsection{Methods and assumptions}
\label{sec:solidcoalitionsmethod}
To study settings beyond the solid coalitions assumption, we first construct plausible voter rankings and candidate distributions; second, we simulate STV elections given a map.

\paragraph{Voter generation and multi-dimensional preferences.}
The key step is to develop a model for how a voter will vote given a menu of (hypothetical) candidates. Up to now, we have only assumed that voters approve all candidates of their party or rank them all above candidates of the other party, but here we require a model for how voters rank candidates both within and across parties. We do so as follows, using the voter file described in \Cref{sec:empiricalmethod}. Each voter in our dataset is assigned scores along several ideological and demographic dimensions derived from survey responses and administrative records. In this work, we use a single scalar \emph{partisan score}, taking values in $[0,100]$, which represents the modeled probability that the voter supports the Democratic candidate in a two-party general election. This score is constructed using party registration, historical primary participation, demographic variables, and related covariates, and has been shown to correlate strongly with issue-specific preference scores \citep{predictwise2021}. We threshold the partisan score at a fixed value of 50: voters with scores above the threshold are classified as Democrats, and those below as Republicans. Each voter is also associated with a census tract corresponding to their home location. Given a map, this uniquely assigns each voter to a district. The dataset is rebalanced to match the state-level vote shares in recent statewide elections, as computed in \Cref{sec:empiricalmethod}. For computational tractability, in each simulation we sample at most 50{,}000 voters per state. For replication purposes while preserving data privacy, in our code repository we provide a subsample of 50{,}000 voters, with noise added to the scores.

\paragraph{Candidate generation.}
Candidates are generated independently for each $(\text{party}, \text{census tract})$ pair. For each such pair, we create a single candidate whose partisan score is set to the median partisan score among voters of that party residing in the tract. This candidate is eligible to run in any district containing that census tract. For computational reasons, we cap the total number of candidates per state at 1{,}000, sampling at most 500 candidates from each party.

\paragraph{Voters and candidates in an election.}
A map partitions census tracts into districts. Since both voters and candidates are associated with census tracts, each belongs to exactly one district under a given map. In a given district-level election, voters may rank only candidates whose census tracts lie within the same district. Depending on the map, a district may contain a large number of candidates (for example, if the district spans most or all of a state). For both computational cost and statistical reasons arising from large candidate pools, we further subsample candidates within each district.

Specifically, from each party we sample at most a number of candidates equal to the number of seats to be filled in the district. This restriction is imposed for the following reason. Under the random utility models we consider, as described below, utilities are composed of a deterministic component that favors ideologically proximate candidates and an i.i.d.\ noise term. When the candidate pool is very large, extreme-value effects dominate: with max concentrating noise distributions such as the Normal distribution, the highest-utility candidates for a given voter will all be drawn from the voter's own party with high probability. As we aim to study crossover voting and go beyond the solid coalitions assumption, we subsample the candidates. This subsampling can reflect a primary or a natural limitation on the number of competitive candidates from each party in a district, enabling meaningful crossover votes.

\paragraph{Turning voter preferences into rankings.} In each simulated election, we need a ranking over candidates from each voter. We generate complete rankings over candidates using a random utility model (RUM). Let $s_v$ denote the partisan score of voter $v$ and $s_c$ the partisan score of candidate $c$. The utility that voter $v$ assigns to candidate $c$ is
\[
u(v,c) = -|s_v - s_c| + \varepsilon_{vc},
\]
where $\varepsilon_{vc}$ are i.i.d.\ idiosyncratic values drawn from a distribution $F$. We consider two choices for $F$: (i) a Normal distribution with mean $0$ and standard deviation $\sigma$, and (ii) a Gumbel distribution with location parameter $0$ and scale parameter $\beta$ (implying mean $0.5772\beta$ and variance $\tfrac{\pi^2}{6}\beta^2$). We report results for $\sigma, \beta \in \{0, 5, 50, 100, 250, 500\}$. When $\varepsilon_{vc}$ follows a Gumbel distribution, the induced distribution over rankings coincides with the Plackett-Luce model \citep{luce1959individual,plackett1975analysis,mcfadden1972conditional}. When $\varepsilon_{vc}$ is Normally distributed, this corresponds to the classical Thurstone-Mosteller random utility model \citep{mosteller1951remarks,thurstone2017law}. Unlike the Gumbel distribution case, Normal noise does not satisfy independence of irrelevant alternatives. While both models differ in their tail behavior, we find that our qualitative insights do not change with the noise model considered. For intuition on the noise levels, with a Gumbel random variable with $\beta = 50$, in our data a voter ranks a candidate from the other party first with probability about 25\%; with $\beta = 5$, this probability is about $5\%$. However, we note that the mapping from such crossover probabilities to map robustness is not straightforward, given that we subsample candidates and voters. Thus, we report results across a range of noise levels, to show that our insights are robust across different levels of crossover voting.

Finally, unlike in actual elections, we assume that there is no \textit{ballot exhaustion}. Each voter submits a complete ranking of all candidates in their district, ordered by decreasing realized utility.

\paragraph{Simulating STV elections.}
 Given the candidates, sampled voters, and the voters' rankings over the candidates, we run fractional STV for each district in each given map to determine the elected candidates. Fractional STV is STV as defined above. In each round, either a winner is selected if they have at least the Droop quota number of first-place votes, or a candidate with the least votes is eliminated. This candidate's votes are transferred, by eliminating this candidate from each voter's list. In fractional STV, votes are transferred as follows. Formally, suppose the winning candidate receives $v > Q$ first-place votes. Then, a fraction $\frac{v - Q}{v}$ of this candidate's votes is in excess of what is needed to win. Thus, each voter for this candidate has their weight multiplied by $\frac{v - Q}{v}$. For example, suppose the Droop quota was 5, and a candidate received 6.3 first-place votes (fractional votes arise from earlier-round transfers). Then, each of their voters has their weight multiplied by $\frac{1.3}{6.3}$ for the next round.

Running these STV elections utilized over 120 CPU-weeks of compute, on top of the map generation discussed in \Cref{sec:empiricalmethod}. This runtime is for \textit{given} maps, underscoring the necessity of \Cref{prop:stvequalspav} to study partisan seat shares without needing to simulate STV during the redistricting optimization process.

\paragraph{Discussion.} There has been much work on spatial voter models aiming to characterize such behavior based on the ``distance'' between the voter and each candidate~\citep{adamsModerateVotersWeigh2017,tausanovitchDoesIdeologicalProximity2018,shorIdeologyUSCongressional2018,jesseeSpatialVoting20042009}. Random Utility Models in general, and the two we use in particular, reflect ``spatial'' theories of voting \citep{enelow1984spatial,merrill1999unified,azari2012random} and have been widely used in the social choice and political science literatures to simulate preference profiles and analyze voting rules. However, modeling voters' ranked preferences is a hard challenge -- it is not clear that voters behave according to such ideological spatial positioning. Thus, our results should not be interpreted as what \textit{would} necessarily happen with multi-member districts. We note, however, that we expect our main insight -- that multi-member districts blunt the effects of gerrymandering -- to hold more generally, as gerrymandering becomes less effective when voter behavior is less predictable. Our approach here provides a first step towards quantifying such effects, by providing a plausible model for voter behavior and simulating STV elections accordingly.

\begin{figure}
	\centering
		\centering
		\includegraphics[width=.5\textwidth]{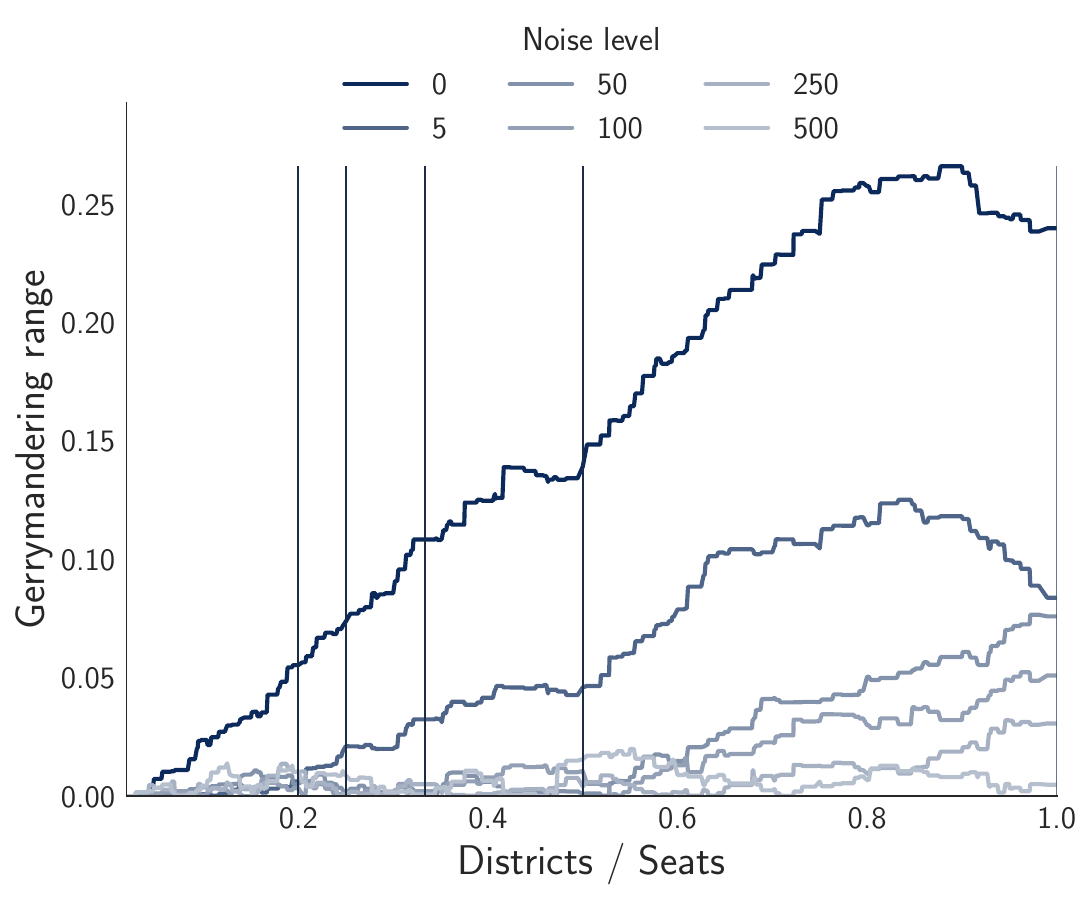}
		\caption{Gerrymandering range (difference in Republican seat share between Republican-optimized and Democrat-optimized maps) for different voter ranking noise levels, to simulate cross-party voting. Here, the idiosyncratic noise term for each voter-candidate pair is drawn from a Gumbel distribution with the noise level corresponding to $\beta$ (with higher noise corresponding to higher variance). Results with Normal distribution noise are in \cref{fig:noisegerrymanderinggap_normal} and are qualitatively similar. As noise levels increase, the range of seat shares feasible using gerrymandering decreases, and the effect is in addition to that of using bigger districts with STV. Intuitively, higher levels of voter noise correspond to gerrymandering being harder, since individual voters differ more from what is predicted.}
		\label{fig:noisegerrymanderinggap}
\end{figure}

\subsection{Results}
\label{sec:newwithoutsolidcoalitions_results}
\Cref{fig:noisegerrymanderinggap} contains our results with Gumbel noise, showing how the gerrymandering range contracts both with increased voter noise (i.e., crossover voting) and district size, with the effects being cumulative. In particular, the zero-noise line shows the gerrymandering range (difference in Republican seat share between Republican-optimized and Democrat-optimized maps) when votes are simulated using the above approach (with sampled voters and synthetic candidates); we average across the 25 most extreme maps calculated for each party for each setting with zero noise. Results without noise largely reflect the range in \Cref{fig:propdifferentmethods}, with some noise introduced by our voter sampling process (recall that we rebalanced our voter file to reflect the state-wise seat shares but not necessarily at each precinct).

Then, for each of the other lines, we show the empirical gerrymandering range for the same set of maps as optimized in the zero-noise case; now, however, voter rankings over candidates are noisy, as described above. This reflects a setting where parties would gerrymander without factoring in crossover voting (in practice, they might factor in the potential of crossover voting but not each voter's realization).

As voter noise increases, there are more crossover votes, and the gerrymandering range decreases. At a high level, the result establishes that crossover voting does not affect the qualitative insights of our work: the gerrymandering range decreases with STV as districts become larger, and much of the benefit is realized with three-member districts. The effect of noise is intuitive, as it makes the optimizer's task more difficult when it is not foreseeable (the optimizer draws maps assuming zero noise). \Cref{fig:noisegerrymanderinggap_normal} replicates this plot using Normal noise, and \Cref{fig:noisegerrymanderinggap_expostoptimization} shows the result when the noise is perfectly foreseeable, i.e., the optimizer observes the voter ranking realizations over candidates when drawing maps, and different maps may lead to different candidate and ranking realizations. Both alternative analyses replicate the results; in the latter case with foreseeable voter rankings, noise increases gerrymandering ability (since the additional noise yields more extreme voter preferences), but gerrymandering range decreases with district size.

More broadly, this section gives an approach to evaluating voting rules and redistricting under more complex voting preferences: generate voter rankings from individual-level data, simulate the social choice rule on a fixed set of maps, and compare outcomes along interpretable dimensions. Here, we relax solid coalitions across parties, allowing voters to rank some opposite-party candidates highly through a random utility model. \Cref{sec:intra} uses the same high-level approach but focuses on intra-party effects. As a result, it retains party-line ordering across parties but allows voters within the same party to disagree about same-party candidates. There, voters rank candidates according either to partisan-score proximity or geographic proximity (and without noise), allowing us to study how MMDs affect intra-party diversity and geographic cohesion. The appendix results suggest a tradeoff: larger MMDs can increase within-party diversity among winners, but very large MMDs weaken the geographic cohesion of the voter coalitions supporting each representative. Thus, medium-sized MMDs may be a sweet spot for balancing these effects, while also achieving most of the proportionality benefits of larger MMDs (as shown in \Cref{sec:proportionality}).

\FloatBarrier

\section{Discussion}
\label{sec:discussion}
\paragraph{Discussion and limitations.} In this work, we focus empirically on the United States and on MMDs. There are many other systems designed for proportionality in use internationally; a comparative study would be of use. Future work can also adapt our methodology to study partisan gerrymandering in state and local legislatures, especially in the states that already use (winner-take-all) MMDs. Such legislatures stand to benefit even more from MMDs than congressional districts given that they cannot rely on interstate cancellation of partisan bias as is currently the case in the House of Representatives. While we expect the high-level insights to hold, details may differ when zooming in with more seats in a smaller region due to, e.g., effects of geography.

There is much left to do to understand multi-member districts in practice, especially how they relate to other constraints and laws. (1) The 2021 version of the Fair Representation Act has a provision\footnote{Sec. 205, here: \url{https://beyer.house.gov/uploadedfiles/fair_representation_act_117th_final.pdf}} that MMDs should not be used in states where doing so would violate the Voting Rights Act; it would be crucial to analyze under what circumstances this could happen and how one would measure violations, for example by diluting the voting power of racial minorities. (2) More generally, we impose no constraints beyond contiguity and population balance, but many areas impose additional requirements (such as minimizing the splitting of county boundaries).  (3) We assume that voters rank all candidates, when in practice there may be ``ballot exhaustion.'' We note that our main results show the effectiveness of two- and three-member districts, where voters may be more likely to submit full rankings. Cataloging such constraints and enforcing them computationally is challenging; we conjecture that adding such features from practice is important but will have second-order effects. Also crucial is studying how intra-party effects such as those studied in Appendix \ref{sec:intra} (and in the work of, e.g., \citet{chicago2019} and \citet{lowell2019}) interact with partisan gerrymandering (or gerrymandering on other characteristics).

\paragraph{Conclusion.} {We study the joint gerrymandering and social choice problem, showing the promise of multi-member districts with non-winner-take-all rules in ensuring partisan proportionality, in terms of both enabling independent commissions and constraining partisan gerrymanders. As an application of our approach, we computationally analyze the Fair Representation Act with three- to five-member districts in the United States.}
Finally, while this work is empirical, it raises theoretical questions that we expect to be of interest to the computational social choice community. For example, much is known theoretically about proportionality in the case of one large MMD (see, e.g., \citet{skowronProportionalityDegreeMultiwinner2019} and references therein). Theoretically analyzing proportionality in the case of multiple MMDs, under random and adversarial ways to construct districts, would be of interest. Our results empirically show that under party-based voting, the proportionality gap decreases (non-monotonically) with district size; proving this effect theoretically under more general assumptions, such as those empirically shown in \Cref{sec:newwithoutsolidcoalitions} with a noisy voter ranking model, is of interest.

	\bibliographystyle{ACM-Reference-Format}
	\bibliography{bib,multimember,fairmandering}

@book{balinski2010fair,
	title={Fair representation: meeting the ideal of one man, one vote},
	author={Balinski, Michel L and Young, H Peyton},
	year={2010},
	publisher={Brookings Institution Press}
}

@misc{predictwise2021,
	author       = {{PredictWise}},
	title        = {What all we did for the 2020 Elections},
	howpublished = {\url{https://predictwise.medium.com/what-all-we-did-for-the-2020-elections-78b66a44f651}},
	note         = {Medium blog post, 11 min read},
	month        = mar,
	year         = {2021},
	url          = {https://predictwise.medium.com/what-all-we-did-for-the-2020-elections-78b66a44f651}
}

@book{bowler2000elections,
  author    = {Bowler, Shaun and Grofman, Bernard},
  publisher = {University of Michigan Press},
  title     = {Elections in Australia, Ireland, and Malta under the Single Transferable Vote: Reflections on an embedded institution},
  year      = {2000}
}

@article{azari2012random,
	title={Random utility theory for social choice},
	author={Azari, Hossein and Parks, David and Xia, Lirong},
	journal={Advances in Neural Information Processing Systems},
	volume={25},
	year={2012}
}

@book{enelow1984spatial,
	title={The spatial theory of voting: An introduction},
	author={Enelow, James M and Hinich, Melvin J},
	year={1984},
	publisher={CUP Archive}
}

@article{mcfadden1972conditional,
	title={Conditional logit analysis of qualitative choice behavior},
	author={McFadden, Daniel},
	year={1972}
}

@book{luce1959individual,
	title={Individual choice behavior},
	author={Luce, R Duncan},
	volume={4},
	year={1959}
}

@book{merrill1999unified,
	title={A unified theory of voting: Directional and proximity spatial models},
	author={Merrill, Samuel and Grofman, Bernard},
	year={1999},
	publisher={Cambridge University Press}
}

@article{mosteller1951remarks,
	title={Remarks on the method of paired comparisons: I. The least squares solution assuming equal standard deviations and equal correlations},
	author={Mosteller, Frederick},
	journal={Psychometrika},
	volume={16},
	number={1},
	pages={3--9},
	year={1951},
	publisher={Springer-Verlag}
}

@incollection{thurstone2017law,
	title={A law of comparative judgment},
	author={Thurstone, Louis L},
	booktitle={Scaling},
	pages={81--92},
	year={2017},
	publisher={Routledge}
}

@article{plackett1975analysis,
	title={The analysis of permutations},
	author={Plackett, Robin L},
	journal={Journal of the Royal Statistical Society Series C: Applied Statistics},
	volume={24},
	number={2},
	pages={193--202},
	year={1975},
	publisher={Oxford University Press}
}

@article{ernst1994appointment,
	title={Appointment methods for the House of Representatives and the court challenges},
	author={Ernst, Lawrence R},
	journal={Management Science},
	volume={40},
	number={10},
	pages={1207--1227},
	year={1994},
	publisher={INFORMS}
}

@book{pukelsheim2017proportional,
	title={Proportional representation},
	author={Pukelsheim, Friedrich},
	year={2017},
	publisher={Springer}
}

@incollection{cembrano2021proportional,
	title={Proportional Apportionment: A Case Study From the Chilean Constitutional Convention},
	author={Cembrano, Javier and Correa, Jose and Diaz, Gonzalo and Verdugo, Victor},
	booktitle={Equity and Access in Algorithms, Mechanisms, and Optimization},
	pages={1--9},
	year={2021}
}

@article{singaporegerry,
	title={How Gerrymandering Creates Unfair Elections in Singapore},
	author={Tung, Ngiam Shih},
	journal={New Naratif},
	year={2020},
	url = {https://newnaratif.com/how-gerrymandering-creates-unfair-elections-in-singapore/}
}

@inproceedings{elkind2017multiwinner,
	title={What do multiwinner voting rules do? An experiment over the two-dimensional euclidean domain},
	author={Elkind, Edith and Faliszewski, Piotr and Laslier, Jean-Fran{\c{c}}ois and Skowron, Piotr and Slinko, Arkadii and Talmon, Nimrod},
	booktitle={Proceedings of the AAAI Conference on Artificial Intelligence},
	volume={31},
	number={1},
	year={2017}
}

@article{cembrano2021multidimensional,
  title={Multidimensional apportionment through discrepancy theory},
  author={Cembrano, Javier and Correa, Jos{\'e} and Verdugo, Victor},
  journal={Proceedings of the 22nd ACM Conference on Economics and Computation},
  pages={287--288},
  year={2021}
}

@article{pegden2017partisan,
	title={A partisan districting protocol with provably nonpartisan outcomes},
	author={Pegden, Wesley and Procaccia, Ariel D and Yu, Dingli},
	journal={arXiv preprint arXiv:1710.08781},
	year={2017}
}

@inproceedings{bachrach2016misrepresentation,
	title={Misrepresentation in District Voting.},
	author={Bachrach, Yoram and Lev, Omer and Lewenberg, Yoad and Zick, Yair},
	booktitle={IJCAI},
	pages={81--87},
	year={2016}
}

@article{benade2020abating,
	title={Abating gerrymandering by mandating fairness},
	author={Benade, Gerdus and Procaccia, Ariel},
	journal={Preprint},
	year={2020}
}

@article{tucker2018cut,
	title={A Cut-And-Choose Mechanism to Prevent Gerrymandering},
	author={Tucker-Foltz, Jamie},
	journal={arXiv preprint arXiv:1802.08351},
	year={2018}
}

@article{benade2021you,
	title={You can have your cake and redistrict it too},
	author={Benade, Gerdus and Procaccia, Ariel and Tucker-Foltz, Jamie},
	journal={ACM Transactions on Economics and Computation},
	year={2021},
	publisher={Association for Computing Machinery}
}

@article{tideman1995single,
  title={The single transferable vote},
  author={Tideman, Nicolaus},
  journal={Journal of Economic Perspectives},
  volume={9},
  number={1},
  pages={27--38},
  year={1995}
}

@article{dummett1984voting,
  title={Voting procedures},
  author={Dummett, Michael},
  year={1984}
}

@article{bullock1989symbolics,
	title = {Symbolics or Substance: {{A}} Critique of the at-Large Election Controversy},
	author = {Bullock III, Charles S},
	year = {1989},
	pages = {91--99},
	publisher = {{JSTOR}},
	journal = {State \& Local Government Review}
}

@article{davidsonAtlargeElectionsMinoritygroup1981,
	title = {At-Large Elections and Minority-Group Representation: {{A}} Re-Examination of Historical and Contemporary Evidence},
	shorttitle = {At-Large Elections and Minority-Group Representation},
	author = {Davidson, Chandler and Korbel, George},
	year = {1981},
	month = nov,
	volume = {43},
	pages = {982--1005},
	publisher = {{The University of Chicago Press}},
	issn = {0022-3816},
	doi = {10.2307/2130184},
	url = {https://www.journals.uchicago.edu/doi/abs/10.2307/2130184},
	urldate = {2021-05-05},
	journal = {The Journal of Politics},
	number = {4}
}

@article{stillVoluntaryConstituenciesModified1991,
	title = {Voluntary Constituencies: Modified at-Large Voting as a Remedy for Minority Vote Dilution in Judicial Elections Overview: Electoral Reform},
	shorttitle = {Voluntary Constituencies},
	author = {Still, Edward},
	year = {1991},
	volume = {9},
	pages = {354--369},
	url = {https://heinonline.org/HOL/P?h=hein.journals/yalpr9&i=362},
	urldate = {2021-05-05},
	journal = {Yale Law \& Policy Review},
	language = {eng},
	number = {2}
}

@article{trounstineContextMattersEffects2008,
	title = {The Context Matters: The Effects of Single-Member versus at-Large Districts on City Council Diversity},
	shorttitle = {The Context Matters},
	author = {Trounstine, Jessica and Valdini, Melody E.},
	year = {2008},
	volume = {52},
	pages = {554--569},
	issn = {1540-5907},
	doi = {10.1111/j.1540-5907.2008.00329.x},
	url = {https://onlinelibrary.wiley.com/doi/abs/10.1111/j.1540-5907.2008.00329.x},
	urldate = {2021-05-05},
	copyright = {\textcopyright 2008, Midwest Political Science Association},
	journal = {American Journal of Political Science},
	language = {en},
	number = {3}
}

@article{walawenderAtlargeElectionsVote1985,
	title = {At-Large Elections and Vote Dilution: An Empirical Study Symposium: Environmental Law: Note},
	shorttitle = {At-Large Elections and Vote Dilution},
	author = {Walawender, Richard A.},
	year = {1985},
	volume = {19},
	pages = {1221--1242},
	url = {https://heinonline.org/HOL/P?h=hein.journals/umijlr19&i=1231},
	urldate = {2021-05-05},
	journal = {University of Michigan Journal of Law Reform},
	language = {eng},
	number = {4}
}

@book{arrow2010handbook,
	title = {Handbook of Social Choice and Welfare},
	author = {Arrow, Kenneth J and Sen, Amartya and Suzumura, Kotaro},
	year = {2010},
	volume = {2},
	publisher = {{Elsevier}}
}

@inproceedings{azizComputationalAspectsMultiwinner2015,
	title = {Computational Aspects of Multi-Winner Approval Voting},
	booktitle = {Proceedings of the 2015 {{International Conference}} on {{Autonomous Agents}} and {{Multiagent Systems}}},
	author = {Aziz, Haris and Gaspers, Serge and Gudmundsson, Joachim and Mackenzie, Simon and Mattei, Nicholas and Walsh, Toby},
	year = {2015},
	pages = {107--115},
	publisher = {{International Foundation for Autonomous Agents and Multiagent Systems}},
	address = {{Richland, SC}},
	url = {http://dl.acm.org/citation.cfm?id=2772879.2772896},
	urldate = {2015-10-22},
	isbn = {978-1-4503-3413-6},
	series = {{{AAMAS}} '15}
}

@book{brandt2016handbook,
	title = {Handbook of Computational Social Choice},
	author = {Brandt, Felix and Conitzer, Vincent and Endriss, Ulle and Lang, J{\'e}r{\^o}me and Procaccia, Ariel D},
	year = {2016},
	publisher = {{Cambridge University Press}}
}

@inproceedings{conitzerParadoxesMultipleElections,
title={Paradoxes of Multiple Elections: An Approximation Approach.},
author={Conitzer, Vincent and Xia, Lirong},
booktitle={KR},
year={2012},
organization={Citeseer}
}

@article{elkindPropertiesMultiwinnerVoting2015,
  title={Properties of multiwinner voting rules},
  author={Elkind, Edith and Faliszewski, Piotr and Skowron, Piotr and Slinko, Arkadii},
  journal={Social Choice and Welfare},
  volume={48},
  pages={599--632},
  year={2017},
  publisher={Springer}
}

@inproceedings{garg2019your,
	title = {Who Is in Your Top Three? {{Optimizing}} Learning in Elections with Many Candidates},
	booktitle = {Proceedings of the {{AAAI}} Conference on Human Computation and Crowdsourcing},
	author = {Garg, Nikhil and Gelauff, Lodewijk L and Sakshuwong, Sukolsak and Goel, Ashish},
	year = {2019},
	volume = {7},
	pages = {22--31}
}

@article{adamsModerateVotersWeigh2017,
  title = {Do Moderate Voters Weigh Candidates' Ideologies? {{Voters}}' Decision Rules in the 2010 Congressional Elections},
  shorttitle = {Do Moderate Voters Weigh Candidates' Ideologies?},
  author = {Adams, James and Engstrom, Erik and Joeston, Danielle and Stone, Walt and Rogowski, Jon and Shor, Boris},
  year = {2017},
  month = mar,
  volume = {39},
  pages = {205--227},
  issn = {1573-6687},
  doi = {10.1007/s11109-016-9355-7},
  url = {https://doi.org/10.1007/s11109-016-9355-7},
  urldate = {2021-02-03},
  journal = {Political Behavior},
  language = {en},
  number = {1}
}

@report{FRAanalysisDuchin,
  title = {Modeling the Fair Representation Act},
  author = {{MGGG Redistricting Lab}},
  year ={2022}
}

@article{azizJustifiedRepresentationApprovalbased2017,
  title = {Justified Representation in Approval-Based Committee Voting},
  author = {Aziz, Haris and Brill, Markus and Conitzer, Vincent and Elkind, Edith and Freeman, Rupert and Walsh, Toby},
  year = {2017},
  month = feb,
  volume = {48},
  pages = {461--485},
  issn = {0176-1714, 1432-217X},
  doi = {10.1007/s00355-016-1019-3},
  url = {https://link.springer.com/article/10.1007/s00355-016-1019-3},
  urldate = {2018-02-28},
  journal = {Social Choice and Welfare},
  language = {en},
  number = {2}
}

@article{beckerComputationalRedistrictingVoting,
  title={Computational redistricting and the voting rights act},
  author={Becker, Amariah and Duchin, Moon and Gold, Dara and Hirsch, Sam},
  journal={Election Law Journal: Rules, Politics, and Policy},
  volume={20},
  number={4},
  pages={407--441},
  year={2021},
  publisher={Mary Ann Liebert, Inc., publishers 140 Huguenot Street, 3rd Floor New~…}
}

@techreport{benadeRankedChoiceVoting2021,
  title = {Ranked Choice Voting and Minority Representation},
  author = {Benade, Gerdus and Buck, Ruth and Duchin, Moon and Gold, Dara and Weighill, Thomas},
  year = {2021},
  month = feb,
  address = {{Rochester, NY}},
  institution = {{Social Science Research Network}},
  doi = {10.2139/ssrn.3778021},
  url = {https://papers.ssrn.com/abstract=3778021},
  urldate = {2021-05-05},
  language = {en},
  number = {ID 3778021},
  type = {{{SSRN Scholarly Paper}}}
}

@article{caragiannisSubsetSelectionImplicit2017,
  title = {Subset Selection via Implicit Utilitarian Voting},
  author = {Caragiannis, Ioannis and Nath, Swaprava and Procaccia, Ariel D. and Shah, Nisarg},
  year = {2017},
  month = jan,
  volume = {58},
  pages = {123--152},
  issn = {1076-9757},
  doi = {10.1613/jair.5282},
  url = {https://www.jair.org/index.php/jair/article/view/11040},
  urldate = {2021-02-03},
  copyright = {Copyright (c)},
  journal = {Journal of Artificial Intelligence Research},
  language = {en}
}

@misc{lowell2019,
	title = {Community-Centered Redistricting in Lowell, Massachusetts},
	year = {2019},
  author = {Buck, Ruth and Duchin, Moon and Gold, Dara and Matthews, JN},
	url = {https://mggg.org/uploads/Lowell-Report.pdf},
}

@misc{chicago2019,
	title = {Study of Reform Proposals for Chicago City Council},
	year = {2019},
	author = {Angulu, Hakeem and Buck, Ruth and DeFord, Daryl and Duchin, Moon and Fain, Howard and Hully, Max and Khan, Maira and Schutzman, Zach and York, Oliver},
	url = {https://mggg.org/Chicago.pdf}
}

@inproceedings{chengGroupFairnessCommittee2019,
  title = {Group Fairness in Committee Selection},
  booktitle = {Proceedings of the 2019 {{ACM Conference}} on {{Economics}} and {{Computation}}},
  author = {Cheng, Yu and Jiang, Zhihao and Munagala, Kamesh and Wang, Kangning},
  year = {2019},
  month = jun,
  pages = {263--279},
  publisher = {{Association for Computing Machinery}},
  address = {{New York, NY, USA}},
  doi = {10.1145/3328526.3329577},
  url = {https://doi.org/10.1145/3328526.3329577},
  urldate = {2021-02-03},
  isbn = {978-1-4503-6792-9},
  series = {{{EC}} '19}
}

@article{dewanPoliticalEconomyModels2011,
  title = {Political Economy Models of Elections},
  author = {Dewan, Torun and Shepsle, Kenneth A.},
  year = {2011},
  month = jun,
  volume = {14},
  pages = {311--330},
  issn = {1094-2939, 1545-1577},
  doi = {10.1146/annurev.polisci.12.042507.094704},
  url = {http://www.annualreviews.org/doi/10.1146/annurev.polisci.12.042507.094704},
  urldate = {2021-02-02},
  journal = {Annual Review of Political Science},
  language = {en},
  number = {1}
}

@article{duchinHomologicalPersistenceGerrymandering2020,
  title = {The (Homological) Persistence of Gerrymandering},
  author = {Duchin, Moon and Needham, Tom and Weighill, Thomas},
  year = {2020},
  month = jul,
  url = {http://arxiv.org/abs/2007.02390},
  urldate = {2020-12-16},
  archiveprefix = {arXiv},
  eprint = {2007.02390},
  eprinttype = {arxiv},
  journal = {arXiv:2007.02390 [physics]},
  language = {en},
  primaryclass = {physics}
}

@article{groseReducingLegislativePolarization2020,
  title = {Reducing Legislative Polarization: Top-Two and Open Primaries Are Associated with More Moderate Legislators},
  shorttitle = {Reducing Legislative Polarization},
  author = {Grose, Christian R.},
  year = {2020},
  volume = {1},
  pages = {267--287},
  issn = {2689-4823, 2689-4815},
  doi = {10.1561/113.00000012},
  url = {http://www.nowpublishers.com/article/Details/PIP-0012},
  urldate = {2021-02-24},
  journal = {Journal of Political Institutions and Political Economy},
  language = {en},
  number = {2}
}

@article{jesseeSpatialVoting20042009,
  title = {Spatial Voting in the 2004 Presidential Election},
  author = {Jessee, Stephen A.},
  year = {2009},
  month = feb,
  volume = {103},
  pages = {59--81},
  publisher = {{Cambridge University Press}},
  issn = {1537-5943, 0003-0554},
  doi = {10.1017/S000305540909008X},
  url = {https://www.cambridge.org/core/journals/american-political-science-review/article/abs/spatial-voting-in-the-2004-presidential-election/3CB37721BF009B0E31A38A778B0FF5FA},
  urldate = {2021-02-03},
  journal = {American Political Science Review},
  language = {en},
  number = {1}
}

@article{lacknerQuantitativeAnalysisMultiwinner2018,
  title = {A Quantitative Analysis of Multi-Winner Rules},
  author = {Lackner, Martin and Skowron, Piotr},
  year = {2018},
  month = jan,
  url = {http://arxiv.org/abs/1801.01527},
  archiveprefix = {arXiv},
  eprint = {1801.01527},
  eprinttype = {arxiv},
  journal = {arXiv:1801.01527 [cs]},
  primaryclass = {cs}
}

@article{mcdanielDoesMoreChoice2018,
  title = {Does More Choice Lead to Reduced Racially Polarized Voting? {{Assessing}} the Impact of Ranked-Choice Voting in Mayoral Elections},
  shorttitle = {Does More Choice Lead to Reduced Racially Polarized Voting?},
  author = {McDaniel, Jason},
  year = {2018},
  volume = {10},
  doi = {10.5070/P2cjpp10241252},
  url = {https://escholarship.org/uc/item/2gm5854x},
  urldate = {2021-02-02},
  journal = {California Journal of Politics and Policy},
  language = {en},
  number = {2}
}

@article{mcgheeHasTopTwo2017,
  title = {Has the Top Two Primary Elected More Moderates?},
  author = {McGhee, Eric and Shor, Boris},
  year = {2017},
  month = dec,
  volume = {15},
  pages = {1053--1066},
  publisher = {{Cambridge University Press}},
  issn = {1537-5927, 1541-0986},
  doi = {10.1017/S1537592717002158},
  url = {https://www.cambridge.org/core/journals/perspectives-on-politics/article/has-the-top-two-primary-elected-more-moderates/1F65856812342373F4A51B233E9BD593},
  urldate = {2021-03-03},
  journal = {Perspectives on Politics},
  language = {en},
  number = {4}
}

@article{rogowskiPrimarySystemsCandidate2015,
  title = {Primary Systems and Candidate Ideology: Evidence from Federal and State Legislative Elections},
  shorttitle = {Primary Systems and Candidate Ideology},
  author = {Rogowski, Jon C. and Langella, Stephanie},
  year = {2015},
  month = sep,
  volume = {43},
  pages = {846--871},
  issn = {1532-673X, 1552-3373},
  doi = {10.1177/1532673X14555177},
  url = {http://journals.sagepub.com/doi/10.1177/1532673X14555177},
  urldate = {2021-02-02},
  journal = {American Politics Research},
  language = {en},
  number = {5}
}

@article{shorIdeologyUSCongressional2018,
  title = {Ideology and the {{US}} Congressional Vote},
  author = {Shor, Boris and Rogowski, Jon C.},
  year = {2018},
  month = apr,
  volume = {6},
  pages = {323--341},
  issn = {2049-8470, 2049-8489},
  doi = {10.1017/psrm.2016.23},
  url = {https://www.cambridge.org/core/product/identifier/S2049847016000236/type/journal_article},
  urldate = {2021-02-01},
  journal = {Political Science Research and Methods},
  language = {en},
  number = {2}
}

@inproceedings{skowronProportionalityDegreeMultiwinner2019,
  title={Proportionality degree of multiwinner rules},
  author={Skowron, Piotr},
  booktitle={Proceedings of the 22nd ACM Conference on Economics and Computation},
  pages={820--840},
  year={2021}
}

@inproceedings{faliszewski2019proportional,
  title={Proportional representation in elections: STV vs PAV},
  author={Faliszewski, Piotr and Skowron, Piotr and Szufa, Stanis{\l}aw and Talmon, Nimrod},
  booktitle={Proceedings of the 18th International Conference on Autonomous Agents and MultiAgent Systems},
  pages={1946--1948},
  year={2019}
}

@article{spencerEscapingThicketRanked2015,
  title = {Escaping the Thicket: The Ranked Choice Voting Solution to America's Redistricting Crisis},
  shorttitle = {Escaping the Thicket},
  author = {Spencer, Andrew and Hughes, Christopher and Richie, Rob},
  year = {2015},
  volume = {46},
  pages = {377--424},
  url = {https://heinonline.org/HOL/P?h=hein.journals/cumlr46&i=389},
  urldate = {2021-02-02},
  journal = {Cumberland Law Review},
  language = {eng},
  number = {2}
}

@article{swamyMultiobjectiveOptimizationPolitical,
  title={Multi-objective optimization for political districting: A scalable multilevel approach},
  author={Swamy, Rahul and King, Douglas and Jacobson, Sheldon},
  journal={Optimization online preprint},
  year={2019}
}

@article{tausanovitchDoesIdeologicalProximity2018,
  title = {Does the Ideological Proximity between Candidates and Voters Affect Voting in {{U}}.{{S}}. House Elections?},
  author = {Tausanovitch, Chris and Warshaw, Christopher},
  year = {2018},
  month = mar,
  volume = {40},
  pages = {223--245},
  issn = {1573-6687},
  doi = {10.1007/s11109-017-9437-1},
  url = {https://doi.org/10.1007/s11109-017-9437-1},
  urldate = {2021-02-03},
  journal = {Political Behavior},
  language = {en},
  number = {1}
}

@article{issacharoff2002gerrymandering,
  title={Gerrymandering and political cartels},
  author={Issacharoff, Samuel},
  journal={Harv. L. Rev.},
  volume={116},
  pages={593},
  year={2002},
  publisher={HeinOnline}
}

@article{wang2016three,
  title={Three tests for practical evaluation of partisan gerrymandering},
  author={Wang, Samuel S-H},
  journal={Stan. L. Rev.},
  volume={68},
  pages={1263},
  year={2016},
  publisher={HeinOnline}
}

@article{gurnee2021fairmandering,
  title={Fairmandering: A column generation heuristic for fairness-optimized political districting},
  author={Gurnee, Wes and Shmoys, David B},
  journal={Proceedings of the 2021 SIAM Conference on Applied and Computational Discrete Algorithms (ACDA21)},
  year={2021},
  doi = {10.1137/1.9781611976830.9},
  URL = {https://epubs.siam.org/doi/abs/10.1137/1.9781611976830.9},
}

@article{fifield2020automated,
  title={Automated redistricting simulation using Markov chain Monte Carlo},
  author={Fifield, Benjamin and Higgins, , Michael and Imai, Kosuke and Tarr, Alexander},
  journal={Journal of Computational and Graphical Statistics},
  volume={29},
  number={4},
  pages={715--728},
  year={2020},
  publisher={Taylor \& Francis}
}

@article{herschlag2020quantifying,
  title={Quantifying gerrymandering in north carolina},
  author={Herschlag, Gregory and Kang, Han Sung and Luo, Justin and Graves, Christy Vaughn and Bangia, Sachet and Ravier, Robert and Mattingly, Jonathan C},
  journal={Statistics and Public Policy},
  volume={7},
  number={1},
  pages={30--38},
  year={2020},
  publisher={Taylor \& Francis}
}

@article{cannon2020voting,
  title={Voting rights, Markov chains, and optimization by short bursts},
  author={Cannon, Sarah and Goldbloom-Helzner, Ari and Gupta, Varun and Matthews, JN and Suwal, Bhushan},
  journal={Methodology and Computing in Applied Probability},
  volume={25},
  number={1},
  pages={36},
  year={2023},
  publisher={Springer}
}

@article{autry2020multi,
  title={Multi-scale merge-split Markov chain Monte Carlo for redistricting},
  author={Autry, Eric A and Carter, Daniel and Herschlag, Gregory and Hunter, Zach and Mattingly, Jonathan C},
  journal={arXiv preprint arXiv:2008.08054},
  year={2020}
}

@article{mccartan2020sequential,
  title={Sequential Monte Carlo for sampling balanced and compact redistricting plans},
  author={McCartan, Cory and Imai, Kosuke},
  journal={arXiv preprint arXiv:2008.06131},
  year={2020}
}

@article{DeFord2021Recombination,
journal = {Harvard Data Science Review},
doi = {10.1162/99608f92.eb30390f},
note = {https://hdsr.mitpress.mit.edu/pub/1ds8ptxu},
publisher = {},
title = {Recombination: A Family of Markov Chains for Redistricting},
url = {https://hdsr.mitpress.mit.edu/pub/1ds8ptxu},
author = {DeFord, Daryl and Duchin, Moon and Solomon, Justin},
date = {2021-03-31},
year = {2021},
month = {3},
day = {31},
}

@article{autry2021metropolized,
  title={Metropolized Multiscale Forest Recombination for Redistricting},
  author={Autry, Eric A and Carter, Daniel and Herschlag, Gregory J and Hunter, Zach and Mattingly, Jonathan C},
  journal={Multiscale Modeling \& Simulation},
  volume={19},
  number={4},
  pages={1885--1914},
  year={2021},
  publisher={SIAM}
}

@article{chen2013unintentional,
  title={Unintentional gerrymandering: Political geography and electoral bias in legislatures},
  author={Chen, Jowei and Rodden, Jonathan},
  journal={Quarterly Journal of Political Science},
  volume={8},
  number={3},
  pages={239--269},
  year={2013},
  publisher={Citeseer}
}

@article{stephanopoulos2015partisan,
  title={Partisan gerrymandering and the efficiency gap},
  author={Stephanopoulos, Nicholas O and McGhee, Eric M},
  journal={U. Chi. L. Rev.},
  volume={82},
  pages={831},
  year={2015},
  publisher={HeinOnline}
}

@data{VOQCHQ_2018,
author = {{MIT Election Data and Science Lab}},
publisher = {Harvard Dataverse},
title = {County Presidential Election Returns 2000-2016},
UNF = {UNF:6:ZZe1xuZ5H2l4NUiSRcRf8Q==},
year = {2018},
version = {V6},
doi = {10.7910/DVN/VOQCHQ},
url = {https://doi.org/10.7910/DVN/VOQCHQ}
}

@article{chikina2017assessing,
  title={Assessing significance in a {M}arkov chain without mixing},
  author={Chikina, Maria and Frieze, Alan and Pegden, Wesley},
  journal={Proceedings of the National Academy of Sciences},
  volume={114},
  number={11},
  pages={2860--2864},
  year={2017},
  publisher={National Acad Sciences}
}

@article{puppe2008computational,
  title={A computational approach to unbiased districting},
  author={Puppe, Clemens and Tasn{\'a}di, Attila},
  journal={Mathematical and Computer Modelling},
  volume={48},
  number={9-10},
  pages={1455--1460},
  year={2008},
  publisher={Elsevier}
}

@article{deford2020computational,
  title={A Computational Approach to Measuring Vote Elasticity and Competitiveness},
  author={DeFord, Daryl and Duchin, Moon and Solomon, Justin},
  journal={Statistics and Public Policy},
  number={just-accepted},
  pages={1--30},
  year={2020},
  publisher={Taylor \& Francis}
}

@article{erikson1972malapportionment,
  title={Malapportionment, gerrymandering, and party fortunes in congressional elections},
  author={Erikson, Robert S},
  journal={The American Political Science Review},
  pages={1234--1245},
  year={1972},
  publisher={JSTOR}
}

@article{kueng2019fair,
  title={Fair redistricting is hard},
  author={Kueng, Richard and Mixon, Dustin G and Villar, Soledad},
  journal={Theoretical Computer Science},
  volume={791},
  pages={28--35},
  year={2019},
  publisher={Elsevier}
}

@article{liu2016pear,
  title={PEAR: a massively parallel evolutionary computation approach for political redistricting optimization and analysis},
  author={Liu, Yan Y and Cho, Wendy K Tam and Wang, Shaowen},
  journal={Swarm and Evolutionary Computation},
  volume={30},
  pages={78--92},
  year={2016},
  publisher={Elsevier}
}

@article{warrington2018quantifying,
  title={Quantifying gerrymandering using the vote distribution},
  author={Warrington, Gregory S},
  journal={Election Law Journal},
  volume={17},
  number={1},
  pages={39--57},
  year={2018},
  publisher={Mary Ann Liebert, Inc. 140 Huguenot Street, 3rd Floor New Rochelle, NY 10801 USA}
}

@book{mcgann2016gerrymandering,
  title={Gerrymandering in America: The House of Representatives, the Supreme Court, and the future of popular sovereignty},
  author={McGann, Anthony J and Smith, Charles Anthony and Latner, Michael and Keena, Alex},
  year={2016},
  publisher={Cambridge University Press}
}

@article{herschlag2017evaluating,
  title={Evaluating partisan gerrymandering in {W}isconsin},
  author={Herschlag, Gregory and Ravier, Robert and Mattingly, Jonathan C},
  journal={arXiv preprint arXiv:1709.01596},
  year={2017}
}

@misc{duchin2018outlier,
  title={Outlier analysis for {P}ennsylvania congressional redistricting},
  author={Duchin, Moon},
  year={2018},
  publisher={Technical report Tufts University}
}

@article{king2015efficient,
  title={Efficient geo-graph contiguity and hole algorithms for geographic zoning and dynamic plane graph partitioning},
  author={King, Douglas M and Jacobson, Sheldon H and Sewell, Edward C},
  journal={Mathematical Programming},
  volume={149},
  number={1-2},
  pages={425--457},
  year={2015},
  publisher={Springer}
}

@article{chatterjee2019partisan,
  title={On partisan bias in redistricting: computational complexity meets the science of gerrymandering},
  author={Chatterjee, Tanima and DasGupta, Bhaskar},
  journal={arXiv preprint arXiv:1910.01565},
  year={2019}
}

@article{duchin2019locating,
  title={Locating the representational baseline: Republicans in Massachusetts},
  author={Duchin, Moon and Gladkova, Taissa and Henninger-Voss, Eugene and Klingensmith, Ben and Newman, Heather and Wheelen, Hannah},
  journal={Election Law Journal: Rules, Politics, and Policy},
  volume={18},
  number={4},
  pages={388--401},
  year={2019},
  publisher={Mary Ann Liebert, Inc., publishers 140 Huguenot Street, 3rd Floor New~…}
}

@data{DVN/NH5S2I_2018,
author = {{Voting and Election Science Team}},
publisher = {Harvard Dataverse},
title = {{2016 Precinct-Level Election Results}},
year = {2018},
version = {V46},
doi = {10.7910/DVN/NH5S2I},
url = {https://doi.org/10.7910/DVN/NH5S2I}
}

@data{DVN/UBKYRU_2019,
author = {{Voting and Election Science Team}},
publisher = {Harvard Dataverse},
title = {{2018 Precinct-Level Election Results}},
year = {2019},
version = {V26},
doi = {10.7910/DVN/UBKYRU},
url = {https://doi.org/10.7910/DVN/UBKYRU}
}

@misc{mgggstates,
  author = {{MGGG}},
  title = {{MGGG States}},
  year = {2020},
  publisher = {GitHub},
  journal = {GitHub repository},
  howpublished = {\url{https://github.com/mggg-states}},
}

@data{DVN/WYXFW3_2011,
author = {Stephen Ansolabehere and Jonathan Rodden},
publisher = {Harvard Dataverse},
title = {{Indiana Data Files}},
UNF = {UNF:5:JRV8NRNO94zv+qfmGQCEjQ==},
year = {2011},
version = {V2},
doi = {10.7910/DVN/WYXFW3},
url = {https://doi.org/10.7910/DVN/WYXFW3}
}

@data{DVN/AN00LH_2011,
author = {Stephen Ansolabehere and Jonathan Rodden},
publisher = {Harvard Dataverse},
title = {{Mississippi Data Files}},
UNF = {UNF:5:uGLLGvmX46cEpLqF1Yru+w==},
year = {2011},
version = {V1},
doi = {10.7910/DVN/AN00LH},
url = {https://doi.org/10.7910/DVN/AN00LH}
}

@data{DVN/KX0YGR_2011,
author = {Stephen Ansolabehere and Jonathan Rodden},
publisher = {Harvard Dataverse},
title = {{New Jersey Data Files}},
UNF = {UNF:5:5ZRlyDduZ6tLQyKoukAikg==},
year = {2011},
version = {V1},
doi = {10.7910/DVN/KX0YGR},
url = {https://doi.org/10.7910/DVN/KX0YGR}
}

@data{DVN/AWE39N_2011,
author = {Stephen Ansolabehere and Jonathan Rodden},
publisher = {Harvard Dataverse},
title = {{New York Data Files}},
UNF = {UNF:5:rTS7l9RupCXZ1TH/+XgN/g==},
year = {2011},
version = {V1},
doi = {10.7910/DVN/AWE39N},
url = {https://doi.org/10.7910/DVN/AWE39N}
}

@data{DVN/UUCWPP_2011,
author = {Stephen Ansolabehere and Jonathan Rodden},
publisher = {Harvard Dataverse},
title = {{Alabama Data Files}},
UNF = {UNF:5:5vSmmKL5MOb63cHhbIivBw==},
year = {2011},
version = {V1},
doi = {10.7910/DVN/UUCWPP},
url = {https://doi.org/10.7910/DVN/UUCWPP}
}

@article{lempert2021ranked,
	title={Ranked-Choice Voting as Reprieve from the Court-Ordered Map},
	author={Lempert, Benjamin P},
	journal={Michigan Law Review},
	volume={119},
	number={8},
	pages={1785--1818},
	year={2021}
}

@techreport{aclu_redistricting_2001,
  title={EVERYTHING YOU
ALWAYS WANTED
TO KNOW ABOUT
REDISTRICTING},
  author={{American Civil Liberties Union}},
  institution={American Civil Liberties Union},
  year={2001},
  url={https://www.aclu.org/sites/default/files/FilesPDFs/redistricting_manual.pdf}
}

@misc{electoralreformsociety_stv,
  author       = {{Electoral Reform Society}},
  howpublished = {\url{https://electoral-reform.org.uk/voting-systems/types-of-voting-system/single-transferable-vote/}},
  note         = {Accessed: 2026-04-19},
  title        = {{Single Transferable Vote}},
  year         = {2026}
}

@inproceedings{borodin2018big,
	author = {Borodin, Allan and Lev, Omer and Shah, Nisarg and Strangway, Tyrone},
	title = {Little House (Seat) on the Prairie: Compactness, Gerrymandering, and Population Distribution},
	year = {2022},
	isbn = {9781450392136},
	publisher = {International Foundation for Autonomous Agents and Multiagent Systems},
	address = {Richland, SC},
	booktitle = {Proceedings of the 21st International Conference on Autonomous Agents and Multiagent Systems},
	pages = {154–162},
	numpages = {9},
	location = {Virtual Event, New Zealand},
	series = {AAMAS '22}
}

@article{carey2011electoral,
  title={The electoral sweet spot: Low-magnitude proportional electoral systems},
  author={Carey, John M and Hix, Simon},
  journal={American Journal of Political Science},
  volume={55},
  number={2},
  pages={383--397},
  year={2011},
  publisher={Wiley Online Library}
}

@article{altman2015redistricting,
  title={Redistricting and polarization},
  author={Altman, Micah and McDonald, Michael},
  journal={American Gridlock: The Sources, Character, and Impact of Political Polarization},
  pages={103--131},
  year={2015}
}

@article{cain2011redistricting,
  title={Redistricting commissions: A better political buffer},
  author={Cain, Bruce E},
  journal={Yale LJ},
  volume={121},
  pages={1808},
  year={2011},
  publisher={HeinOnline}
}

@misc{spencer_who_2020,
    author = {Doug Spencer},
    title = {Who Draws the Lines?},
    year = {2020},
    url = {https://redistricting.lls.edu/national-overview/?colorby=Institution&level=Congress&cycle=2020},
    note = {All about redistricting}
}

@article{van2015network,
  title={Network-based vertex dissolution},
  author={Van Bevern, Rene and Bredereck, Robert and Chen, Jiehua and Froese, Vincent and Niedermeier, Rolf and Woeginger, Gerhard J},
  journal={SIAM Journal on Discrete Mathematics},
  volume={29},
  number={2},
  pages={888--914},
  year={2015},
  publisher={SIAM}
}

@inproceedings{lewenberg2017divide,
  title={Divide and conquer: Using geographic manipulation to win district-based elections},
  author={Lewenberg, Yoad and Lev, Omer and Rosenschein, Jeffrey S},
  booktitle={Proceedings of the 16th Conference on Autonomous Agents and MultiAgent Systems},
  pages={624--632},
  year={2017}
}

@inproceedings{bredereck2021strategic,
  title={Strategic campaign management in apportionment elections},
  author={Bredereck, Robert and Faliszewski, Piotr and Furdyna, Micha{\l} and Kaczmarczyk, Andrzej and Lackner, Martin},
  booktitle={Proceedings of the Twenty-Ninth International Conference on International Joint Conferences on Artificial Intelligence},
  pages={103--109},
  year={2021}
}

@article{guney2018efficient,
  title={Efficient election campaign optimization using integer programming},
  author={G{\"u}ney, Evren},
  journal={Journal of Industrial Engineering and Management-JIEM},
  year={2018},
  publisher={Omniascience}
}

@article{dyer1985complexity,
  title={On the complexity of partitioning graphs into connected subgraphs},
  author={Dyer, Martin E and Frieze, Alan M},
  journal={Discrete Applied Mathematics},
  volume={10},
  number={2},
  pages={139--153},
  year={1985},
  publisher={Elsevier}
}

@article{ainsworth2022district,
  title={District competitiveness increases voter turnout: evidence from repeated redistricting in North Carolina},
  author={Ainsworth, Robert and Munoz, Emanuel Garcia and Gomez, Andres Munoz},
  journal={University of Florida},
  year={2022}
}

@article{cancela2016explaining,
  title={Explaining voter turnout: A meta-analysis of national and subnational elections},
  author={Cancela, Jo{\~a}o and Geys, Benny},
  journal={Electoral studies},
  volume={42},
  pages={264--275},
  year={2016},
  publisher={Elsevier}
}

@article{moskowitz2019reevaluating,
  title={Reevaluating competition and turnout in US house elections},
  author={Moskowitz, Daniel J and Schneer, Benjamin and others},
  journal={Quarterly Journal of Political Science},
  volume={14},
  number={2},
  pages={191--223},
  year={2019},
  publisher={Now Publishers, Inc.}
}

@article{gerber2020one,
  title={One in a million: Field experiments on perceived closeness of the election and voter turnout},
  author={Gerber, Alan and Hoffman, Mitchell and Morgan, John and Raymond, Collin},
  journal={American Economic Journal: Applied Economics},
  volume={12},
  number={3},
  pages={287--325},
  year={2020},
  publisher={American Economic Association 2014 Broadway, Suite 305, Nashville, TN 37203-2425}
}

@article{deford2019redistricting,
  title={Redistricting reform in Virginia: Districting criteria in context},
  author={DeFord, Daryl and Duchin, Moon},
  journal={Virginia Policy Review},
  volume={12},
  number={2},
  pages={120--146},
  year={2019}
}

@article{banzhafMultimemberElectoralDistricts1966,
  title = {Multi-Member Electoral Districts. Do They Violate the "{{One}} Man, One Vote" Principle},
  author = {Banzhaf, John F.},
  year = {1966},
  month = jul,
  volume = {75},
  pages = {1309},
  issn = {00440094},
  doi = {10.2307/795047},
  url = {https://www.jstor.org/stable/795047?origin=crossref},
  urldate = {2021-06-29},
  journal = {The Yale Law Journal},
  language = {en},
  number = {8}
}

@article{horwill1925proportional,
  title={Proportional Representation: Its Dangers and Defects},
  author={Horwill, George},
  year={1925}
}

@article{taagepera1989seats,
  title={Seats and votes: The effects and determinants of electoral systems},
  author={Taagepera, Rein and Shugart, Matthew Soberg},
  year={1989}
}

@incollection{lijphart2000patterns,
  title = {Patterns of Democracy: Government Forms and Performance in Thirty-Six Countries},
  author = {Arend Lijphart},
  booktitle = {Patterns of Democracy},
  chapter = {8},
  publisher = {Yale University Press},
  year = {1999},
  address = {New Haven},
}

@misc{fra,
  author =       "{Fair Representation Act}",
  year =         "2024",
  month 	= "March",
  title =        "Fair representation act",
  lastaccessed = "May 7, 2024",
  url =          "https://www.congress.gov/bill/118th-congress/house-bill/7740",
}

@article{dauerMultimemberDistrictsDade1966,
  title = {Multi-Member Districts in Dade County: Study of a Problem and a Delegation},
  shorttitle = {Multi-Member Districts in Dade County},
  author = {Dauer, Manning J.},
  year = {1966},
  month = aug,
  volume = {28},
  pages = {617--638},
  issn = {0022-3816, 1468-2508},
  doi = {10.2307/2128159},
  url = {https://www.journals.uchicago.edu/doi/10.2307/2128159},
  urldate = {2021-06-29},
  journal = {The Journal of Politics},
  language = {en},
  number = {3}
}

@article{derfnerMultimemberDistrictsBlack72,
  title = {Multi-Member Districts and Black Voters},
  author = {Derfner, Armand},
  volume={2},
  pages={120},
  year={1972},
    journal = {The Black Law Journal},
  language = {en}
}

@article{gerberMinorityRepresentationMultimember1998,
  title = {Minority Representation in Multimember Districts},
  author = {Gerber, Elisabeth R. and Morton, Rebecca B. and Rietz, Thomas A.},
  year = {1998},
  month = mar,
  volume = {92},
  pages = {127--144},
  issn = {0003-0554, 1537-5943},
  doi = {10.2307/2585933},
  url = {https://www.cambridge.org/core/product/identifier/S0003055400209167/type/journal_article},
  urldate = {2021-06-29},
  journal = {American Political Science Review},
  language = {en},
  number = {1}
}

@article{DuchinSchoenbach2022371393,
url = {https://doi.org/10.1515/for-2022-2064},
title = {Redistricting for Proportionality},
title = {},
author = {Moon Duchin and Gabe Schoenbach},
pages = {371--393},
volume = {20},
number = {3-4},
journal = {The Forum},
doi = {doi:10.1515/for-2022-2064},
year = {2022},
}

@misc{fishburn2018approval,
  title={Approval Voting Approach to Subset Selection},
  author={Fishburn, Peter C and Peke{\v{c}}, Aleksandar},
  year={2018}
}

@article{hamiltonLegislativeConstituenciesSingemember1967,
  title = {Legislative Constituencies: Single-Member Districts, Multi-Member Districts, and Floterial Districts},
  shorttitle = {Legislative Constituencies},
  author = {Hamilton, Howard D.},
  year = {1967},
  volume = {20},
  pages = {321--340},
  publisher = {{[University of Utah, Sage Publications, Inc., Western Political Science Association]}},
  issn = {0043-4078},
  doi = {10.2307/445408},
  url = {https://www.jstor.org/stable/445408},
  urldate = {2021-06-29},
  journal = {The Western Political Quarterly},
  number = {2}
}

@article{bullock1993changing,
  title={Changing from multimember to single-member districts: partisan, racial, and gender consequences},
  author={Bullock III, Charles S and Gaddie, Ronald Keith},
  journal={State \& Local Government Review},
  pages={155--163},
  year={1993},
  publisher={JSTOR}
}

@article{klain1955new,
  title = {A New Look at the Constituencies: The Need for a Recount and a Reappraisal},
  author = {Klain, Maurice},
  year = {1955},
  volume = {49},
  pages = {1105--1119},
  publisher = {{Cambridge University Press}},
  journal = {American Political Science Review},
  number = {4}
}

@article{niemiCandidaciesCompetitivenessMultimember1991,
  title = {Candidacies and Competitiveness in Multimember Districts},
  author = {Niemi, Richard G. and Jackman, Simon and Winsky, Laura R.},
  year = {1991},
  volume = {16},
  pages = {91--109},
  publisher = {{[Wiley, Comparative Legislative Research Center]}},
  issn = {0362-9805},
  doi = {10.2307/439969},
  url = {https://www.jstor.org/stable/439969},
  urldate = {2021-06-29},
  journal = {Legislative Studies Quarterly},
  number = {1}
}

@misc{ballotpedia_2021, title={Arizona {S}tate {L}egislature}, url={https://ballotpedia.org/Arizona_State_Legislature}, journal={Ballotpedia}, author={Ballotpedia}, year={2021},   urldate = {2021-06-29}}

@misc{ballotpedia_2021stateleg, title={State legislative chambers that use multi-member districts}, url={https://ballotpedia.org/State_legislative_chambers_that_use_multi-member_districts}, journal={Ballotpedia}, author={Ballotpedia}, year={2021},   urldate = {2021-06-29}}

@misc{fairvoteorgrcv_2021, title={Details about Ranked Choice Voting}, url={https://www.fairvote.org/rcv}, journal={FairVote}, author={{FairVote}}, year={2021},   urldate = {2021-06-29}}

@report{reportamericansciences_2020, place={Cambridge, MA}, series={The Practice of Democratic Citizenship}, title={Our Common Purpose: Reinventing American Democracy for the 21st Century}, url={https://www.amacad.org/sites/default/files/publication/downloads/2020-Democratic-Citizenship_Our-Common-Purpose_0.pdf}, institution={American Academy of Arts and Sciences}, author={{American Academy of Arts and Sciences}}, year={2020}, collection={The Practice of Democratic Citizenship},   urldate = {2021-06-29}}

@article{niemiImpactMultimemberDistricts1985,
  title = {The Impact of Multimember Districts on Party Representation in {{U}}. {{S}}. State Legislatures},
  author = {Niemi, Richard G. and Hill, Jeffrey S. and Grofman, Bernard},
  year = {1985},
  volume = {10},
  pages = {441--455},
  publisher = {{[Wiley, Comparative Legislative Research Center]}},
  issn = {0362-9805},
  doi = {10.2307/440068},
  url = {https://www.jstor.org/stable/440068},
  urldate = {2021-06-29},
  journal = {Legislative Studies Quarterly},
  number = {4}
}

@misc{page_gilens_2018,
  title = {Making American Democracy Representative},
  author = {Page, Benjamin I. and Gilens, Martin},
  year = {2018},
  month = oct,
  url = {https://prospect.org/power/making-american-democracy-representative/}
}

@article{silva1964compared,
  title={Compared Values of the Single-and the Multi-Member Legislative District},
  author={Silva, Ruth C},
  journal={Western Political Quarterly},
  volume={17},
  number={3},
  pages={504--516},
  year={1964},
  publisher={Sage Publications Sage CA: Thousand Oaks, CA}
}

\newpage
\FloatBarrier
\appendix
\crefalias{section}{appsec}
\crefalias{subsection}{appsec}
\crefalias{subsubsection}{appsec}

\section{Supplementary results}
\FloatBarrier

\cref{fig:boxall} shows the complete state-by-state results for STV-based elections where the high-level patterns from \cref{fig:propdifferentmethods} can be seen in more granular detail. For example, the bump structure at uniform district sizes is clear for Michigan (MI), Florida (FL), Wisconsin (WI), California (CA), North Carolina (NC), and Pennsylvania (PA), among others. In CA, MA, NY, OK, and TN single-member proportional plans are not even possible due to the diffusion of minority party voters. In all states, the maximum gerrymandering capability peaks at about 1.5 seats per district and starts to decay rapidly with larger districts.

In \cref{fig:advantage_by_rule} we see that within our ensemble of maps, under full control of the redistricting process, Democrats could actually more effectively gerrymander the House of Representatives than Republicans, a finding counter to the conventional wisdom that geography favors Republican gerrymandering \citep{chen2013unintentional,borodin2018big}. One difference may be that our analysis is still using only ``natural'' districts -- those which aren't wildly contorted -- and therefore does not include the most surgical gerrymanders. This is a standard challenge in redistricting algorithms, which are effective for large-scale analysis but can often be outperformed (with respect to any metric) in any specific setting by experts. Regardless, as shown in \cref{fig:boxall}, it is true that Democrats can gerrymander their large states (CA, NY, and IL) more effectively than Republicans can gerrymander theirs (TX and FL), and that when ignoring the VRA, Democrats can crack large cities into many Democratic-leaning wedge-shaped districts. Furthermore, the gap in advantage is largest at two- and three-member districts because these thresholds enable Democrats to more efficiently place their voters than Republicans. For STV in particular, Democrats can break the $\frac23$ or $\frac34$ threshold needed for a sweep in many urban districts, while also still clearing the $\frac13$ or $\frac14$ threshold needed to gain one seat in rural districts, resulting in fewer wasted votes across most types of districts. While critics might point to this as a deal-breaker, it is important to recognize that this is only true in the limit of Democratic control, and such a scenario is deeply unrealistic.

In \cref{fig:propmethods_median}, we show how the median, max Republican, and max Democratic absolute proportionality gap changes as a function of the voting rule and ratio of districts to seats. The median gap drops by about a factor of three between single-member and two-member districts and continues to slowly decay with larger districts. The median gap is a relevant metric because it proxies how easy it is to create a proportional map. Similarly, if the median map is fair, then there exist many fair maps, and it becomes easier to optimize for other desirable criteria like proportional racial representation, maintaining political subdivisions, and compactness. Also in \cref{fig:propmethods_median}, we see that STV and Thiele squared track each other very closely, with the exception of Democratic gerrymanders of two- and three-member districts. This follows similar logic as the overall Democratic advantage analysis above, except that because Thiele squared requires more votes for a sweep, Democrats can no longer rely on full control of districts with just above $\frac23$ of votes, and so end up wasting many votes in more heavily democratic areas by just missing the sweep threshold.

\begin{figure*}
				\centering
		\includegraphics[width=.76\linewidth]{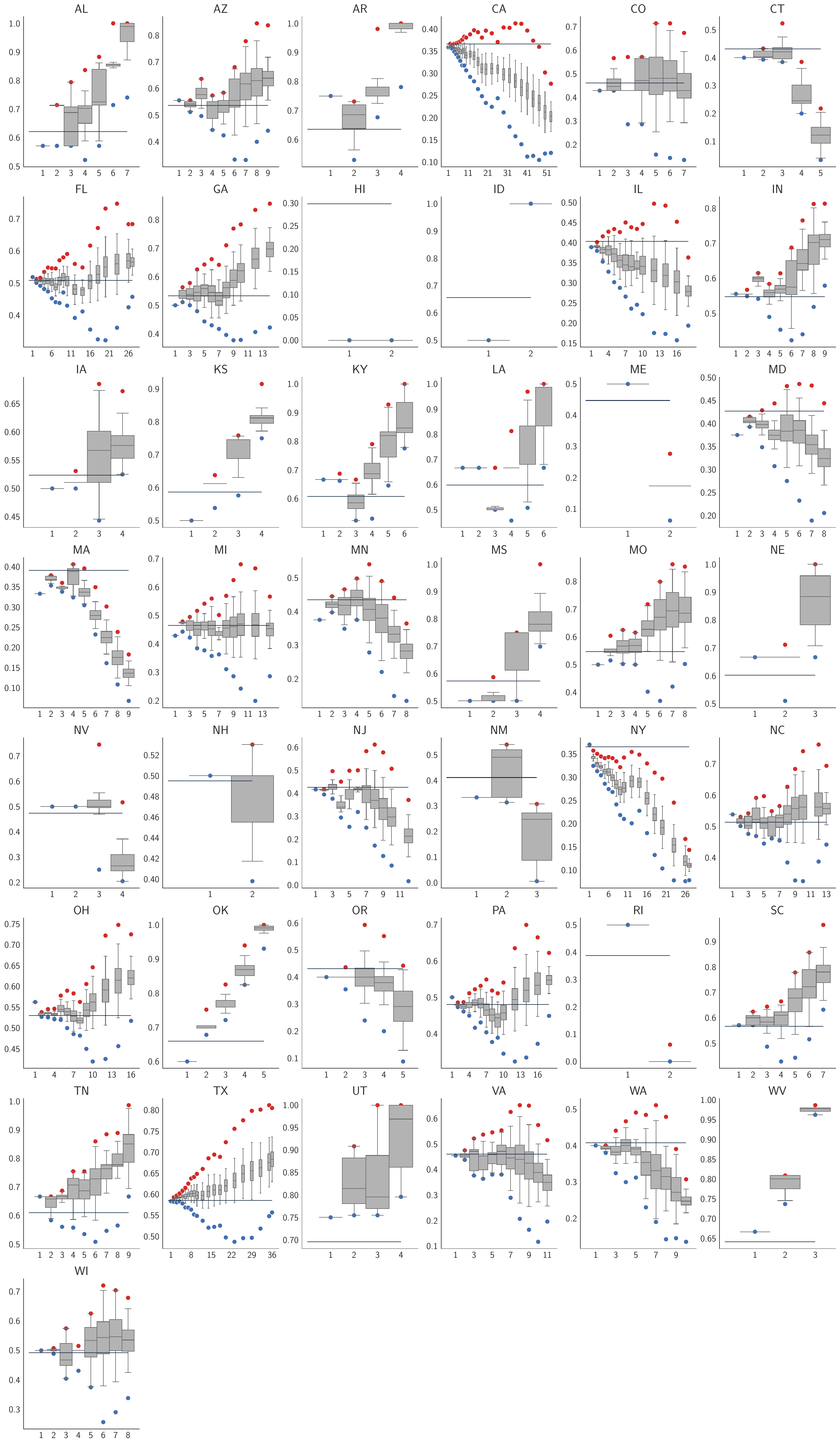}
		\caption{\Cref{fig:propbystatebox}, repeated for each state with at least 2 Representatives. The y-axis is Republican seat share.}
		\label{fig:boxall}
\end{figure*}

\begin{figure}
    \centering
    \includegraphics[width=0.97\textwidth]{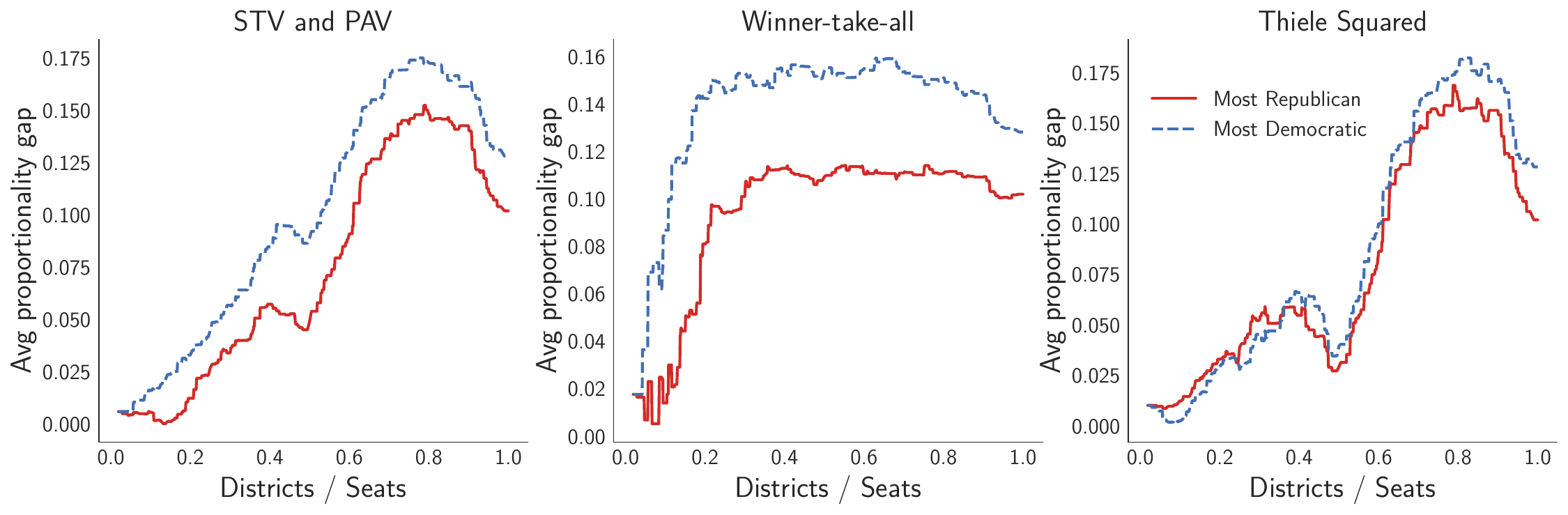}
    \caption{Proportionality gap in favor of each party in  the most advantageous map}
    \label{fig:advantage_by_rule}
\end{figure}

\begin{figure*}
	\centering
	\includegraphics[width=1\linewidth]{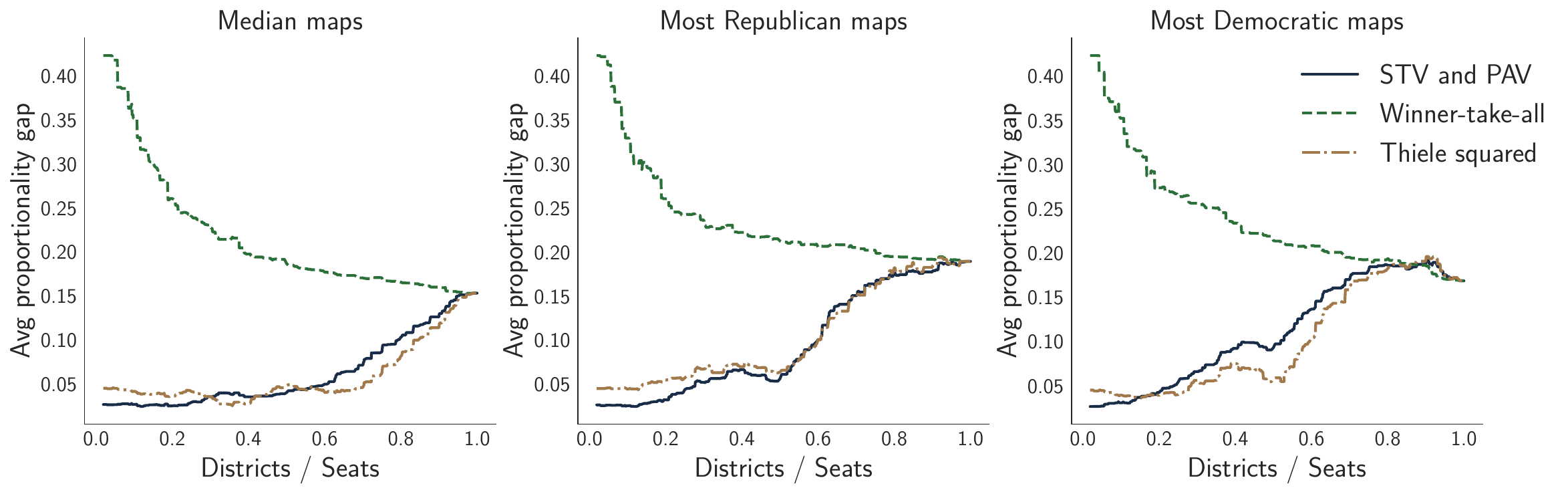}
	\caption{\Cref{fig:propmethods_mostfair}, repeated for Median, Most Republican, and Most Democratic maps.}
	\label{fig:propmethods_median}
\end{figure*}

\begin{figure}
    \centering
    \includegraphics[width=\linewidth]{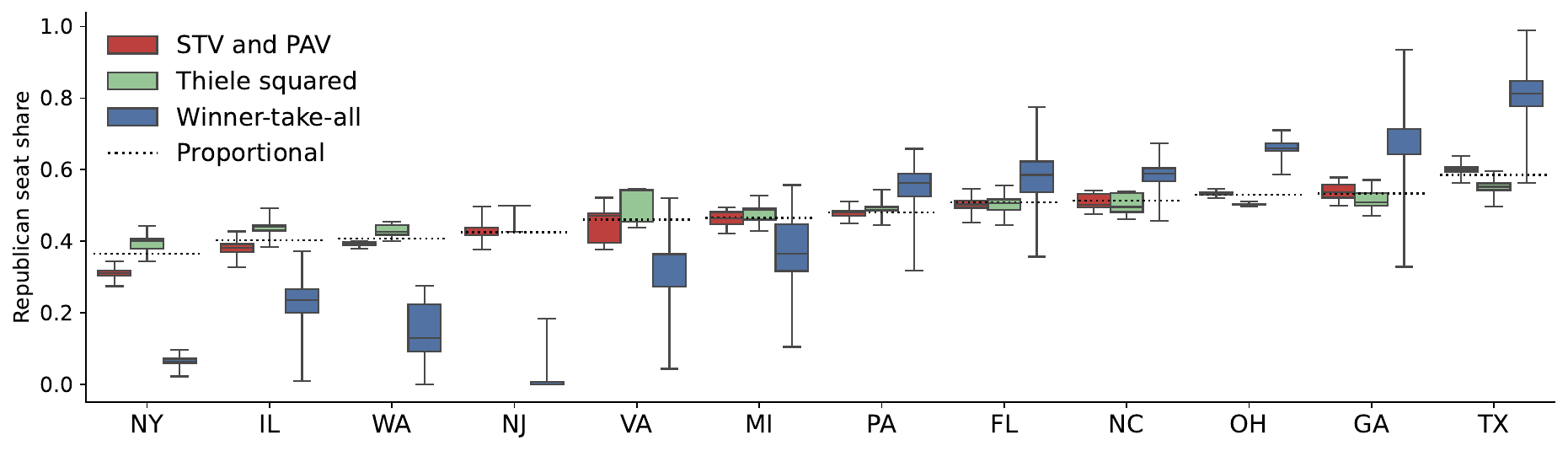}
    \caption{Seat share distribution of Fair Representation Act-compliant ensembles by voting rule.}
    \label{fig:hr4000_by_rule}
\end{figure}

\begin{figure}
    \centering
    \includegraphics[width=\linewidth]{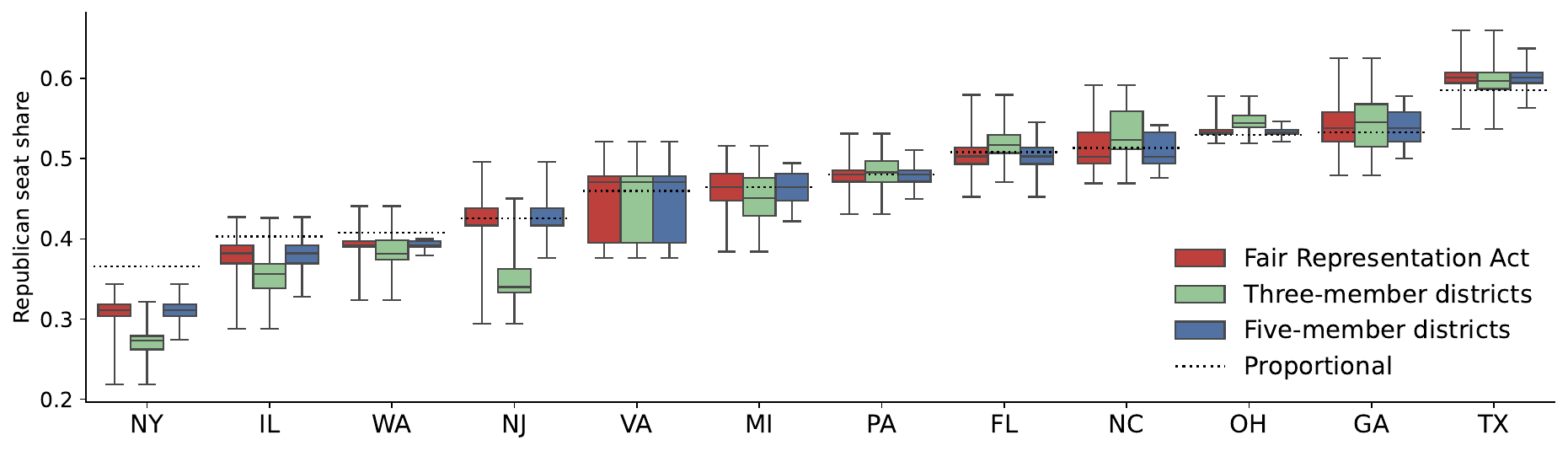}
    \caption{Comparison of Republican seat share distribution of Fair Representation Act-compliant ensembles with baseline for a selection of states using the STV/PAV voting rule (see Appendix \cref{fig:hr4000_by_rule} for more rules). The Fair Representation Act has $N_k \in \{3, 4, 5\}$. The baseline is to use only three-member and only five-member districts (with four-member districts for overflow as necessary).}
    \label{fig:hr4000_size_baseline}
\end{figure}

\begin{figure}
    \centering
    \includegraphics[width=\linewidth]{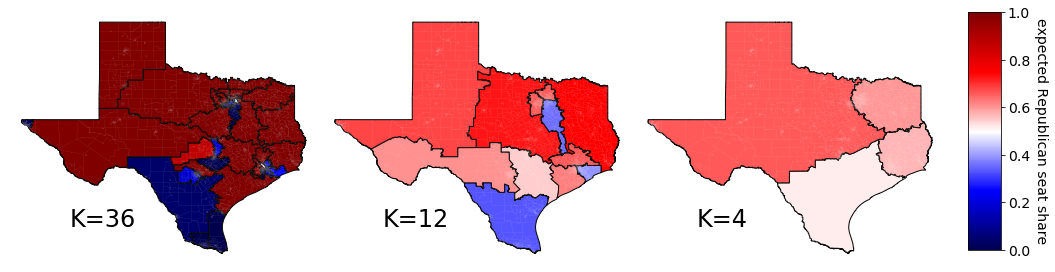}
    \caption{Example fair multi-member plans for Texas. Color denotes expected Republican seat share.}
    \label{fig:tx_plans}
\end{figure}

\begin{figure}[tbhp]
    \centering
    \includegraphics[width=\linewidth]{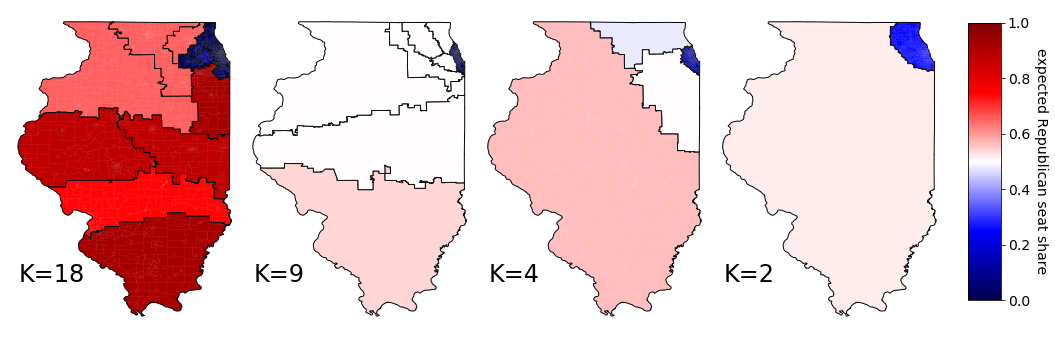}
    \caption{Example fair multi-member plans for Illinois. Color denotes expected Republican seat share.}
    \label{fig:il_plans}
\end{figure}

\begin{figure}[tbhp]
	\centering
		\centering
		\includegraphics[width=.5\textwidth]{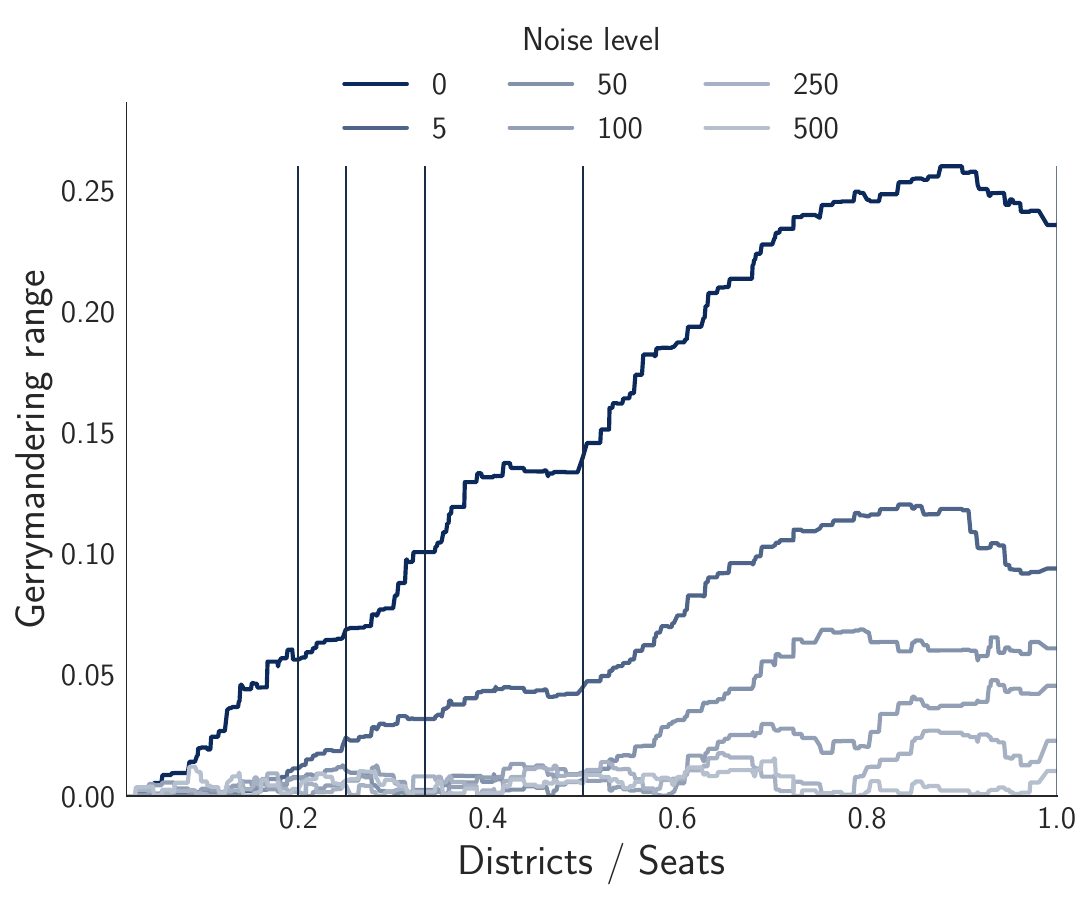} 
		\caption{Replication of \cref{fig:noisegerrymanderinggap} but with idiosyncratic noise drawn from a Normal distribution. Here, the noise level corresponds to the standard deviation $\sigma$. Results are qualitatively identical.}
		\label{fig:noisegerrymanderinggap_normal}
\end{figure}

\begin{figure}[tbhp]
	\centering
		\begin{subfigure}[b]{0.48\textwidth}
		\centering
		\includegraphics[width=.95\textwidth]{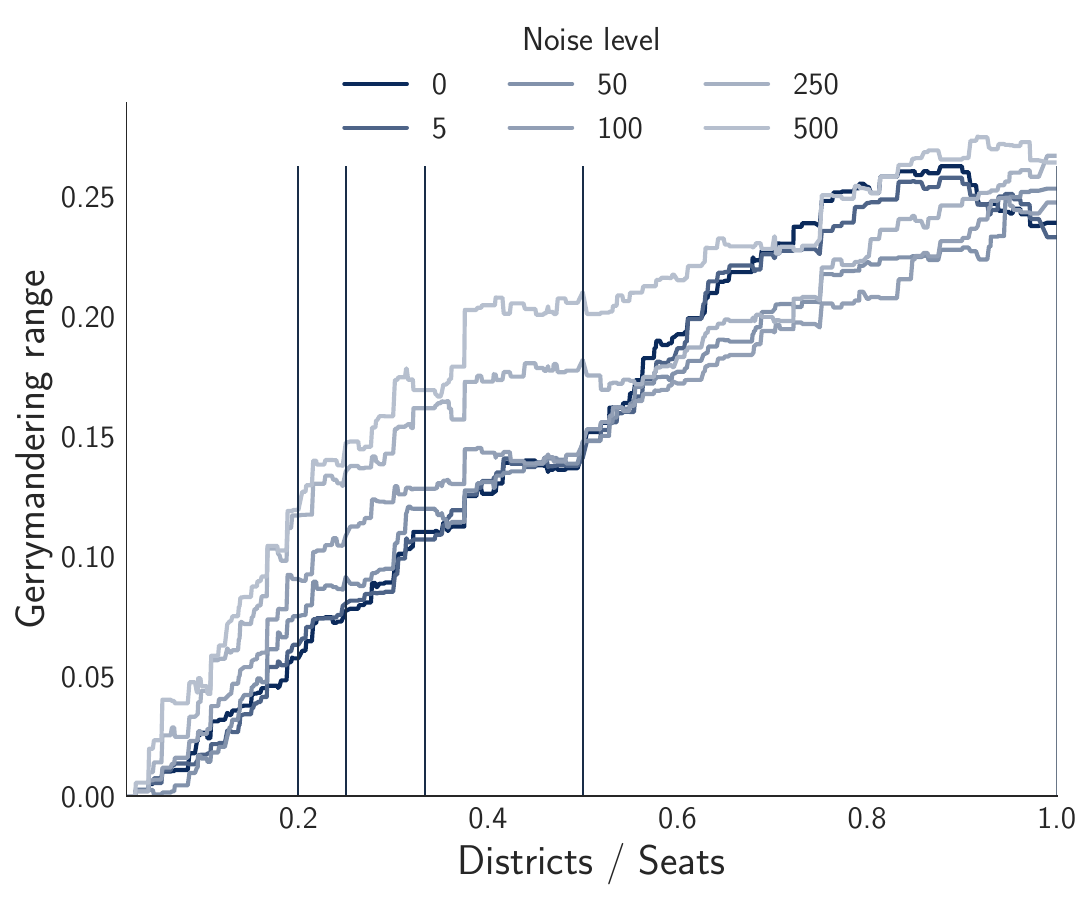} 
		\caption{Gumbel noise}
			\end{subfigure}
~
		\begin{subfigure}[b]{0.48\textwidth}
		\centering
		\includegraphics[width=.95\textwidth]{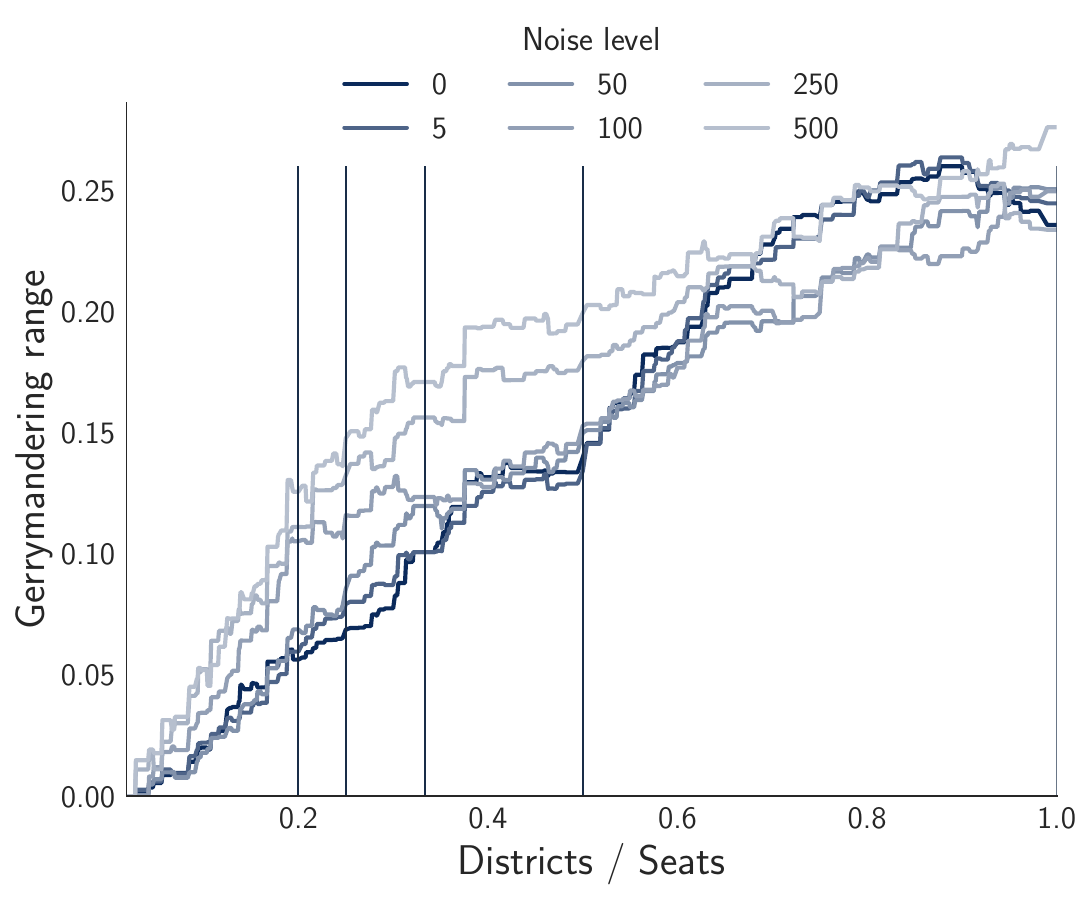} 
		\caption{Normal noise}

			\end{subfigure}
            		\caption{Replication of \cref{fig:noisegerrymanderinggap} but with the optimizer having perfect knowledge of voter-noise realizations (which may also differ by map). Now, gerrymandering ability increases with knowable noise (since different maps may have different voter noise distributions, allowing for more extreme outcomes with noise). However, as before, gerrymandering range decreases with district size, with STV -- even with an extreme case of optimizer knowledge, MMDs decrease their ability to construct disproportionate outcomes.}
		\label{fig:noisegerrymanderinggap_expostoptimization}
\end{figure}

\FloatBarrier
\section{Results: Intra-party diversity}
\label{sec:intra}

So far, we have studied how multi-member districts affect the balance of power between parties. However, one of the primary justifications of Ranked Choice Voting (either for SMDs or MMDs) is that it enables minor parties or ideologies to gain seats -- it blunts the game-theoretic logic that tends to push winner-take-all democracies toward two parties. In this section, we analyze such claims by showing the effects of MMDs on intra-party diversity.

As in \Cref{sec:newwithoutsolidcoalitions}, the \textit{solid coalitions} assumption of \Cref{prop:stvequalspav} does not hold. There, we analyzed crossover votes between parties; now, we analyze crossover votes \textit{within each party} but between coalitions within that party; we must now generate intra-party rankings and simulate STV.\footnote{In this section, we exclusively consider STV, as it allows voters within a party to prioritize different candidates, without risking their party overall representation by not approving a same-party candidate (as could happen with approval voting based methods like Thiele rules). Studying approval methods would require analyzing primaries, to select exactly $N_k$ candidates from each party.} There is a further conceptual challenge: \textit{coalitions} within parties are not well-defined, and identifying them with data is challenging; with many candidates, no two voters may share the same exact ballot preference order. In theory, results are often proven with regard to arbitrary-sized coalitions that approve the same candidates (cf. \citet{skowronProportionalityDegreeMultiwinner2019}), but relating the guarantees back to coalitions in practice is difficult; it is unclear which coalitions proportionality should be defined with respect to, especially with limited ranked choice data, and how exactly to define proportionality for these groups.

We leverage additional political structure to tackle these challenges. (1) First, instead of considering arbitrary coalitions, we examine voter rankings emerging from differing assumptions on how voters value two dimensions: (approximate) political preference, and geographic preference. We choose these dimensions since single-member districts automatically impose one dimension (geographic) as more important, while in theory MMDs allow voters to prefer political connection over geographic proximity. (2) Second, we measure intra-party effects through two measures: the \textit{diversity of the winning candidate set}, and the \textit{diversity of the coalitions supporting each winner}. If the former increases, then that means more distinct within-party preferences are represented by the winning set. If the latter increases, that means each particular winner is accountable to a more diverse intra-party voter base. However, we do not claim that these dimensions are the most important -- our simulation provides a method to understand how various dimensions trade off; insight for any given political setting requires careful consideration of the most important dimensions and their relationships. Rollout of MMDs might further affect partisan behavior and relationships, making it challenging to study such effects precisely before an implementation.

\subsection{Methods and assumptions}
\label{sec:intramethod}

To study the case when the solid coalitions assumption is broken, we first need to construct plausible intra-party voter rankings and candidate distributions;  second, we need to simulate STV elections given a map. Methodological details are deferred to \Cref{sec:intramethoddetails}. These methods are largely similar to those of \Cref{sec:newwithoutsolidcoalitions}.

\parbold{Generating voters and candidates.} For voters, we simply use the voters in each state in the voter file -- subsampling 5{,}000 voters per district in each simulation. For replication purposes while preserving data privacy, in our code repository we provide a subsample of 50{,}000 voters, with noise added to the scores. We further generate one candidate per (party, census tract) combination for each simulation, with a maximum of 1{,}000 candidates per election.

\parbold{Constructing intra-party rankings.} We require further assumptions on how voters differentially rank candidates \textit{within} their preferred party, to study how MMDs may change the characteristics of winners within a party.
The key challenge is to develop a model for how a voter -- given their characteristics -- will vote given a menu of (hypothetical) candidates. Up to now, we have only assumed that voters approve all candidates of their party or rank them all above candidates of the other party.

We do so as follows, using the voter file described in \Cref{sec:empiricalmethod}. Recall that we have individual-level voter data in each tract, along with demographic information and ideological scores; in particular, we use the voter's geographic location (census tract of home address) and a univariate \textit{partisan score} indicating the strength of their party affiliation (most Republican to most Democratic). We generate many synthetic candidates for each party in each district, with varying partisan scores; the distribution of candidates reflects that of the voters.

Finally,  we assume that voters rank all candidates. They rank all same-party candidates over other-party candidates. Then, within each party, they rank candidates in order of the distance between their characteristics (either partisan scores or geographic location, in different simulations).\footnote{This approach is conceptually related to that of \citet{beckerComputationalRedistrictingVoting}, who use Ecological Inference (EI) to study racial groups' vote choices in primaries, to study how to draw SMDs such that a minority group's preferred candidate wins both a primary and the general election -- with the insight that racial composition does not solely determine whether a district is effective for minorities, as within-party vote choice may not only depend on race. Our approach replaces EI with a calibrated, individual-level voter file; i.e., we assume that voters order candidates using their individual-level partisan scores and allow for the possibility that voters with different demographic characteristics may nevertheless vote similarly. For MMD analysis, both methods require extrapolating the scores to rankings, and thus require assumptions on how voters will vote for hypothetical candidates. %
\citet{benadeRankedChoiceVoting2021} use a combination of historical voting behavior and ranked-choice model assumptions to generate such rankings, while we sample candidates and impose a spatial model between candidates and voters.} As a result, the \textit{solid coalitions} assumption no longer holds within parties, where voter preferences may not be expressed neatly in terms of subparties.

We note that while there has been much work on spatial voter models aiming to characterize such behavior based on the ``distance'' between the voter and each candidate~\citep{adamsModerateVotersWeigh2017,tausanovitchDoesIdeologicalProximity2018,shorIdeologyUSCongressional2018,jesseeSpatialVoting20042009}, it is a hard challenge -- it is not clear that voters behave according to such ideological spatial positioning, beyond the relative consistency of which party one votes for. Thus, our results should not be interpreted as what \textit{would} necessarily happen with multi-member districts, but what intra-party coalitions \textit{could or could not} be formed with MMDs given voter interest, and how these coalitions differ from those possible under SMDs.

\parbold{Running STV elections.} Given the candidates, sampled voters, and the voters' rankings over the candidates, we run fractional STV for each district in each given map. Fractional STV is STV as defined above. In each round, either a winner is selected if they have at least the Droop quota number of first-place votes, or a candidate with the least votes is eliminated. This candidate's votes are transferred, by eliminating this candidate from each voter's list. In fractional STV, votes are transferred as follows. Formally, suppose the winning candidate receives $v > Q$ first-place votes. Then, a fraction $\frac{v - Q}{v}$ of this candidate's votes is in excess of what is needed to win. Thus, each voter for this candidate has their weight multiplied by $\frac{v - Q}{v}$. For example, suppose the Droop quota was 5, and a candidate received 6.3 first-place votes (fractional votes are possible due to earlier round transfers). Then, each of their voters has their weight multiplied by $\frac{1.3}{6.3}$.

We simulate STV for the random, neutral maps calculated above, as we wish to study intra-party effects as distinct from partisan ones.
Running these STV elections utilized over 60 CPU-weeks of compute, on top of the map generation discussed in \Cref{sec:empiricalmethod}. This runtime is for \textit{given} maps, underscoring the necessity of \Cref{prop:stvequalspav} to study partisan seat shares without needing to simulate STV during the redistricting optimization process.

\begin{figure}
	\centering
	\begin{subfigure}[b]{0.48\textwidth}
		\centering
		\includegraphics[width=\textwidth]{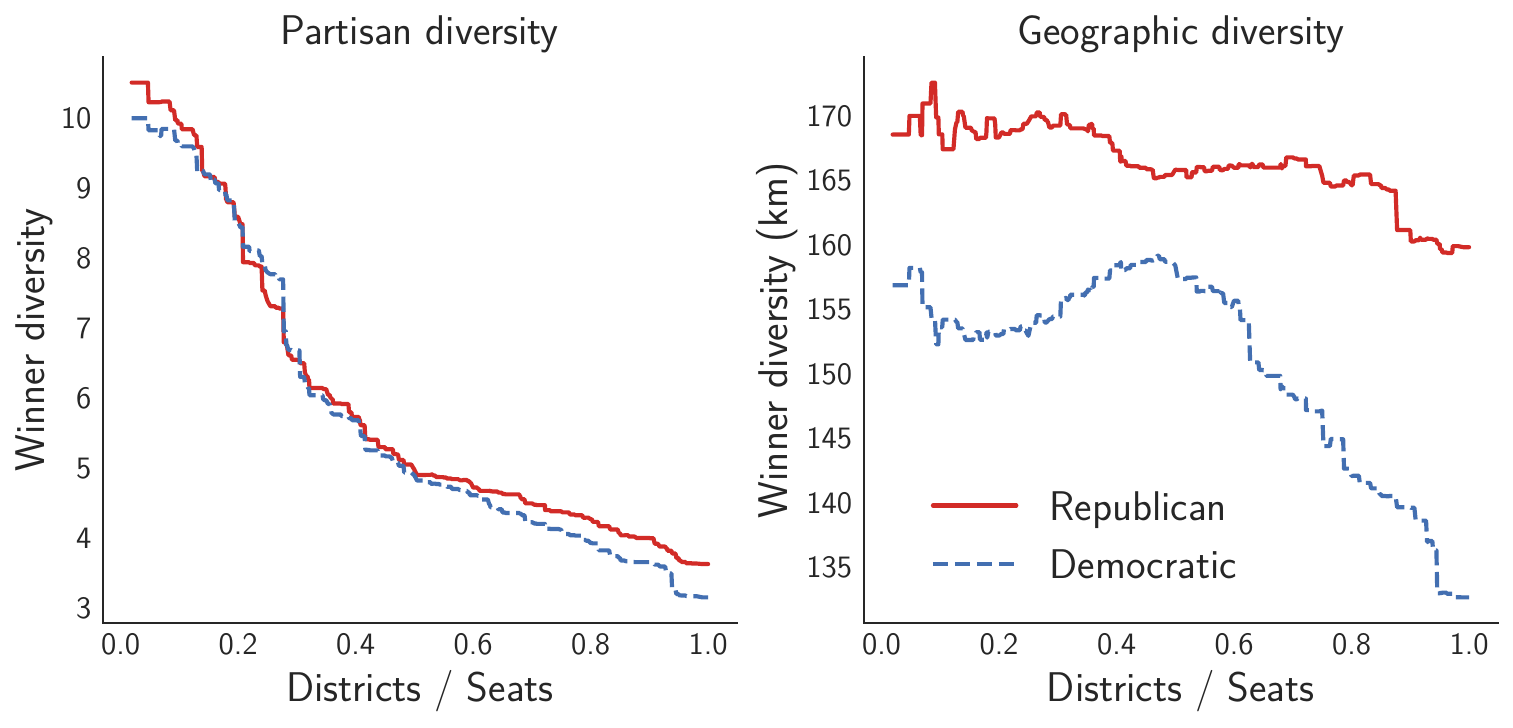}
		\caption{Diversity of winning candidate set}
		\label{fig:winnerdiversity}
	\end{subfigure}
	\hfill
	\begin{subfigure}[b]{0.48\textwidth}
	\centering
	\includegraphics[width=\textwidth]{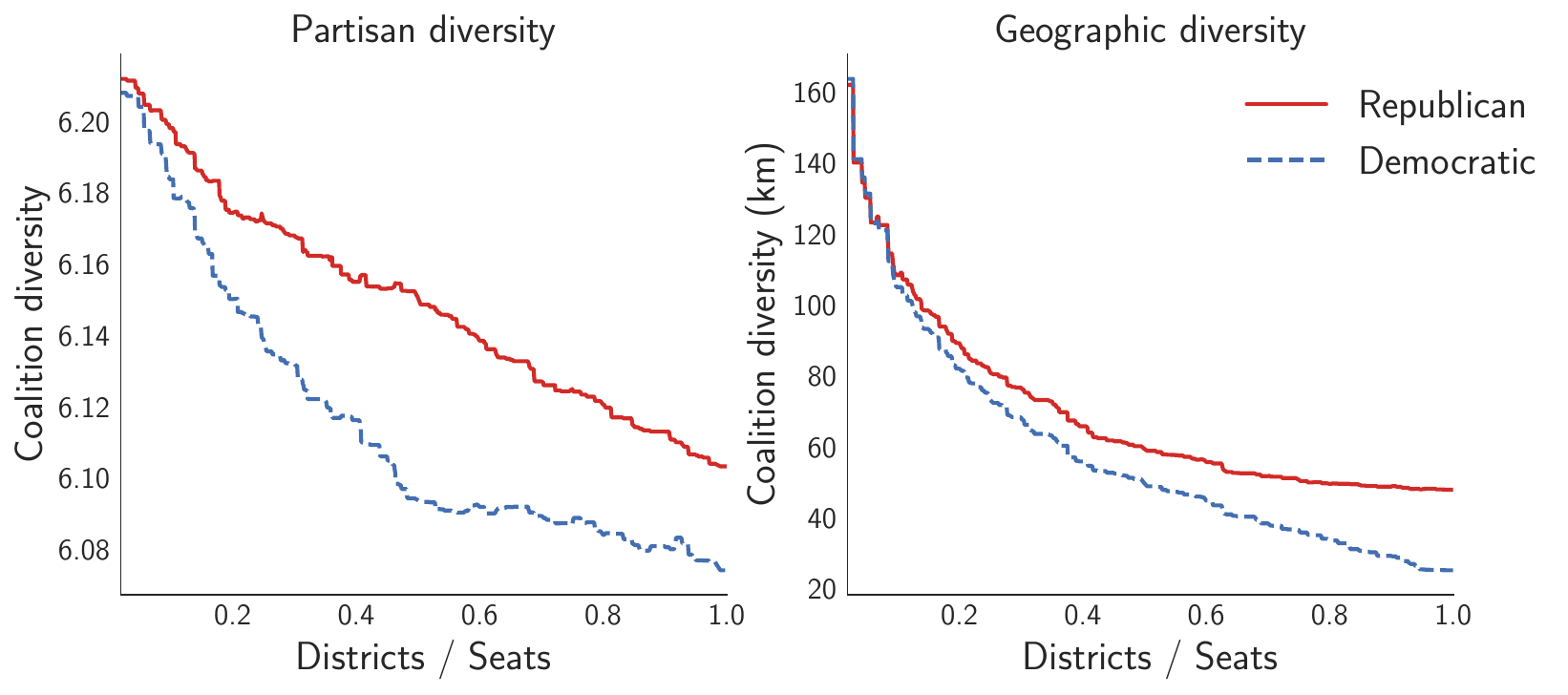}
	\caption{Diversity of coalitions supporting winners}
	\label{fig:coalitiondiversity}
\end{subfigure}
	\caption{How intra-party diversity measures vary with the number of districts. (a) The diversity of the winning candidate set.  (b) The diversity of the \textit{voters} who supported a given winning candidate, averaged by party. These results establish that, with STV, a more diverse set of winners can be elected, i.e., minority viewpoints \textit{within} a party are supported. Simultaneously, each winning candidate draws support from a more diverse coalition of voters. Republicans (red line) tend to be more geographically spread out than Democrats (blue line), and so have a higher level of geographic diversity as measured at every district size. Republicans having higher partisan diversity coalitions may be related to the same phenomenon, since coalitions to gather winning votes may also be more spread out geographically and for partisan scores.}
	\label{fig:intra}
\end{figure}

\begin{figure}
	\centering
	\begin{subfigure}[b]{0.48\textwidth}
		\centering
		\includegraphics[width=\textwidth]{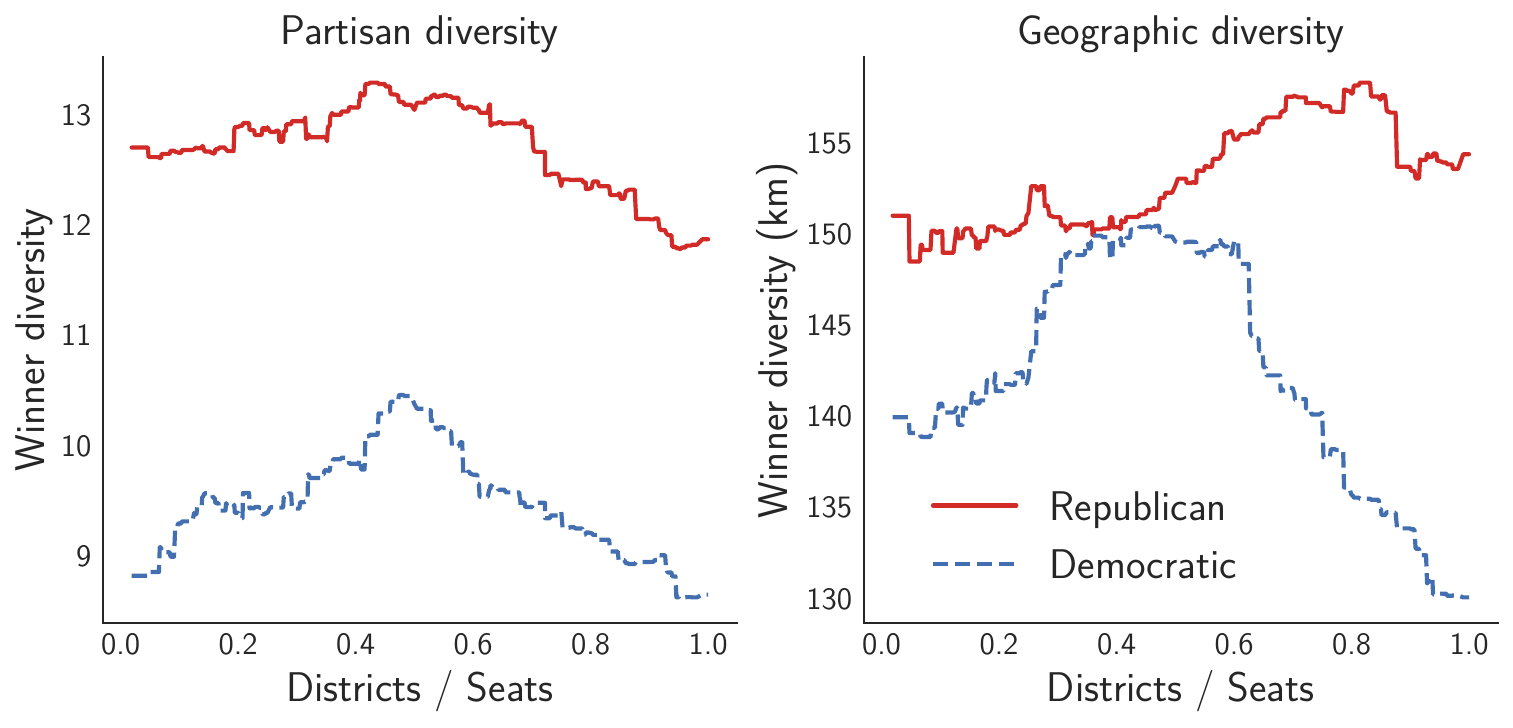}
		\caption{Diversity of winning candidate set}
		\label{fig:winnerdiversitygeog}
	\end{subfigure}
	\hfill
	\begin{subfigure}[b]{0.48\textwidth}
	\centering
	\includegraphics[width=\textwidth]{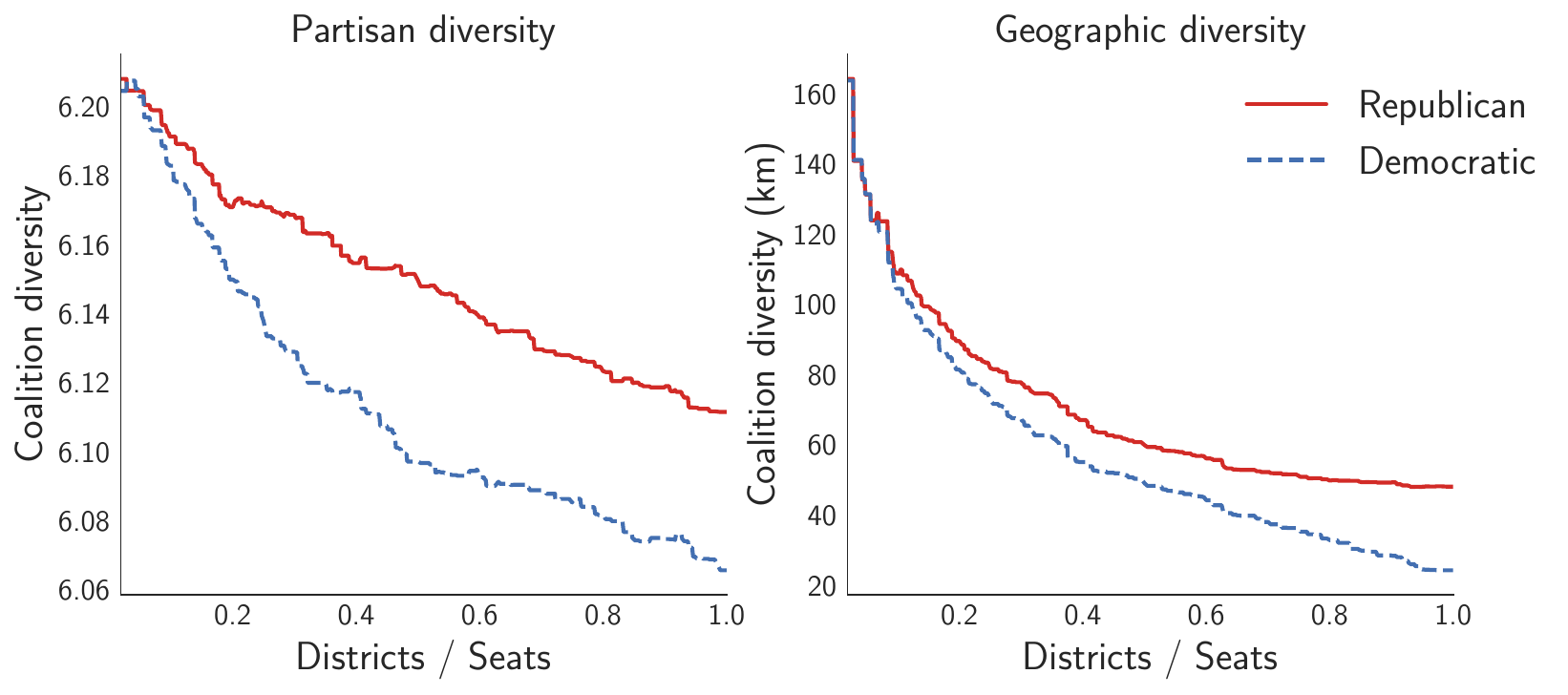}
	\caption{Diversity of coalitions supporting winners}
	\label{fig:coalitiondiversitygeog}
\end{subfigure}
	\caption{Same as \Cref{fig:intra}, except voters now rank within party based on geographic distance.}
	\label{fig:intrageo}
\end{figure}

\subsection{Results}
\label{sec:intraresults}

\Cref{fig:intra} contains our intra-party results, when we assume that within a party voters rank according to partisan score. First, \Cref{fig:winnerdiversity} illustrates how MMDs affect the diversity of the winning set within each party. For each map, we determine the set of winners for each party and then calculate the standard deviation of their partisan scores. For geographic diversity, we calculate the average Euclidean distance of each winner from the centroid of the winners' locations. We find that with STV and large districts, a more diverse set of winners emerges, within each party. The intuition for partisan score is simple and is similar to that regarding the diffusion of minority party voters across a state (the Massachusetts problem). When voters rank according to partisan score, then similar voters across the state can pool their votes to ensure that a favored candidate wins. Using SMDs, these voters may be split across different districts, such that in each a different intra-party coalition chooses the winning candidate. Surprisingly, larger districts simultaneously increase the geographic diversity of winners, even when voters rank according to partisan score. These distance functions are defined in more detail in \Cref{sec:intramethoddetails}.

Second, \Cref{fig:coalitiondiversity} shows that with MMDs each winning candidate draws support from a more diverse coalition of voters. For each winner, we determine the voters who contributed votes in the STV round in which the candidate was elected, weighted by how many votes they provided (since we do fractional STV). For partisan score, we calculate the (weighted) standard deviation of the partisan scores of the voter coalition. For geographic diversity, we calculate the (weighted) average voter Euclidean distance from the district centroid. Finally, we average across winners within each party.

This result establishes that MMDs may come at a cost, in terms of the \textit{geographic representation} aspect of our representative democracy: insofar as it is valuable for a representative to be accountable to a geographically cohesive set of voters (such as for providing constituent services or acquiring funding for projects), large MMDs weaken such ties; voter coalitions move from about 25 kilometers from the coalition center on average to almost 160 kilometers. Two- or three-member districts, however, come at a far smaller cost.

Appendix \Cref{fig:intrageo} contains the same plots when voters' intra-party rank is according to geography. Perhaps as expected, under this assumption MMDs do not increase the diversity of the winners according to partisan scores -- if voters do not band together based on partisan scores, then their mutually preferred candidates on this dimension do not win. However, the findings regarding \textit{coalition} diversity remain virtually identical: even when voters rank within a party based on geographical distance to a candidate, each winner represents a far more geographically dispersed coalition than with SMDs.

While this finding may seem paradoxical, it can occur when voters for one party are not evenly distributed throughout the state -- whereas voters in a (smaller) city may be large enough under SMDs to elect a preferred candidate, in MMDs they may be overruled by a larger group of same-party voters in another city. This finding adds a warning to the notion of proportionality---small, potentially intersectionally defined groups may be better served by small districts, even as theoretical proportionality guarantees a monotonic increase in district size. Most practically for the United States, more research is needed on the effect of MMDs on minority voters (cf. \citet{benadeRankedChoiceVoting2021}).

Overall, our results caution in choosing too large a district size for MMDs, especially as our results in \Cref{sec:proportionality} suggest that most of the proportionality benefits can be achieved with two- or three-member districts---and results in this section suggest that such districts come at a smaller cost in terms of geographic cohesiveness.
More broadly, this section gives an approach to evaluating the effects of voting rules and redistricting on intra-party outcomes, by generating various voter rankings from individual-level data and then studying the effects with respect to interpretable dimensions.

\subsection{Intra-party method details}
\label{sec:intramethoddetails}

Here, we provide additional detail for \Cref{sec:intra}.

\paragraph{Voter multi-dimensional preferences.} In the voter file provided to us, each voter is assigned scores on a variety of ideological and demographic dimensions, including party preference and preferences on specific issues. See the following blog post for additional details on these scores, under ``Survey-based machine learning support scores'': \url{https://predictwise.medium.com/what-all-we-did-for-the-2020-elections-78b66a44f651}. Here we use the \textit{partisan score} value, which is an overall score from 0 to 100 designed to represent how likely a voter is to lean \textit{Democrat} in the two-party vote; this score is derived from party registration, past primary voting data, demographic information, etc. As the post shows, this score correlates well with issue-preference scores. We threshold this value; above the threshold, a voter is declared a \textit{Democrat} and below the threshold, a \textit{Republican}.

As a second dimension, we use \textit{geography} -- based on the voter's home location; for computational simplicity and privacy reasons, we use the centroid of the census tract in which the voter's home location resides, as opposed to the latitude/longitude information of the voter.

\paragraph{Candidate generation.} Next, we generate one candidate for each (party, census tract) combination as follows. (Results do not change with multiple candidates for each (party, census tract) combination). The home address of the candidate is taken to be the centroid of the census tract. The \textit{partisan score} of the candidate is taken to be the median partisan score of voters in the corresponding party in that tract.

\paragraph{Voters and candidates in an election.} Recall -- from the map generation stage -- that a given \textit{map} is a partition of the census tracts into districts. Note that each voter and candidate belongs to a census tract, and thus a unique district in each map.

In a given election (as characterized by the map, social choice function, and number of winners in each district), each voter has available to them a subset of the candidates (those belonging to a census tract in the same district as the voter). The method to generate candidates might produce many candidates or voters in a given district (e.g., if the district makes up the entire state); for computational reasons, we sample at most 5{,}000 voters and 1{,}000 candidates from each district.

\paragraph{Turning voter preferences into rankings.} In each simulated election, we need a ranking over candidates from each voter. These are generated as follows. First, as in the assumption of \Cref{prop:stvequalspav}, we assume that each voter ranks their own party candidates over all other-party candidates. Within each party, voters rank candidates according to a \textit{distance function}. As introduced in \Cref{sec:intra}, for some settings we assume that the distance function is first determined by the \textit{partisan score} -- voters rank candidates according to the absolute value of the difference of their partisan scores, with ties if any broken by geography. In other settings, we assume that the distance function is determined first by \textit{geography}, with voters preferring closer candidates according to the Euclidean distance between the latitude/longitude associated with each voter and candidate (and ties broken by partisan score). Unlike reality, we assume that there is no \textit{ballot exhaustion} -- these simulated voters submit complete rankings.

\paragraph{Simulating STV elections.} From the above steps, for each election in each district, we have a set of up to 5{,}000 voters and 1{,}000 candidates who have submitted rankings over the candidates. We run fractional STV for each such district.

This entire process took over 60 CPU-weeks to compute for all settings, with the largest computational load being generating the individual voter rankings and then running fractional STV.

\clearpage
\section{Proofs}

\begin{lemma}
\label{lem:pavstv}
Let $y_R \in (0, 1]$. Consider
\begin{align*}
n_R(y_R, \lambda_{PAV})
&= \min_n \left[ \argmax_n \left[
y_R \sum_{i=1}^n \lambda(i) \right.\right. \\
&\qquad\left.\left.
+ (1-y_R)\sum_{i = 1}^{M - n} \lambda(i)
\right]\right].
\end{align*}

and let $\ell$ be the unique integer such that  \[
y_R(M+1) - 1 \leq \ell < y_R(M+1)
\]

Then, $n_R(y_R, \lambda_\text{PAV}) = \ell$, for all $y_R, M$.
\end{lemma}
\begin{proof}

We have that $n_R(y_R, \lambda_{PAV})$ is $n$ such that:
\begin{align*}
    n \geq 1 \implies y_R \lambda(n) > (1 - y_R)\lambda(M - n + 1)\\
    n < M \implies y_R\lambda(n+1) \leq  (1 - y_R)\lambda(M - n)
\end{align*}
The first condition requires that choosing the $n$th candidate from party \textit{R} is strictly (due to tie-breaking against \textit{R}) more valuable than choosing the ($M-n + 1$)th candidate from party \textit{D}.

The second condition requires that choosing the ($n+1$)th candidate from party \textit{R} is no more valuable than choosing the ($M-n$)th candidate from party \textit{D}.

Rewriting the first condition, we have
\begin{align*}
    n \geq 1 \implies \frac{y_R}{n} &> \frac{1 - y_R}{M - n + 1}\\
    \iff \frac{y_R}{n} + \frac{y_R}{M - n + 1} &> \frac{1 }{M - n + 1}\\
    \iff \frac{y_R(M+1)}{n(M - n + 1)}  &> \frac{1 }{M - n + 1}\\
    \iff \frac{y_R}{n}  &> \frac{1 }{(M+1)}
\end{align*}
Similarly, rewriting the second condition yields
\begin{align*}
    n < M \implies \frac{y_R}{n+1}  &\leq \frac{1 }{(M+1)}
\end{align*}

\end{proof}

\propSTVclosed*
\begin{proof}
The first part of the result is known from \citet{tideman1995single} and \citet{dummett1984voting}. For completeness, we provide a proof using our notation and in our exact setting.

For expositional simplicity, we assume that $V \text{ mod } (M+1) = 0$, i.e., the number of votes is evenly divisible by one plus the number of winners, though the same arguments extend. Then, let $Q = \floor*{\frac{V}{M + 1}}+1 = \frac{V}{M + 1}+1$ be the Droop quota.

To understand the intuition, note that the given seat shares would immediately hold if each party had a coordinator who could optimally decide how voters of that party rank candidates within the party. In that case, the coordinator would ensure that as many candidates as possible have first-place votes equal to the Droop quota, with no first-place votes going to candidates who will be eliminated. The resulting seat shares follow, given tie-breaking against party \textit{R} (and, if needed, the surplus vote transfer procedure for the additional single vote needed in the Droop quota).

However, this argument is not sufficient because, in principle, candidates could receive meaningful votes from voters of the other party. Further, without such a coordinator, a sub-optimal arrangement of votes could potentially lead to the elimination of candidates who would be elected with such a coordinator. (The former reason is fundamental, and is why the exact formula does not hold in general without parties or even for more than two parties without further assumptions; the latter possibility is exactly the issue that STV surplus vote transfers are designed to avoid, and eliminating it is bookkeeping). The proof centers around eliminating these two possibilities.

The key step in the proof is noting that, under the assumptions, votes can only be transferred from a candidate of one party to a candidate of the other party (either after a candidate is eliminated or selected as a winner) if candidates from the sending party have been exhausted, as the sending party voters rank all other party members after all their own party members. %

Thus, for any set of voter rankings under the assumptions, the per-party seat share remains the same under re-arrangements of how each voter ranks members of the \textit{other} party (at the point that such rearrangements matter for who is elected, only one party's candidates are left, and so the partisan seat share does not change). 
Thus, without loss of generality, for the rest of the proof we assume that all voters within a party share the same ordering for candidates of the \textit{other} party.  

  Recall that the Droop quota is designed such that exactly $M$ total candidates meet the Droop quota across rounds, if all voters submit full rankings and $V \geq M(M+1)$. (For simplicity and without loss of generality, we assume that the election continues even when the number of candidates remaining equals the number of seats to be filled, so that they each meet the quota after transfers).

  More than $M$ candidates reaching the quota would require at least $(M+1)Q = V + M+1 > V$ total votes (as votes necessary to reach the quota are never transferred).
  
  Similarly, we have that at least $M$ candidates reach the Droop quota after transfers. There are enough votes, optimally spread, for enough candidates to reach the quota:  
  \begin{align*}
      MQ \leq V 
      &\iff M \left[\frac{V}{M + 1} \right] + M \leq V \\
      &\iff \left[\frac{M}{M + 1} - \frac{M+1}{M + 1} \right] V \leq -M \\
      &\iff -\left[\frac{1}{M + 1}\right] V \leq -M \\
      &\iff V \geq M(M+1)
  \end{align*}
  Note that the above argument holds at every round. Once we've elected $W$ candidates, we need to elect $M - W$ more, and there are $V - WQ$ first-place votes left:
  \begin{align*}
    (M - W) Q &= MQ - WQ \leq V - WQ
  \end{align*}
  and so by the pigeonhole principle, when there are $M-W$ candidates left (possibly after eliminating some), there is at least one candidate with at least as many votes as the Droop quota.

The above facts establish that we can carry out the initial rounds of STV separately for each party, until in each party there are either no candidates remaining or there is one remaining candidate with less than $Q$ votes: no winners have been selected with votes transferred across parties up to this stage; as long as we only elect candidates with votes at or above the Droop quota, we do not mistakenly elect any candidate separately that we would not have together; and if this stage is ever reached, the \textit{identities} of the elected or eliminated candidates do not matter, since up to now there is a conservation of votes by party and so the per-party counts remain the same.

Now, consider each party $p$, and suppose in the current round that $W_p$ candidates have been elected for the party. Recall that $y_pV$ is the number of voters for party $p$. 

Further suppose that there are $m_p \geq 1$ candidates left for the party, and they collectively have $q_p = y_pV - W_pQ$ first-place votes (either original first-place votes, or votes after transfers) among them. We repeat the above argument. If $q_p\geq Qm_p$, then by the pigeonhole principle, there exists at least one candidate of the party such that the votes for that candidate meet the Droop quota $Q$. Then, that candidate is declared a winner, and its surplus votes are transferred to other candidates of the same party. Then, iterate with $W_p + 1$, $m_p - 1$, and $q_p - Q$. Otherwise, eliminate the candidate in the party with the least number of votes, and iterate with $m_p -1$ and $q_p$ until $m_p=0$ or $m_p=1$ with the remaining candidate having less than $Q$ votes.

Suppose at the end of these separate processes, $W_p$ candidates from each party $p$ have been elected. We know that there is at most one candidate from each party $p$ remaining, with $y_pV - W_pQ$ votes. 

Then
\begin{align*}
    n_p(y_p, STV) \geq W_p &\geq \floor{\frac{y_pV}{Q}} \\
    &= \floor{\frac{y_pV}{\frac{V}{M + 1}+1}} \\
    &= \floor{\frac{y_p(M+1)}{1+\frac{M + 1}{V}}} \\
    & \geq \floor{\frac{y_p(M+1)}{1+\frac{M + 1}{M(M+1)}}} & V \geq M(M+1) \\
    &= \floor{y_p M} \\
\end{align*}
Since this holds for both parties simultaneously, we have
\begin{align*}
    \floor{y_R M} &\leq n_R(y_R, STV) \\
    &\leq M - \floor{M - My_R} = \ceil{y_RM}.
\end{align*}
Whether $n_R(y_R, STV)$ is at the ceiling or floor depends on whether it has the majority of the remaining votes when each party has at most one candidate remaining.

The proof is finished by applying \Cref{lem:pavstv}. %

\end{proof}

\section{Empirical methodological detail}
\subsection{Map Generation}
All details of the original Stochastic Hierarchical Partitioning (SHP) algorithm can be found in the paper by \citet{gurnee2021fairmandering} and associated appendices; we further use their same geographic and electoral data (see Appendix Table 2 of \citet{gurnee2021fairmandering}). Here, we discuss the relevant algorithmic parameters and differences from that algorithm.

\parbold{Multi-member.}
The main algorithmic difference from the original SHP algorithm is that we adapted ours to generate multi-member districts. To do this, instead of parameterizing a sample tree node by a region $R$ and total number of seats $s$, we needed to also specify the number of districts that node contains. This is required because at an intermediate node, the number of districts is not immediately derivable from the total number of seats in that node because of ambiguity in the number of $N_k$ versus $N_k+1$ sized districts. Therefore, the number of seats is used just to balance population, and the number of districts is used for all other tree operations (sampling valid splits, maintaining balance, etc.).

\parbold{Ensembles.}
For each pair (state, $K$), with $K \in \{2, 3, \dots, 10, 12, \dots, 20, 23, \dots 53\} \cup \{N\}$ and $K \leq N$, we sampled the root node $\left(\frac{1000}{K}\right)^{1.2}$ times and each internal node $\left(\frac{300}{K}\right)^{0.5}$ times. These constants were chosen to balance computational cost and optimization quality. We used random-iterative center selection with Voronoi-weighted capacity matching to sample region centers and sizes. All districts are population balanced within a $1\%$ tolerance. Each of these ensembles was then scored, optimized, and subsampled to create a final distribution over partisan outcomes.

\end{document}